\documentclass[letterpaper,11pt]{article}

\usepackage{amsmath, amssymb, amsthm}

\usepackage{authblk}

\usepackage{graphicx}
\graphicspath{{figures/}}
\usepackage{pgfplots}

\usepackage{caption}
\usepackage{enumitem}
\usepackage{mdframed}

\usepackage{soul}

\usepackage{todonotes}

\usepackage{thm-restate}

\usepackage{algorithm}
\usepackage[noend]{algpseudocode}

\usepackage[hypertexnames=false,hidelinks]{hyperref}
\hypersetup{
	colorlinks   = true, %
	urlcolor     = blue, %
	linkcolor    = blue, %
	citecolor   = red %
}

\usepackage[capitalize]{cleveref}

\usepackage{dsfont}

\usepackage{xspace}

\usepackage{mathtools}
\usepackage{bbm}
\usepackage{amssymb}

\newenvironment{claimproof}{\begin{proof}}{\end{proof}}

\crefname{claim}{Claim}{Claims}
\Crefname{claim}{Claim}{Claims}

\crefname{algorithm}{Algorithm}{Algorithms}
\Crefname{algorithm}{Algorithm}{Algorithms}
\crefname{ALC@unique}{Line}{Lines}
\Crefname{ALC@unique}{Line}{Lines}

\newlist{properties}{enumerate}{1}
\setlist[properties,1]{label={\textbf{P\arabic*}}, align=left, left=0pt, itemindent=*}

\crefname{propertiesi}{property}{properties}

\newcommand{\cF}{\mathcal{F}}

\newcommand{\floor}[1]{ {\left\lfloor #1 \right\rfloor}}
\newcommand{\ceil}[1]{ {\left\lceil #1 \right\rceil}}
\newcommand{\abs}[1]{{\left| #1\right|}}

\DeclareMathOperator{\sign}{sign}

\newcommand{\eps}{\varepsilon}

\newcommand{\A}{{\mathcal{A}}}

\newcommand{\prbinom}{\textnormal{\texttt{binom}}}

\newcommand{\runtime}[1][\alpha, \beta, c, q]{\textnormal{\texttt{convert}}\left(#1\right)}

\newcommand{\core}[1]{\textnormal{\texttt{Core}}(#1)}

\newcommand{\sat}{\textnormal{\texttt{SOL}}}

\newcommand{\randext}{\textnormal{\texttt{RandAndExtend}}}
\newcommand{\iterrandext}{\textnormal{\texttt{SamplingWithABlackBox}}}

\newcommand{\D}[2]{\mathcal{D}\left(#1\, \middle\|\,#2 \right)}

\newcommand{\OPT}{\textnormal{\texttt{OPT}}}

\newcommand{\Oh}{\mathcal{O}}

\newcommand{\goodd}{\textnormal{\texttt{interval}}(\alpha, \beta)}

\newcommand{\params}{\textnormal{\texttt{parameters}}}

\newcommand{\procext}[4]{\left( #1, #2, #3, #4 \right)\textnormal{-\texttt{procedure}}}
\newcommand{\proc}[2]{\left( #1, #2 \right)\textnormal{-\texttt{procedure}}}

\newcommand{\procalg}[2]{\mathcal{P}_{#1,#2}}

\algnewcommand\algorithmicconf{\textbf{Configuration:}}
\algnewcommand\Configuration{\item[\algorithmicconf]}

\algnewcommand\algorithmicinput{\textbf{Input}:}
\algnewcommand\Input{\item[\algorithmicinput]}

\algnewcommand\algorithmicoutput{\textbf{Output}}
\algnewcommand\Output{\item[\algorithmic\poutput]}

\algnewcommand\algorithmicmyreturn{\textbf{Return}}
\algnewcommand\RETURN{\item[\algorithmicmyreturn]}

\newtheorem{theorem}{Theorem}[section]
\newtheorem{lemma}[theorem]{Lemma}
\newtheorem{definition}[theorem]{Definition}
\newtheorem{claim}[theorem]{Claim}

\newtheorem{corollary}[theorem]{Corollary}
\newtheorem{observation}[theorem]{Observation}

\newcommand{\nmgpivd}[1][\mathcal{G},\Pi]{$(#1)$\textnormal{\texttt{-Vertex Deletion}}}
\newcommand{\gpivd}[1][\mathcal{G},\Pi]{(#1)\textnormal{\texttt{-Vertex Deletion}}}
\newcommand{\sgpivd}[1][\mathcal{G},\Pi]{(#1)\textnormal{\texttt{-Del}}}

\newcommand{\lpvc}[1][\ell]{\ensuremath{#1}\textnormal{-path Vertex Cover}}

\newcommand{\vcmaxdegthree}{\textsc{Vertex Cover on graphs with maximal degree $3$}\xspace}

\newcommand{\one}{\mathbbm{1}}
\newcommand{\E}{{\mathbb{E}}}

\newcommand{\cG}{\mathcal{G}}

\newcommand{\FVS}{\textnormal{\texttt{FVS}}\xspace}
\newcommand{\POVD}{\textnormal{\texttt{POVD}}\xspace}

\newcommand{\fvsPi}{\Pi^{\texttt{FVS}}}
\newcommand{\povdPi}{\Pi^{\texttt{POVD}}}

\newcommand{\wvc}{\textsc{Weighted Vertex Cover}\xspace}
\newcommand{\vc}{\textsc{Vertex Cover}\xspace}
\newcommand{\pathvc}[1]{\ensuremath{#1}\textsc{-Path Vertex Cover}\xspace}
\newcommand{\hs}[1]{\ensuremath{#1}\textsc{-Hitting Set}\xspace}
\newcommand{\fvs}{\textsc{Feedback Vertex Set}\xspace}
\newcommand{\dfvst}{\textsc{Directed Feedback Vertex Set on Tournaments}\xspace}
\newcommand{\povd}{\textsc{Pathwidth One Vertex Deletion}\xspace}
\newcommand{\spovd}{\textsc{POVD}\xspace}

\newcommand{\sdeltal}{\delta^{*}_{\textnormal{left}}}
\newcommand{\sdeltar}{\delta^{*}_{\textnormal{right}}}

\newcommand{\defproblem}[4]{
  \vspace{1mm}
  \noindent\fbox{
  \begin{minipage}{0.96\textwidth}
  \begin{tabular*}{\textwidth}{@{\extracolsep{\fill}}lr} #1 \\ \end{tabular*}
  {\bf{Input:}} #2  \\
  {\bf{Question:}} #3\\
  {\bf{$\sgpivd$ Equivalence:}} #4
  \end{minipage}
  }
  \vspace{1mm}
}

\usepackage[margin=1in]{geometry}

\title{Sampling with a Black Box: Faster Parameterized Approximation Algorithms for Vertex Deletion Problems}
\author[1]{Bar\i\c{s} Can Esmer\thanks{The author is part of Saarbrücken Graduate School of Computer Science, Germany.}}
\author[2]{Ariel Kulik\thanks{This project has received funding from the European Union’s Horizon 2020 research and innovation programme under grant agreement No 852780-ERC (SUBMODULAR).}}
\affil[1]{CISPA Helmholtz Center for Information Security, Saarbr\"ucken, Germany. \texttt{baris-can.esmer@cispa.de}}
\affil[2]{Computer Science Department, Technion, Haifa, Israel. \texttt{kulik@cs.technion.ac.il}}

\date{}

\begin{document}

\maketitle

\begin{abstract}
In this paper we introduce {\em Sampling with a Black Box}, a generic technique for the design of parameterized approximation algorithms for vertex deletion problems (e.g., \vc, \fvs, etc.). The technique relies on two components:
\begin{itemize}
	\item A {\em Sampling Step}. A polynomial time randomized algorithm which given a graph $G$ returns a random vertex $v$ such that the optimum of $G\setminus \{v\}$ is smaller by $1$ than the optimum of $G$ with some prescribed probability $q$. We show such algorithms exists for multiple vertex deletion problems.
	 \item A Black Box algorithm which is either an exact parameterized  algorithm or a polynomial time approximation algorithm. 
\end{itemize}
Our technique combines these two components together. The sampling step is applied iteratively to remove vertices from the input graph, and then the solution is extended using the black box algorithm.  The process is repeated sufficiently many times so that the target approximation ratio is attained with a constant probability. The main novelty of our work lies in the analysis of the framework and the optimization of the parameters it uses.    
 
We use the technique to derive parameterized approximation algorithm for several vertex deletion problems, including \fvs, \hs{d} and \pathvc{\ell}. In particular, for every approximation ratio $1<\beta<2$, we attain a parameterized $\beta$-approximation for \fvs which is faster than the parameterized $\beta$-approximation of [Jana,   Lokshtanov,  Mandal, Rai and Saurabh,  MFCS 23']. Furthermore, our algorithms are always faster than the algorithms attained using Fidelity Preserving Transformations [Fellows, Kulik, Rosamond, and Shachnai, JCSS 18'].

\end{abstract}
\thispagestyle{empty}

\section{Introduction}
\label{sec:introduction}

A vast body of research has been dedicated to  basic vertex deletion problems such as \vc, \hs{3} and \fvs. In these problems, the objective is to delete a minimum cardinality set of vertices from the input (hyper-)graph so that the remaining (hyper-)graph  satisfies a specific property (edge-free, cycle-free, etc.). As many of these problems are NP-hard,  multiple algorithmic results  focus on either polynomial time approximations or exact parameterized algorithms.  In between these two classes of algorithms lies the class of {\em parameterized approximation algorithms}. These algorithms aim to provide approximation ratios which cannot be attained in polynomial time. They operate within a parameterized running time, which is faster than the exact, parameterized state-of-the-art.

In this paper we explore how existing exact parameterized algorithms and polynomial time approximation algorithms can be used together with {\em sampling steps} to derive efficient parameterized approximation algorithms. Informally, a sampling  step with success probability $q\in (0,1)$ is a polynomial time algorithm which, given an input graph $G=(V,E)$, returns a random vertex $v\in V$. The vertex $v$ should satisfy, with probability $q$ or more, that removing $v$ from $G$ reduces its optimum (i.e., the number of vertices one needs to remove from the graph for it to satisfy the property) by $1$. As we show in this paper, such algorithms can be easily obtained for various vertex deletion problems.

Our technique, {\em Sampling with a Black Box}, applies the sampling step $t$ times and subsequently uses the existing parameterized/approximation algorithms to complete the solution. The whole process is executed sufficiently many times so that a $\beta$-approximate solution is found with a constant probability.  The main novelty of our work lies in the analysis of the framework, involving tail bounds for binomial distribution and optimization of the parameters used by the technique.

Sampling with a Black-Box is applicable to a wide collection of vertex deletion problems. We show sampling steps exist for every vertex deletion problems, for which the property can be described by a finite set of forbidden vertex induced  hypergraphs, such as $\hs{d}$, $\pathvc{\ell}$ and \dfvst. We further provide sampling steps for \fvs (\FVS) and  \povd, where the set of forbidden hypergraphs is infinite.
Moreover, even though our setting doesn't explicitly allow it, the results developed in this paper also apply to some problems in the directed graph setting. In particular, we show that there exists a sampling step for \dfvst. Thus, the technique is applicable for each of these problems.  

We compare Sampling with a Black Box to existing benchmarks. 
\begin{itemize}
	\item
	In \cite{janaParameterizedApproximationScheme2023}, Jana, Lokshtanov, Mandal,  Rai, and Saurabh developed a parameterized $\beta$-approximation for \fvs (\FVS) for every $1<\beta<2$.  The objective in \FVS\, is to remove a minimum number of vertices from an undirected graph, so the remaining graph does not contain cycles. 
	Their approach relies on the same ingredients as ours: utilize the state of art parameterized and approximation algorithms in conjunction with a variant of the well known  randomized branching rule of Becker et al. \cite{beckerRandomizedAlgorithmsLoop2000}.
	
	Similar to \cite{janaParameterizedApproximationScheme2023}, we utilize the randomized branching rule of  \cite{beckerRandomizedAlgorithmsLoop2000} to derive a sampling step with success probability $\frac{1}{4}$. We use Sampling with a Blackbox together with this sampling step to attain a parameterized $\beta$-approximation for \fvs, for every $1<\beta<2$.
	 Though we use the same core principles, we attain a faster running time for {\em every} approximation ratio between $1$ and $2$ -- see comparison in \Cref{fig:fvs}. 
	
	The improved running time stems from our tighter analysis and careful selection of parameters, which also provides flexibility in the design of the sampling steps. For example, to attain  a parameterized $1.1$-approximation for \FVS, the authors of \cite{janaParameterizedApproximationScheme2023} use ideas from \cite{beckerRandomizedAlgorithmsLoop2000} to derive a simple algorithm which randomly picks an {\em edge} in the graph $G$ such that at least one of its endpoints is in a minimum solution with probability $\frac{1}{2}$. 
	If the selected edge satisfies this property we refer to it as {\em correct}. 
	In their scheme, every time an edge is picked, both its endpoints are added to the solution, reducing the optimum by one while increasing the solution size by two, assuming the picked edge is correct.
	
	The analysis in \cite{janaParameterizedApproximationScheme2023} only considers the case in which  the picked edges  are {\em always} correct. This allows for a simple analysis and renders the parameter selection straightforward. Keeping our focus on a $1.1$-approximation,  the algorithm has to pick $0.1\cdot k$ edges before it invokes the state-of-art parameterized algorithm with the parameter $k'=0.9\cdot k$. The success probability of the above -- the probability that all picked edges are correct and therefore a $1.1$-approximate solution is attained - is $0.5^{0.1\cdot k}$, and the running time is  $c^{0.9\cdot k}\cdot n^{\Oh(n)}$,  where $c=2.7$ is the running time of the best known exact parameterized  algorithm for \FVS~\cite {liDetectingFeedbackVertex2022}. This leads to a running time of $2^{0.1\cdot k}\cdot c^{0.9\cdot k}\approx 2.62^k$.
	
	A major limiting factor in the analysis of \cite{janaParameterizedApproximationScheme2023} is the focus on the event in which {\em all} picked edges are correct. 
	For example, to attain  a parameterized $1.1$-approximation one can consider picking $0.09k$ edges (and add both endpoints to the solution), and then invoke the exact parameterized algorithm with $k'=0.92k$. 
	Now, this procedure finds a $1.1$-approximate solution if $0.08k$ (or more) of the picked edges are correct. 
	A careful analysis shows that the probability of such event is $\approx 0.708^{0.09\cdot k}\approx 0.969^k$. 
	By repeating this procedure $\frac{1}{0.969^k}$ times, a $1.1$-approximate solution is found with a constant probability.  The overall running time is $\left(\frac{c}{0.969}\right)^k\cdot n^{\Oh(1)} \approx 2.57^k\cdot n^{\Oh(1)}$, which is already an improvement over \cite{janaParameterizedApproximationScheme2023}.
	
	The above example illustrates the power of a more flexible analysis and a careful selection of parameters such as the number of sampled edges. Furthermore, our approach allows for more powerful sampling steps- instead  of picking both endpoints of the selected edges, we can select one at random. 
	Intuitively, the benefit of sampling one vertex at a time is that it increases the {\em variance} of the number vertices selected from the optimum, therefore raising the probability of the rare event the analysis is focused on. 
	This change leads to a further improvement to the running time.  Overall, we attained a parameterized $1.1$-approximation for \FVS\ in time $2.483^k\cdot n^{\Oh(1)}$.

	\begin{figure}[t!]
		\centering
		\input{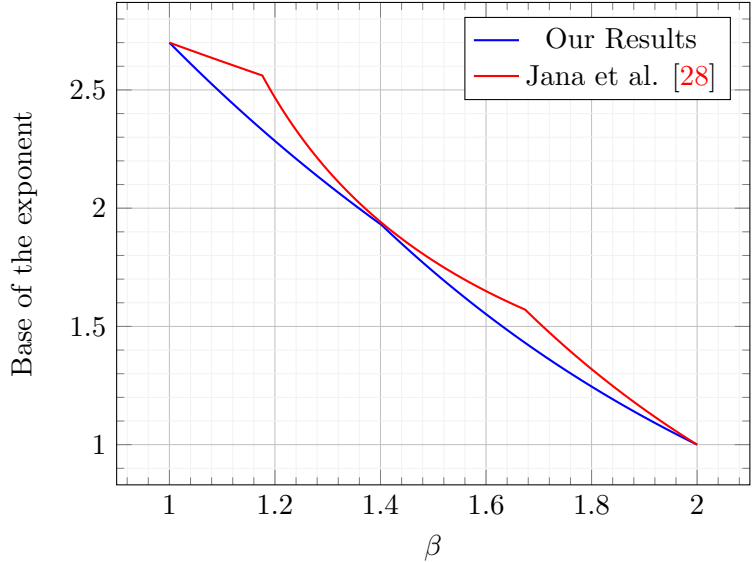}
		\caption{Comparison of the running times for $\fvs$. The $x$-axis corresponds to the approximation ratio, while the $y$-axis corresponds to the base of the exponent in the running time. A dot at $(\beta,c)$ means that there is a parameterized $\beta$-approximation for \FVS\ in time $c^k\cdot n^{\Oh(1)}$.} 
		\label{fig:fvs}
	\end{figure}
	
	\item In \cite{Fellows2018}, Fellows, Kulik, Rosamond and Shachnai provided a generic technique, called Fidelity Preserving Transformations, which can be applied for every vertex deletion problem in which the property can be described by a finite set of forbidden  vertex induced subgraphs. Assuming the maximum number of vertices in a forbidden induced subgraph is $\eta$, and the problem has  an exact  parameterized $c^k\cdot n^{\Oh(1)}$ algorithm, the technique  of \cite{Fellows2018} yields a $\beta$-approximation in time $c^{\frac{\eta-\beta}{\eta-1}\cdot k }\cdot n^{\Oh(1)}$.\footnote{The result in \cite{Fellows2018} was stated for \vc and \hs{3}, but can be easily generalized to the stated result.} 
	We prove that the running time we obtain for every problem in the class is always faster than the running time attained by \cite{Fellows2018}, for every $1<\beta<\eta$. 
	 \Cref{fig:3pvc} provides a comparison between the running times we attain and those of \cite{Fellows2018} for \pathvc{3}, in which the objective is to remove vertices from a graph so the remaining graph does not have a path of $3$ vertices. 
	
	\begin{figure}[t]
		\centering
		\input{figures/3pathvc.tex}		
		\caption{Comparison of the running times for $\pathvc{3}$. The $x$-axis corresponds to the approximation ratio, while the $y$-axis corresponds to the base of the exponent in the running time. A dot at $(\beta,c)$ means that there is a parameterized $\beta$-approximation for \FVS\ in time $c^k\cdot n^{\Oh(1)}$.}
		\label{fig:3pvc}
	\end{figure}
\end{itemize}

\subsection{Related Work}

The field of parameterized approximation aims to derive approximation algorithms which run in parameterized running time. In the classic setting, the considered problem is one which is not expected to have an exact parameterized algorithm, and further has a hardness of approximation lower bound which indicates that a polynomial time $\alpha$-approximation is unlikely to exist, for some $\alpha>1$. In such a setting, the objective is to derive a $\beta$-approximation algorithm that runs in parameterized running time for $\beta<\alpha$, or to show that such an algorithm cannot exist (subject to common complexity assumptions). In recent years, there has been a surge of breakthrough results in the field, providing both algorithms (e.g., \cite{La14,GKW19,DFK+18,Manu19,LSS20, BLM20, ABB+23,DKS23,ILSS23}) and hardness of approximation results  (e.g., \cite{GLRSW24,LRSW23, Wlod20,CCKLMNT20,BBE+21,Manu20,KK22,LRSW23b,CFLL23}).

Within the broad field of parameterized approximations, a subset of  works consider problems which are in FPT. For problems in FPT, there is always a parameterized $\beta$-approximation algorithm for all $\beta\geq 1$, as the exact algorithm is also an  approximation algorithm. Therefore, the main focus of these works is on the {\em trade-off} between approximation and running time. 

A prominent approach to attain such parameterized-approximation algorithms relies on using existing exact parameterized algorithms as black-box.
In \cite{BEP11} the authors showed  that an exact parameterized $c^k\cdot n^{\Oh(1)}$ algorithm for \vc implies a parameterized $\beta$-approximation in time $c^{(2-\beta )k}$. The same running time has been attained for \vc by Fellows et al. \cite{Fellows2018} through a more generic framework which can be applied to additional problems.  

Other works focused on directly designing parameterized approximation algorithms. 
In \cite{BF13,brankovicParameterizedApproximationAlgorithms2012} Brankovic and Fernau provided  parameterized approximation algorithms for \vc and \hs{3}.  The designed algorithms are branching algorithms which involve branching rules  aimed for approximation and interleave approximative reduction rules.  The works provide, among others, a parameterized  $1.5$-approximation for $\vc$ in time $1.0883^k\cdot n^{\Oh(1)}$ and a parameterized $2$-approximation for \hs{3} in time $1.29^k\cdot n^{\Oh(1)}$. 

In \cite{KulikS2020}, Kulik and Shachnai  showed that the use of randomized branching can lead to significantly faster parameterized-approximation algorithms.
 In particular,  they attained a parameterized  $1.5$-approximation for \vc in time $1.017^k \cdot n^{\Oh(1)}$ and a parameterized  $2$-approximation for \hs{3} in time $1.0659^k\cdot n^{\Oh(1)}$. 
 The idea in randomized branching is that the algorithm picks one of the branching options at random. Subsequently, a good approximation is attained as long as the randomly picked  options are  not too far from the correct options.  The branching rules used by~\cite{KulikS2020} are fairly involved in comparison to the sampling steps used by this paper and their analysis required a non-trivial mathematical machinery.

While \cite{KulikS2020} showed that randomized branching is a powerful technique for the design of parameterized-approximation algorithms, it has only been applied to a limited set of problems.
 In \cite{KulikS2020}, the authors provided applications of the technique for \vc and \hs{3}. Furthermore, for approximation ratios close to $1$, the randomized-branching algorithms are inferior to the exact algorithms in terms of running times (or to the combination of those with \cite{Fellows2018}). A central goal of this paper is to overcome some of these difficulties: Sampling with a Black Box harnesses the power of randomized branching, it is applicable for a wide set of problems, and always improves upon the running times attained using the best known exact algorithms in conjunction with \cite{Fellows2018}.

As already mentioned, in \cite{janaParameterizedApproximationScheme2023}, the authors developed parameterized approximation algorithms for \FVS\ using the combination of randomized branching and existing algorithms as black-box. Their work indeed served as a motivation for this paper. We point out that the approach in \cite{janaParameterizedApproximationScheme2023} is restricted for \FVS\,and the resulting running times are inferior to ours. 

The concept of randomly sampling a partial solution, which is subsequently extended using a parameterized algorithm, is central to the {\em monotone local search} technique  of Fomin, Gaspers, Lokshtanov and Saurabh \cite{FGLS19}. Later variants of this technique are designed to obtain exponential time approximation algorithms \cite{EKMNS22,EKMNS23,EKMNS24}. They also use the same basic argument, which states that sampling a partial solution not too far from the optimum suffices to attain an approximate solution with high probability.

\paragraph{Organization} Section~\Cref{sec:definitions} gives several standard definitions used throughout the paper. In \Cref{sec:our_results} we formally state the results of the paper. \Cref{sec:applications} lists applications of the technique for several problems. Additional applications are  given in \Cref{sec:addtl_applications}. \Cref{sec:algorithms} provide our main  algorithm together with its proof of correctness. 
In \Cref{sec:sdelta_expl_formula} we provide a simpler formula for the running time Sampling with a Black Box for several cases. 
In \Cref{sec:sampling_step} we derive sampling steps for several problems.  \Cref{sec:fidelity_proof} provides the proof that the running time of Sampling with a Black Box is faster than Fidelity Preserving Transformations. Finally, we  discuss our results in \Cref{sec:discussion}.
For the sake of presentation, we include the proofs of some claims in the appendix.

\section{Preliminaries}
\label{sec:definitions}

In this section, we present several standard  concepts and notations used throughout this paper.
 
\paragraph{Graph Notations.} Given an hypergraph $G=(V,E)$ we use $V(G)=V$ and $E(G)=E$ to denote the sets of vertices and hyperedges of the graph (respectively). 

\paragraph{Vertex Induced Subhypergraphs and Vertex Deletion.} Given a a hypergraph $G=(V,E)$ and a subset of vertices $U\subseteq V$ , the {\em vertex induced subhypergraph} of $G$ and $U$ is the hypergraph $G[U] = (U,E')$ where $	E' \coloneqq \{e \in E \mid e \subseteq U\}$.
 We also define the {\em vertex deletion} of $U$ from $G$ by $G\setminus U = G[V\setminus U]$. 
  For a single vertex $u\in V(G)$ we use the shorthand $G\setminus v = G\setminus \{v\}$. 

\paragraph{Hypergraph Properties.} A {\em hypergraph property} is a set $\Pi$ of hypergraphs such that for every $G\in \Pi$ and $G'$ isomorphic to $G$ it holds that $G'\in\Pi$  as well. A hypergraph property $\Pi$ is called {\em hereditary} if for every $G\in \Pi$ and $U\subseteq V(G)$  it holds that $G[U]\in \Pi$ as well.  Furthermore, we say that $\Pi$ is polynomial-time decidable if given a hypergraph $G$ we can decide in polynomial-time whether $G \in \Pi$ or not.
We assume that all hypergraph properties discussed in this paper are hereditary and  polynomial-time decidable; that is by saying the $\Pi$ is a hypergraph property we also mean that it is hereditary and polynomial-time decidable.

\paragraph{A closed set of hypergraph.} We say that a set $\cG$ of hypergraph is {\em closed} if every vertex induced subhypergraph of a graph in $\cG$ is also in $\cG$. That is, for every $G\in \cG$ and $S\subseteq V(G)$ it also holds that $G[S] \in \cG$.

\paragraph{Kullback-Leibler Divergence} Given two number $a,b\in [0,1]$, the {\em Kullback-Leibler divergence} of $a$ and $b$ is 
$$
\D{a}{b} = a\cdot \ln\left(\frac{a}{b}\right)+(1-a)\cdot \ln\left(\frac{1-a}{1-b}\right).
$$
We follow that standard convention that $0\cdot \ln 0 =0\cdot \ln \frac{0}{x}  =0$, which implies
\begin{equation}
	\D{1}{b}=  1\cdot \ln\left( \frac{1}{b}\right) + (1-1) \cdot \ln \left(\frac{1-1}{1-b}\right) = -\ln(b).\label{eq:KL_eq_1}
\end{equation}

\section{Our Results}
\label{sec:our_results}

In this section we formally state our results, 
starting with some formal definitions. The definitions of vertex deletion problems and parameterized approximation algorithms are given in \Cref{sec:prob_def}. \Cref{sec:sampling_step_subsec} provides the definition of sampling steps. Then, we state our main result in \Cref{sec:blackboxsampling}. The comparison of our results  to those of \cite{Fellows2018} is finally given in \Cref{sec:comparison}. 
\subsection{Vertex Deletion Problems}
\label{sec:prob_def}
The focus of this paper is the class of $\gpivd$ problems. 
For any hypergraph property $\Pi$ and a closed set of hypergraphs $\mathcal{G}$, the input for $\gpivd$ ($\sgpivd$ for short) is a  hypergraph $G \in \mathcal{G}$. A {\em solution} is a subset of vertices $S\subseteq G(V)$  such that $G\setminus S\in \Pi$. 
We use $\sat_{\Pi}(G) = \{ S\subseteq G(V)\,|\,G\setminus S\in \Pi\}$ to denote the set of all solutions. The objective is to find a smallest cardinality solution $S \in \sat_\Pi(G)$.
In the decision version of the problem, we are given a hypergraph $G \in \mathcal{G}$, an integer $k \geq 0$ and we are asked whether there exists a set $S \in \sat_\Pi(G)$ such that $\abs{S} \leq k$.

The family of $\gpivd$ problems include many well known problems. Some notable examples are:
\begin{itemize}
	\item \vc. In this case $\mathcal{G}$ corresponds to graphs, i.e. hypergraphs with edge cardinality exactly 2. The hypergraph property $\Pi$ consists of all edgeless graphs.
	\item \vc on graphs of degree at most 3. Similar to the case above, in this case $\mathcal{G}$ corresponds to graphs of degree at most 3. $\Pi$ again consists of all edgeless graphs.
	\item \hs{d}. Similar to Vertex Cover, in this case $\mathcal{G}$ corresponds to hypergraphs with edge cardinality exactly $d$. $\Pi$ also consists of all edgeless hypergraphs.
	\item \fvs. In this case $\mathcal{G}$ is the set of all graphs and $\Pi$ is the set of graphs that have no cycles.
\end{itemize}

We use $\OPT_{\cG,\Pi}(G)$ to denote the size of an optimal solution for $G\in \cG$ with respect to the $\gpivd$ problem. That is, $\OPT_{\cG,\Pi}(G) = \min\Big\{ \abs{S}\,\Big|\,S \in \sat_{\Pi}(G) \Big\}$.  If $\cG$ and $\Pi$ are known by context we use $\OPT$ instead of $\OPT_{\cG,\Pi}$. 

Our goal is to develop {\em parameterized approximation} algorithms for $\gpivd$ problems, where the parameter is the solution size.  Such algorithms return a solution of size $\alpha \cdot k$ or less (with probability at least $\frac{1}{2}$) if the optimum of the instance is at most $k$
\begin{definition}[Parameterized Approximation]
	\label{definition:alpha_approx}
	Let $\Pi$ be a hypergraph property and $\mathcal{G}$ be a closed set of hypergraphs. 
	An algorithm $\mathcal{A}$ is a {\em randomized parameterized $\alpha$-approximation algorithm} for $\gpivd$  if it takes a graph $G \in \mathcal{G}$ and an integer $k \geq 0$ as input, and returns  a solution $S\in \sat_{\Pi}(G)$ which satisfies the following.
	\begin{itemize}
		\item If $\OPT_{\cG,\Pi}(G) \leq k$,  then $\Pr\left( \abs{S}\leq \alpha \cdot k \right) \geq \frac{1}{2}$.
	\end{itemize}
	Moreover, the running time of $\mathcal{A}$ is $f(k) \cdot n^{\Oh(1)}$ for some function $f$.
\end{definition}

We note that the above definition of parameterized approximation algorithms captures the classic definition of exact parameterized algorithms ($\alpha=1$) and polynomial time approximation algorithms ($f(k) = O(1)$) as special cases. 
Sampling with a Black-Box converts a randomized parameterized  $\alpha$-approximation algorithm   $\mathcal{A}$ which runs in time $c^k\cdot n^{\Oh(1)}$ to a randomized parameterized  $\beta$-approximation algorithm $\mathcal{B}$  which runs in time $d^k\cdot n^{\Oh(1)}$. 
The conversion relies on a simple  problem dependent {\em sampling step}.
In all our application the algorithm $\mathcal{A}$ is either the state-of-art exact parameterized algorithm  (i.e., $\alpha=1$), or the state-of-art polynomial time algorithm (i.e., $c=1$) for the problem. 

\subsection{Sampling Steps}\label{sec:sampling_step_subsec}
We exploit inherent properties of a specific $\gpivd$  problem to come up with
basic sampling strategies.
From a high-level perspective, a sampling step is a polynomial-time algorithm that takes as input a hypergraph $G \in \mathcal{G} \setminus \Pi$ and returns a vertex $v \in V(G)$,
such that the size of the optimal solution of $G\setminus v$ decreases by $1$, with a prescribed probability.

\begin{definition}[Sampling Step]\label{definition:sampling_step}
	
	Let $\cG$ be a closed set of hypergraph, $\Pi$ be a hypergraph property and $q\in (0,1)$. A {\em sampling step with success probability $q$} is a  polynomial time randomized algorithm $\mathcal{R}$
	that takes as input a hypergraph $G \in \mathcal{G} \setminus \Pi$ and  returns a vertex $v \in V(G)$ such that 
		 $$\Pr\Bigl( \OPT_\Pi(G \setminus v) \leq \OPT_\Pi(G) - 1 \Bigr) \geq q.$$
\end{definition}

For example, consider
$\vc$,  which is a $\gpivd$ problem where $\mathcal{G}$ is the set of all
graphs and $\Pi$ consists of all edgeless graphs. The following is a very
simple sampling step for $\vc$ with success probability $\frac{1}{2}$:
pick an arbitrary edge and return each of its endpoints with probability
$\frac{1}{2}$. This algorithm clearly runs in polynomial time. Moreover, for
each $\vc$ $S$ of $G$ and for each edge $e$, at least one of the endpoints of  $e$  belongs to $S$. Therefore, its a sampling step with success probability $q=\frac{1}{2}$ for $\vc$. We provide sampling steps for the general case where $\Pi$ is defined by a finite set of forbidden subgraphs, for $\fvs$ and for $\povd$.

On their own, sampling steps can be easily used to derive parameterized approximation algorithms. To obtain a $\beta$-approximation, one may use the sampling step $\beta\cdot k$ times, where each execution returns a vertex $v$ which is removed from the graph and added to the solution $S$. After $\beta\cdot k$ steps, with some probability $P$, the set $S$ of returned vertices is indeed a solution. Thus, by repeating this multiple sampling step for  $\frac{1}{P}$ times, one gets a $\beta$-approximate solution with probability $\frac{1}{2}$. 
By carefully tracing the probability in the above argument, we show that $P\approx \left(\exp\left(\beta \cdot \D{\frac{1}{\beta}}{q}\right) \right)^{-k}$, as demonstrated by the following lemma.\footnote{Recall $\D{\cdot}{\cdot}$ stands for  the Kullback-Leibler divergence which has been defined in \Cref{sec:definitions}.} 

\begin{lemma}
\label{lem:sampling_to_approximation}
Let  $\mathcal{G}$ be a closed  set of hypergraphs  and $\Pi$ be an hypergraph property such that there is a sampling step with success probability $q$ for $\sgpivd$. Then for every $1\leq \beta \leq \frac{1}{q}$ there is a randomized parameterized $\beta$-approximation for $\sgpivd$ which runs in time $d^{k}\cdot n^{\Oh(1)}$ where  $d=\exp\left(\beta\cdot \D{\frac{1}{\beta}}{q}\right) $.
\end{lemma}

The proof of \cref{lem:sampling_to_approximation} can be found in \cref{sec:algorithms}.
Observe that  $\exp\left(\frac{1}{q}\cdot \D{\frac{1}{\left(\frac{1}{q}\right)}}{q}\right)= 1$ for every $q\in (0,1)$. Thus, \Cref{lem:sampling_to_approximation} provides a polynomial time $\frac{1}{q}$-approximation algorithm for $\sgpivd$. This justifies the restriction of the lemma to $\beta\leq \frac{1}{q}$.

While  \Cref{lem:sampling_to_approximation} provides a parameterized $\beta$-approximation algorithm for essentially any $\beta$, the resulting algorithms are often clearly far from optimal. For example, for $\fvs$ ($\FVS$) we provide a sampling step with success probability $\frac{1}{4}$.
Thus, by \Cref{lem:sampling_to_approximation} we get a  parameterized $1.1$-approximation for $\FVS$ which runs in time $\approx 2.944^k \cdot n^{\Oh(1)}$. However, $\FVS$ has an {\em exact} parameterized algorithm which runs in time $2.7^k\cdot n^{\Oh(1)}$. That is, the running time of the approximation algorithm is slower than that of the exact algorithm. Our goal is to combine the sampling step with the exact algorithm to attain improved running time in such cases.

\subsection{Sampling with a Black-Box}
\label{sec:blackboxsampling}

Our main theorem states the following: a sampling step with success probability $q$, together with a parameterized $\alpha$-approximation algorithm $\mathcal{A}$ which runs in time $c^k\cdot n^{\Oh(1)}$, can be used to obtain a $\beta$-approximation algorithm $\mathcal{B}$. 
By \Cref{lem:sampling_to_approximation}, the sampling step can be used to obtain a randomized parameterized $\alpha$-approximation which runs in time $\left( \exp\left(\alpha\cdot \D{\frac{1}{\alpha}}{q}\right)\right)^k\cdot n^{\Oh(1)}$. Our assumption is that $\mathcal{A}$ is at least as fast as the algorithm provided by \Cref{lem:sampling_to_approximation}, hence we only consider the case in which~$c\leq  \exp\left(\alpha\cdot \D{\frac{1}{\alpha}}{q}\right)$. 

We use the following functions to express the running time of $\mathcal{B}$.  Define two function $\sdeltal(\alpha,c,q)$ and  $\sdeltar(\alpha,c,q)$ as the unique numbers $\sdeltal(\alpha,c,q)\in \left(1,\alpha \right]$ and $\sdeltar(\alpha,c,q)\in \left[\alpha,\infty \right)$ which satisfy
	\begin{equation}\label{eq:deltast_def}
	\D{\frac{1}{\alpha}}{\frac{1}{\sdeltal(\alpha,c,q)}} = \D{\frac{1}{\alpha}}{\frac{1}{\sdeltar(\alpha,c,q)}}= \D{\frac{1}{\alpha}}{q} - \frac{\ln(c)}{\alpha}.
\end{equation}
	We write \(\sdeltal = \sdeltal(\alpha,c,q)\) and \(\sdeltar = \sdeltar(\alpha,c,q)\) if the values of $\alpha, c$ and $q$ are clear from the context.
The following lemma provides conditions which guarantee that $\sdeltal$ and $\sdeltar$ are well defined in certain domains.
\begin{restatable}{lemma}{sdeltawelldefined}\label{lemma:sdelta_well_defined}
For every $c\geq 1$, $0<q<1$ and $\alpha\geq 1$ such that $c\leq \exp\left(\alpha\cdot \D{\frac{1}{\alpha}}{q}\right)$, the value of $\sdeltar(\alpha,c,q)$ is well defined. Furthermore, if $\alpha>1$, then $\sdeltal(\alpha,c,q)$ is also well defined.	
\end{restatable}

The proof of \cref{lemma:sdelta_well_defined} can be found in \Cref{sec:sdelta_expl_formula}.
We note that $\D{\frac{1}{\alpha}}{\frac{1}{x}}$ is monotone in  the domains  $x\in (1,\alpha]$ and $x\in [\alpha,\infty)$. Hence the values of  $\sdeltal(\alpha,c,q)$  and $\sdeltar(\alpha,c,q)$ can be easily evaluated to arbitrary precision using  a simple binary search. 

We use $\sdeltal$ and $\sdeltar$ to express the running time of $\mathcal{B}$.
For every $\beta,\alpha \geq 1$, $0< q \leq 1$ and $c\geq 1$ such that $c\leq \exp\left(\alpha \cdot \D{\frac{1}{\alpha}}{q}\right)$, we define
	\begin{equation}
		\label{eq:runtime}
		\runtime \coloneqq \begin{cases}
			c \cdot \exp\biggl(\frac{ \sdeltar \cdot \D{\frac{1}{\sdeltar}}{q}- \ln(c) }{\sdeltar - \alpha} \cdot \left( \beta - \alpha \right) \biggr) &\text{if } \sdeltar(\alpha, c, q) > \beta \geq\alpha\\ 
			c \cdot \exp\biggl(\frac{\sdeltal \cdot \D{\frac{1}{\sdeltal}}{q}- \ln(c) }{\sdeltal - \alpha} \cdot \left( \beta - \alpha \right) \biggr) &\text{if } \sdeltal(\alpha, c, q) < \beta \leq \alpha\\
			\exp\left(\beta\cdot \D{\frac{1}{\beta}}{q}\right) &\textnormal{ otherwise}
		\end{cases}
	\end{equation}
	As we can compute $\sdeltal$ and $\sdeltar$, it follows we can also compute $\runtime{}$ to arbitrary precision.
	
Recall that \Cref{lem:sampling_to_approximation} provides a polynomial time $\frac{1}{q}$-approximation algorithm for $\sgpivd$. Consequently, we restrict our consideration to $\alpha$ and $\beta$ values that are less than or equal to $\frac{1}{q}$. Our main technical result is the following.
\begin{restatable}[Sampling with a Black-Box]{theorem}{samplingblackbox}\label{theorem:summary_sampling_step}
	Let $\cG$ be a closed set of hypergraphs and $\Pi$ be a hypergraph property.
	Assume the following:
	\begin{itemize}
		\item There is a sampling step with success probability $q\in(0,1)$ for $\sgpivd$.
		\item There is a randomized parameterized $\alpha$-approximation algorithm for $\sgpivd$ which runs in time $c^{k}\cdot n^{\Oh(1)}$ for some $c \geq 1$, $1 \leq \alpha \leq \frac{1}{q}$ and $c\leq \exp\left(\alpha \cdot \D{\frac{1}{\alpha}}{q} \right)$. 
	\end{itemize}
Then, for every $1\leq \beta\leq \frac{1}{q}$, there is a randomized parameterized $\beta$-approximation for $\sgpivd$ which runs in time $\left( \runtime\right)^{k} \cdot n^{\Oh(1)}$.
\end{restatable}

The proof of \Cref{theorem:summary_sampling_step} is given in \Cref{sec:algorithms}. 
For example, by \Cref{theorem:summary_sampling_step}, we can use the sampling step with success probability $\frac{1}{4}$ for $\FVS$, together with the exact parameterized algorithm for the problem which runs in time $2.7^k\cdot n^{\Oh(1)}$ \cite{liDetectingFeedbackVertex2022}, to get a randomized parameterized $1.1$-approximation algorithm with running time $2.49^k \cdot n^{\Oh(1)}$. The algorithm achieves a better running time than that of the exact parameterized algorithm, as well as the running time which can be attain solely by the sampling step, i.e. $2.944^{k} \cdot n^{\Oh(1)}$ (\Cref{lem:sampling_to_approximation}). We note that this running time is also superior to the running time of $2.62^{k}\cdot n^{\Oh(1)}$ for the same approximation ratio given in \cite{janaParameterizedApproximationScheme2023}. 

The running time of the algorithm generated by \Cref{theorem:summary_sampling_step}  (i.e., the value of $\runtime$) can always be computed efficiently, though this computation requires a binary search for the evaluation of $\sdeltar$ and $\sdeltal$. For the special cases of $\alpha =1$ and well as ($\alpha=2$ and $c=1$) we  provide a closed form expression for $\runtime$. 
\begin{restatable}[simple formula for $\alpha=1$]{theorem}{simpslaphaone}
\label{theorem:summary_sampling_step_alpha_1}
For every   $0 < q < 1$, $1 \leq \beta \leq \frac{1}{q}$ and $c \geq 1$ such that $c \leq \exp\left(1\cdot \D{\frac{1}{1}}{q}\right)  = \frac{1}{q}$ it holds that
\begin{equation*}
	\runtime[1,\beta,c,q] = \begin{cases}
		c \cdot \left( \frac{1 - c \cdot q}{1 - q} \right)^{\beta - 1} &\text{if } 1 \leq \beta < \frac{1}{q \cdot c}\\ 
		\exp\left( \beta \cdot \D{\frac{1}{\beta}}{q} \right) &\text{if } \frac{1}{q \cdot c} \leq \beta \leq \frac{1}{q} 
	\end{cases}
\end{equation*}
\end{restatable}

\begin{restatable}[simple formula for $\alpha =2$ and $c=1$]{theorem}{simplealphatwo}
	\label{theorem:summary_sampling_step_alpha_2}
	Let $0 < q \leq \frac{1}{2}$ and  $1 \leq \beta \leq 2$.
	Then it holds that 
	\begin{equation*}
		\runtime[2,\beta,1,q] = \begin{cases}
			\exp\left( \beta \cdot \D{\frac{1}{\beta}}{q} \right) &\text{if } 1 \leq \beta \leq \frac{1}{1 - q}\\
			\left( \frac{q}{1 - q} \right)^{\beta - 2}  &\text{if } \frac{1}{1 - q} < \beta \leq 2\\
		\end{cases}
	\end{equation*}
\end{restatable}

Both \Cref{theorem:summary_sampling_step_alpha_1,theorem:summary_sampling_step_alpha_2}  follow from a closed form formula which we can attain for $\sdeltar$ or~$\sdeltal$  for the specific values of $\alpha$ and $c$ considered in the theorems. The proof of \Cref{theorem:summary_sampling_step_alpha_1,theorem:summary_sampling_step_alpha_2} is given in \Cref{sec:sdelta_expl_formula}. 

\subsection{Comparison to Fidelity Preserving Transformations}
\label{sec:comparison}

We compare our technique, Sampling with a Black-Box, with the technique of \cite{Fellows2018} for \nmgpivd\  problems in which $\Pi$ is defined by a finite set of forbidden vertex  induced hypergraphs. 

\begin{definition}\label{definition:forbidden_hypergraph}
	Let $\Omega = \{F_1, \ldots, F_\ell\} $ be a finite set of hypergraphs for $\ell > 0$.
		Then $\Pi^{\Omega}$ is the hypergraph property where  a hypergraph $G$ belongs to $\Pi^{\Omega}$
	if and only if there is no vertex induced subhypergraph $X$ of $G$ such that $X$ is isomorphic to $F_i$,
	for some $1 \leq i \leq \ell$.
\end{definition}
We note that $\Pi^{\Omega}$ is always hereditary and polynomial-time decidable.  We also note the the family of graphs properties defined by a finite set of forbidden hypergraph suffices to define many fundamental graph problems such as \vc, \hs{d}, \pathvc{\ell} and \dfvst.  

For a set of hypergraphs $\Omega = \{F_1, \ldots, F_\ell\}$, define
$\eta(\Omega) \coloneqq \max_{1 \leq i \leq \ell} \abs{V\left( F_i \right) }$, the maximal number of vertices of a hypergraph in $\Omega$. The following result has been (implicitly) given in  \cite{Fellows2018}.
\begin{lemma}[\cite{Fellows2018}]
	\label{lemma:fellows2018}
	Let $\Omega$ be a finite set of hypergraphs and let $\cG$ be a closed set of hypergraphs. Furthermore, assume there is an randomized exact parameterized $c^k \cdot n^{\Oh(1)}$ algorithms for $\gpivd[\cG, \Pi^{\Omega}]$. Then for every $1\leq\beta\leq \eta(\Omega)$ there is a randomized parameterized $\beta$-approximation algorithm for $\gpivd[\cG, \Pi^{\Omega}]$ which runs in time $$c^{\frac{\eta(\Omega)-\beta}{\eta(\Omega)-1}\cdot k}\cdot n^{\Oh(1)}.$$ 
\end{lemma}

There is a simple and generic way to design a sampling step for $\Pi^{\Omega}$. The sampling step finds a set of vertices $S\subseteq V(G)$  of the input graph such that $G[V]$ is isomorphic for a graph in $\Omega$, and returns a vertex from $S$ uniformly at random. This leads to the following lemma. 
\begin{lemma}\label{lemma:finite_forb_sampling}
	Let $\Omega$ be a finite set of hypergraphs and let $\cG$ be  a closed set of hypergraphs.
	There is a sampling step for $\gpivd[\mathcal{G}, \Pi^{\Omega}]$ with success probability $\frac{1}{\eta(\Omega)}$.
\end{lemma}
A formal proof for \Cref{lemma:finite_forb_sampling} is given in 	\Cref{sec:sampling_step}. Together with \Cref{theorem:summary_sampling_step} the lemma implies the following. 
\begin{corollary}\label{cor:finite_forbidden}
		Let $\Omega$ be a finite set of hypergraphs and $\cG$ be  a closed set of hypergraphs. Furthermore, assume there is a randomized exact parameterized $c^k \cdot n^{\Oh(1)}$ algorithms for $\gpivd[\cG, \Pi^{\Omega}]$. Then for every $1\leq\beta\leq \eta(\Omega)$ there is a randomized parameterized $\beta$-approximation algorithm for $\gpivd[\cG, \Pi^{\Omega}]$ which runs in time $$\left(\runtime[1,\beta,c,\frac{1}{\eta(\Omega)}] \right) ^{k}\cdot n^{\Oh(1)}.$$ 
\end{corollary}

Observe that \Cref{lemma:fellows2018} and \Cref{cor:finite_forbidden} only differ in the resulting running time. The following lemma implies  that for every $1<\beta<\eta(\Omega)$ the running time of \Cref{cor:finite_forbidden}, the running time of Sampling with a Black-Box, is always strictly better than the running time of \Cref{lemma:fellows2018}, the result of \cite{Fellows2018}. 

\begin{restatable}{lemma}{fidelitycomparison}
	\label{lemma:comparison}
	For every   $\eta \in \mathbb{N}$ such that $\eta \geq 2$, $1<c<\eta$  and $1<\beta <\eta$ it holds that	
	$$\runtime[1,\beta,c,\frac{1}{\eta}] \,<\, c^{\frac{\eta-\beta}{\eta-1}}.$$
	
\end{restatable}
The proof of \Cref{lemma:comparison} is given in \cref{sec:fidelity_proof}.

\section{Applications}
\label{sec:applications}
In this section we will describe some problems to which our results can be applied.
For each problem, we utilize sampling steps to obtain parameterized approximation algorithms. 
We also compare the running time of our algorithm with a benchmark whenever applicable.

\subsection{\fvs}
Recall that given a graph $G$ and integer $k$, the \fvs problem asks whether
there exists a set $S \subseteq V(G)$ of size at most $k$ such that $G \setminus S$ is acyclic.
\fvs can also be described as a $\gpivd$ problem,
where $\mathcal{G}$ is the set of all graphs and $\Pi$ is the set of graphs that have no cycles.
First, we demonstrate that there exists a sampling step for \fvs.
\begin{lemma}\label{lemma:fvs_sampling}
	\fvs has a sampling step with success probability $\frac{1}{4}$.
\end{lemma}

The proof of \cref{lemma:fvs_sampling} can be found in \cref{sec:fvs}.
The sampling step presented in \cref{sec:fvs} begins by removing vertices
of degree at most 1. Then, a vertex is sampled from the remaining vertices,
where the sampling probability for each vertex is proportional to its
degree in the remaining graph.

In \cite{liDetectingFeedbackVertex2022}, the authors present an FPT algorithm which runs in time $2.7^{k} \cdot n^{\Oh(1)}$ (i.e., in the terminology of this paper, $\alpha = 1, c = 2.7$).
Moreover, \fvs also has a 2-approximation algorithm that runs in polynomial time
(i.e., $\alpha = 2$, $c = 1$) \cite{bafna2ApproximationAlgorithmUndirected1999}.
In the following, we demonstrate how the sampling
step, together with the existing algorithms,
is used to develop a new approximation algorithm with a better running time.

\begin{theorem}
	For each $1 \leq \beta \leq 2$, \fvs has a $\beta$-approximation algorithm
	which runs in time $d^{k} \cdot n^{\Oh(1)}$ where
	\begin{equation}\label{eq:fvs_running_time_final}
		d = \begin{cases}
				2.7 \cdot \left( \frac{1.3}{3} \right) ^{\beta - 1} &\text{if }\, 1 \leq \beta < 1.402\\			
				 \left(\frac{1}{3}\right)^{\beta - 2} &\text{if }\, 1.402 \leq \beta \leq 2\\
			\end{cases}.
	\end{equation}
\end{theorem}

\begin{proof}
	Let $\mathcal{A}_1$ denote the FPT algorithm from \cite{liDetectingFeedbackVertex2022} and $\mathcal{A}_2$ denote the 2-approximation algorithm from \cite{bafna2ApproximationAlgorithmUndirected1999} that runs in polynomial time.
	Note that when we consider $\beta$-approximation algorithms, we can focus on the values of
	$\beta$ in the range $1 < \beta \leq 2$ because of $\mathcal{A}_2$,
	as for larger values of $\beta$ we immediately get a $\beta$-approximation
	algorithm that runs in polynomial time.
	
	By using $\mathcal{A}_1$ and \cref{theorem:summary_sampling_step_alpha_1}, the first $\beta$-approximation algorithm we obtain
	has the running time $d^{k} \cdot n^{\Oh(1)}$ where
		\begin{equation}\label{eq:fvs_running_time_1}
			d = \begin{cases}
				2.7 \cdot \left(0.433\right) ^{\beta - 1} &\text{if } 1 \leq \beta < 1.481\\
				\exp\left( \beta \cdot \D{\frac{1}{\beta}}{\frac{1}{4}} \right) &\text{if } 1.481 \leq \beta \leq 2.
			\end{cases}
		\end{equation}
	
	Similarly, by using $\mathcal{A}_2$ and \cref{theorem:summary_sampling_step_alpha_2}, the second $\beta$-approximation algorithm we obtain has the running time
	$d^{k} \cdot n^{\Oh(1)}$ where
	\begin{equation}\label{eq:fvs_running_time_2}
		d = \begin{cases}
				 \exp\left( \beta \cdot \D{\frac{1}{\beta}}{\frac{1}{4}} \right)  &\text{if } 1 \leq \beta < 1.333\\
				 \left(\frac{1}{3}\right)^{\beta - 2} &\text{if } 1.333 \leq \beta \leq 2\\
			\end{cases}.
	\end{equation}
	
	Note that for each $1 < \beta$, we can choose the algorithm with the faster running
	time out of \eqref{eq:fvs_running_time_1} and \eqref{eq:fvs_running_time_2}.
	Therefore we compare the base of the exponents in the running time and pick the smallest one.
	After a straightforward calculation, one can observe that for $1 \leq \beta < \frac{\ln\left( \frac{13}{9} \right) }{\ln\left( 1.3 \right) } \approx 1.402$,
	$2.7 \cdot \left( \frac{1.3}{3} \right) ^{\beta - 1}$ is the smallest number. Similarly,
	for $1.402 \leq \beta \leq 2$, the smallest number becomes $\left(\frac{1}{3}\right)^{\beta - 2}$.
	Therefore,
	for each $1 < \beta < 2$,
	we obtain an algorithm with a
	a running time of $d^{k} \cdot n^{\Oh(1)}$ where
	\begin{equation}\label{eq:fvs_final_base}
		d = \begin{cases}
				2.7 \cdot \left( \frac{1.3}{3} \right) ^{\beta - 1} &\text{if } 1 \leq \beta < 1.402\\			
				 \left(\frac{1}{3}\right)^{\beta - 2} &\text{if } 1.402 \leq \beta \leq 2\\
			\end{cases}.
	\end{equation}
\end{proof}

In \cite{janaParameterizedApproximationScheme2023}, the authors present a $\beta$-approximation
algorithm for each $1 < \beta \leq 2$. It can be visually (see \cref{fig:fvs}) and numerically (see \cref{tab:fvs}) verified that our algorithm demonstrates a
strictly better running-time .

\begin{table}[h]
    \centering
    \begin{tabular}{|c|c|c|}
        \hline
	$\beta$ & Our Algorithm \eqref{eq:fvs_final_base} & \cite{janaParameterizedApproximationScheme2023} \\
        \hline
        1.1 & $2.483$ & $2.620$ \\
        1.2 & $2.284$ & $2.467$ \\
        1.3 & $2.101$ & $2.160$ \\
        1.4 & $1.932$ & $1.942$ \\
        1.5 & $1.732$ & $1.778$ \\
        1.6 & $1.552$ & $1.649$ \\
        1.7 & $1.390$ & $1.56$ \\
        1.8 & $1.246$ & $1.319$ \\
        1.9 & $1.116$ & $1.149$ \\
        \hline
    \end{tabular}
    \caption{Comparison of the base of exponents of different algorithms for \FVS.
 	   For each row with a $\beta$ value $b$, a value $d$ in the second or third column
	implies a $b$-approximation algorithm with running time $d^{k} \cdot n^{\Oh(1)}$.}    
    \label{tab:fvs}
\end{table}

\subsection{\povd}
Given a graph $G$ and integer $k$, the \povd (\POVD) problem ask whether there exists a set $S \subseteq V(G)$
of size at most $k$ such that $G \setminus S$ has pathwidth at most 1.
Initially, we demonstrate that there exists a sampling step for \POVD.

\begin{lemma}\label{lemma:povd_sampling}
	\povd has a sampling step with probability $\frac{1}{7}$.
\end{lemma}

The proof of \cref{lemma:povd_sampling} can be found in \cref{sec:povd}.
The sampling step for \POVD is very similar to that for \FVS, with a slight modification.
Similar to \FVS, \POVD can be described by a set of forbidden subgraphs.
Moreover, the set of forbidden subgraphs for \POVD includes one additional graph with 7 vertices. Therefore \POVD has a sampling step with probability $\frac{1}{7}$,
instead of $\frac{1}{4}$ as in the case of \FVS. 

To the best of our knowledge, parameterized approximation algorithms have not been studied for \POVD.
However, there exists an FPT algorithm for \POVD with running time $3.888^{k} \cdot n^{\Oh(1)}$ ($\alpha = 1, c = 3.888$) \cite{tsurFasterAlgorithmPathwidth2022}.
Next, we demonstrate how combining the aforementioned sampling step with a parameterized algorithm yields a new approximation algorithm. Refer to \cref{fig:povdplot,tab:povd} for the corresponding running time.

\begin{theorem}
	For each $1 \leq \beta \leq 7$, \povd has a $\beta$-approximation algorithm
	which runs in time $d^{k} \cdot n^{\Oh(1)}$ where
	\begin{equation}\label{eq:povd_running_time_final}
		d = \begin{cases}
		3.888 \cdot \left(0.519\right)^{\beta - 1} &\text{if } 1 \leq \beta \leq 1.8 \\ 
		\exp\left( \beta \cdot \D{\frac{1}{\beta}}{\frac{1}{7}} \right) &\text{if } 1.8 < \beta < 7.
	\end{cases}
	\end{equation}	
\end{theorem}

\begin{proof}
	Let $\mathcal{A}$ be the FPT algorithm from \cite{tsurFasterAlgorithmPathwidth2022}, with running time $3.888^{k} \cdot n^{\Oh(1)}$.
	By using $\mathcal{A}$ and \cref{theorem:summary_sampling_step_alpha_1}, we obtain a $\beta$-approximation algorithm with running time $d^{k} \cdot n^{\Oh(1)}$ where
	\begin{equation*}
		d = \begin{cases}
		3.888 \cdot \left(0.519\right)^{\beta - 1} &\text{if } 1 \leq \beta \leq 1.8 \\ 
		\exp\left( \beta \cdot \D{\frac{1}{\beta}}{\frac{1}{7}} \right) &\text{if } 1.8 < \beta < 7.
	\end{cases}
	\end{equation*}
\end{proof}

\begin{figure}[h!]
	\centering
	\input{figures/povd.tex}
	\caption{A plot of the running time of our algorithm for \povd.
	The $x$-axis corresponds to the approximation ratio, while the $y$-axis corresponds to the base of the exponent in the running time.
	A point $(\beta, d)$ in the plot describes a running time of the form $d^{k} \cdot n^{\Oh(1)}$
	for a $\beta$-approximation. }	
	\label{fig:povdplot}
\end{figure}

\begin{table}[h]
    \centering
    \begin{tabular}{|c|c|c|c|c|c|c|c|c|c|c|c|}
        \hline
        $\beta$ & Value & $\beta$ & Value & $\beta$ & Value & $\beta$ & Value & $\beta$ & Value & $\beta$ & Value \\
        \hline
        1.1 & 3.6412 & 2.1 & 1.9391 & 3.1 & 1.3776 & 4.1 & 1.1573 & 5.1 & 1.0553 & 6.1 & 1.0107 \\
        \hline
        1.2 & 3.4100 & 2.2 & 1.8498 & 3.2 & 1.3466 & 4.2 & 1.1433 & 5.2 & 1.0488 & 6.2 & 1.0083 \\
        \hline
        1.3 & 3.1936 & 2.3 & 1.7713 & 3.3 & 1.3181 & 4.3 & 1.1303 & 5.3 & 1.0428 & 6.3 & 1.0063 \\
        \hline
        1.4 & 2.9908 & 2.4 & 1.7018 & 3.4 & 1.2920 & 4.4 & 1.1183 & 5.4 & 1.0374 & 6.4 & 1.0046 \\
        \hline
        1.5 & 2.8010 & 2.5 & 1.6399 & 3.5 & 1.2679 & 4.5 & 1.1071 & 5.5 & 1.0323 & 6.5 & 1.0031 \\
        \hline
        1.6 & 2.6232 & 2.6 & 1.5844 & 3.6 & 1.2457 & 4.6 & 1.0968 & 5.6 & 1.0277 & 6.6 & 1.0020 \\
        \hline
        1.7 & 2.4567 & 2.7 & 1.5346 & 3.7 & 1.2252 & 4.7 & 1.0871 & 5.7 & 1.0236 & 6.7 & 1.0011 \\
        \hline
        1.8 & 2.3007 & 2.8 & 1.4895 & 3.8 & 1.2062 & 4.8 & 1.0782 & 5.8 & 1.0198 & 6.8 & 1.0005 \\
        \hline
        1.9 & 2.1604 & 2.9 & 1.4487 & 3.9 & 1.1886 & 4.9 & 1.0700 & 5.9 & 1.0164 & 6.9 & 1.0001 \\
        \hline
        2.0 & 2.0417 & 3.0 & 1.4115 & 4.0 & 1.1724 & 5.0 & 1.0624 & 6.0 & 1.0134 & & \\
        \hline
    \end{tabular}
    \caption{The table displays the base of exponents for our algorithm designed for \povd.
    Each pair $(b,d)$, listed in the same row and consecutive $\beta$ and value columns,
	represents a $\beta$-approximation algorithm with a running time $d^{k} \cdot n^{\Oh(1)}$.}
    \label{tab:povd}
\end{table}

\subsection{\boldmath $(\mathcal{G},\Pi)$-Vertex Deletion for a finite set of forbidden sub-hypergraphs}
\label{sec:pi_fixed_appl}
There are many problems that can be described as \nmgpivd\  problems in which
$\Pi$ is defined by a finite set $\Omega$ of forbidden vertex induced hypergraphs. For
each of those problems, by \cref{lemma:finite_forb_sampling} there exists a
sampling step with success probability $\frac{1}{\eta}$, where $\eta$ is the maximum
number of vertices of a hypergraph in $\Omega$. In the following, we will demonstrate
how we can obtain parameterized approximation algorithms for such problems.
For the sake of presentation, we will focus on a specific problem called \pathvc{3}.

Given a graph $G$, a subset of vertices $S \subseteq V(G)$ is called an $\lpvc$
if every path of length $\ell$ contains a vertex from $S$.
The $\pathvc{\ell}$ problem asks whether there exists an $\lpvc$ of size at most $k$ where
$k$ is the parameter \cite{bresarMinimumKpathVertex2011}. $\pathvc{\ell}$ can be described
as a $\gpivd$ where $\mathcal{G}$ is the set of graphs and $\Pi$ is the set of graphs
with maximum path length at most $\ell - 1$. Alternatively, let $F$ be a path with $\ell$
vertices where we define $\Omega \coloneqq \{F\}$ and $\eta(\Omega) \coloneqq \ell$.
It holds that $\pathvc{\ell}$ is equivalent to $\gpivd[\mathcal{G}, \Pi^{\Omega}]$.
Therefore, by \cref{lemma:finite_forb_sampling}, there is
a sampling step for $\pathvc{\ell}$ with success probability $\frac{1}{\ell}$.

In the following, we will consider $\ell = 3$, i.e. the problem \pathvc{3}.

There exists an FPT algorithm for \pathvc that runs in time
$1.708^{k} \cdot n^{\Oh(1)}$ ($\alpha = 1, c = 1.708$)  \cite{cervenyGeneratingFasterAlgorithms2023}.
Moreover, there is also a 2-approximation algorithm 
that runs in polynomial time. ($\alpha = 2, c = 1$) \cite{tuFactorApproximationAlgorithm2011}.

\begin{theorem}
	For each $1 < \beta < 2$, \pathvc{3} has a $\beta$-approximation algorithm
	which runs in time $d^{k} \cdot n^{\Oh(1)}$ where
	\begin{equation}\label{eq:3pvc_runtime}
		d = \begin{cases}
			1.708 \cdot \left(0.644\right)^{\beta - 1} &\text{if } 1 \leq \beta < 1.6143\\
			\left(0.5\right)^{\beta - 2}  &\text{if } 1.6143 \leq \beta \leq 2
		\end{cases}
	\end{equation}
	
\end{theorem}

\begin{proof}
	Let $\mathcal{A}_1$ denote the FPT algorithm from \cite{cervenyGeneratingFasterAlgorithms2023} and $\mathcal{A}_2$ denote the 2-approximation algorithm from \cite{tuFactorApproximationAlgorithm2011}
	that runs in polynomial time. Because of $\mathcal{A}_2$, as in the case of \fvs,
	we can focus on the values of $\beta$ in the range $1 \leq \beta \leq 2$.

	By using $\mathcal{A}_1$ and \cref{theorem:summary_sampling_step_alpha_1},
	it holds that for each $\beta > 1$ there exists a $\beta$-approximation that runs
	in time $d^{k} \cdot n^{\Oh(1)}$ where
		\begin{equation}\label{eq:lpathvc_eq_1}
			d = \begin{cases}
				1.708 \cdot \left(0.644\right)^{\beta - 1} &\text{if } 1 \leq \beta < 1.752\\
				\exp\left( \beta \cdot \D{\frac{1}{\beta}}{\frac{1}{3}} \right) &\text{if } 1.752 \leq \beta \leq 2.
			\end{cases}
		\end{equation}	
	
	By using $\mathcal{A}_2$ together with \cref{theorem:summary_sampling_step_alpha_2},
	it holds that for every $1 < \beta < 2$ there exists a $\beta$-approximation algorithm
	that runs in time $d^{k} \cdot n^{\Oh(1)}$ where
	\begin{equation}\label{eq:lpathvc_eq_2}
		d = \begin{cases}
			\exp\left( \beta \cdot \D{\frac{1}{\beta}}{\frac{1}{3}} \right) &\text{if } 1 < \beta < 1.5\\
			\left(0.5\right)^{\beta - 2}  &\text{if } 1.5 \leq \beta < 2.
		\end{cases}
	\end{equation}	
	
	By taking the minimum of the running times in \eqref{eq:lpathvc_eq_1} and \eqref{eq:lpathvc_eq_2},
	the base of exponent becomes
	\begin{equation*}
		d = \begin{cases}
			1.708 \cdot \left(0.644\right)^{\beta - 1} &\text{if } 1 < \beta < 1.6143\\
			\left(0.5\right)^{\beta - 2}  &\text{if } 1.6143 \leq \beta < 2
		\end{cases}
	\end{equation*}
\end{proof}

In \cite{Fellows2018}, for each $1 \leq \beta \leq 2$, the authors present a $\beta$-approximation algorithm for \pathvc{3} with running time $1.708^{\frac{3-\beta}{2}\cdot k }\cdot n^{\Oh(1)}$.
As can be visually (see \cref{fig:3pvc}) or numerically (see \cref{tab:3pvc}) verified,
our algorithm has a strictly better running time for all values of $1 < \beta \leq 2$ .

\begin{table}[h]
    \centering
    \begin{tabular}{|c|c|c|}
        \hline
	$\beta$ & Our Algorithm \eqref{eq:3pvc_runtime} & \cite{Fellows2018} \\
        \hline
	1.1 & 1.6345 & 1.6628 \\
        \hline
	1.2 & 1.5641 & 1.6189 \\
        \hline
        1.3 & 1.4968 & 1.5762 \\
        \hline
        1.4 & 1.4323 & 1.5345 \\
        \hline
        1.5 & 1.3707 & 1.4940 \\
        \hline
        1.6 & 1.3117 & 1.4545 \\
        \hline
        1.7 & 1.2311 & 1.416\\
        \hline
        1.8 & 1.1487 & 1.3787 \\
        \hline
        1.9 & 1.0718 & 1.3423\\
        \hline
    \end{tabular}
    \caption{Comparison of the base of exponents of different algorithms for \pathvc{3}.
 	   For each row with a $\beta$ value $b$, a value $d$ in the second or third column
	implies a $b$-approximation algorithm with running time $d^{k} \cdot n^{\Oh(1)}$.}    
    \label{tab:3pvc}
\end{table}

\section{Sampling with a Black Box}
\label{sec:algorithms}

In this section we present our main technique, Sampling with a Black Box, and prove \Cref{theorem:summary_sampling_step}.
The technique is designed using three main components, enabling a modular analysis of each part. We use the notion of $\procext{\delta}{p}{r}{T}$ to abstract the outcome of iteratively  executing a sampling step. We use this abstract notion in $\randext$ which combines the $(\delta,p,r,T)$-procedure together with the black box parameterized $\alpha$-approximation algorithm.  On  its own, $\randext$ only attains a $\beta$-approximate solution with a low probability. Our main algorithm, $\iterrandext$, executes $\randext$ multiple times to get  a $\beta$-approximate solution with a constant  probability. The defined algorithm  depends on a parameter $\delta$, for which we find the optimal value.

We  start with the formal definition of a  $\procext{\delta}{p}{r}{T}$. As already mentioned,
a procedure serves as an abstraction of iterative use of a sampling step.
It returns a vertex set $S \subseteq V(G)$
with certain properties related to the size of $S$ and the value of $\OPT(G\setminus S)$. 
\begin{definition}\label{definition:proc}
	Let $\Pi$ be a hypergraph property and $\mathcal{G}$ be a closed set of hypergraphs.	
	For all $\delta \geq 1$, $r \geq 0$, $0 < p \leq 1$ and $T\geq 0$, a $\procext{\delta}{p}{r}{T}$ for $(\mathcal{G},\Pi)$ is a polynomial time randomized algorithm that takes as input a hypergraph $G \in \mathcal{G}$, an integer $t \geq 0$ and returns a set $S \subseteq V(G)$ with the following properties:
	\begin{properties}%
	\item It holds that $\abs{S} \leq \delta \cdot t$.\label{prop:first}
	 
	\item \label{prop:second} Suppose that $\OPT_{\cG,\Pi}(G) \leq k$ for some $k \geq 0$. If $t \geq T$, then with probability at least $\frac{p^{t}}{(t+1)^{r}}$ it holds that $G \setminus S$ has a solution of size at most $\max\left( 0, k -t \right)$, i.e.

	\begin{equation*}
		\Pr\Bigl( \OPT_{\cG,\Pi}(G\setminus S) \leq \max\left( 0, k -t \right) \Bigr) \geq \frac{p^{t}}{(t+1)^{r}}.	
	\end{equation*}
		\end{properties}
	Additionally, we use the notation $\proc{\delta}{p}$ to refer to a $\procext{\delta}{p}{r}{T}$ for some constants $r,T \geq 0$. 
\end{definition}

Observe that if there is a $\proc{\delta}{p}$ for $\gpivd$, then
we can make use of it to obtain a $\delta$
approximation as follows.
Suppose $G$ has a solution of size $k$.
Then, by setting $t = k$
in \Cref{prop:second}, it holds that with probability at least
$\frac{p^{k}}{(k + 1)^{r}}$, $G\setminus S$ has a solution of size~$0$, i.e. $G\setminus S
\in \Pi$ . By our assumption, we can check in polynomial time whether a
graph belongs to the property $\Pi$. By \Cref{prop:first} it holds that $\abs{S}\leq \delta\cdot  k$. Additionally, we can repeat this
algorithm $p^{-k} \cdot n^{\Oh(1)}$ times to obtain a $\delta$-approximate solution constant
probability. We summarize these insights in the following observation. 

\begin{observation}\label{observation:proc_to_approx}
	If  there is a $\proc{\delta}{p}$ for  $\gpivd$, then there is  parameterized $\delta$-approximation algorithm 
	for $\gpivd$  with
	running time $(1 / p)^{k} \cdot n^{\Oh(1)}$.
\end{observation}

In the remainder of \cref{sec:algorithms}, we fix the values of $0 < q \leq 1$, $1 \leq \alpha \leq \frac{1}{q}$, $1 \leq \beta \leq \frac{1}{q}$ and
$1 \leq c \leq \exp\left( \alpha \cdot \D{\frac{1}{\alpha}}{q} \right)$ to specific numbers,
unless specified explicitly.
For notational simplicity, we omit $\alpha$, $\beta$, and $c$ from the subscript of functions dependent on these variables.
Moreover, let $\Pi$ be a fixed polynomial-time decidable hypergraph property, $\mathcal{G}$ be a closed set of hypergraphs
and $\A$ be an $\alpha$-approximation algorithm for the $\gpivd$ problem with running time $c^{k} \cdot n^{\Oh(1)}$.
Let us also define the set of values $\delta$ can take, given $\alpha$ and $\beta$.
\begin{definition}\label{definition:good_set}
	For $\alpha, \beta \geq 1$ such that $\alpha \neq \beta$, we define the set $\goodd$ as
	\begin{equation*}
		\goodd \coloneqq \begin{cases}
			[ \beta, \infty ) &\text{if } \beta > \alpha\\
			[1,\beta] &\text{if } \beta < \alpha.\\			
		\end{cases}
	\end{equation*}
\end{definition}

The next component in our technique is $\randext$, given in \Cref{algo:rand_and_extend}. The algorithm is configured with  a $\proc{\delta}{p}$ $\procalg{\delta}{p}$.  Given an hypergraph $G$
and an integer $t$ in the input, the algorithm  invokes the procedure $\procalg{\delta}{p}$ which returns a random set
$S \subseteq V(G)$ of size at most $\delta \cdot t$,
and then runs the parameterized $\alpha$-approximation algorithm $\A$  on the remaining hypergraph $G \setminus S$.
The idea is to hope that $G \setminus S$ has a solution of size at most $\frac{\beta \cdot k - \delta \cdot t}{\alpha}$.
Note that the parameter for the $\alpha$-approximation algorithm is also $\frac{\beta \cdot k - \delta \cdot t}{\alpha}$,
which ensures that the approximation algorithm returns a set of size at most $\beta \cdot k - \delta \cdot t$,
with high probability.
In the event that this holds, by adding $S$ to the returned set, we obtain a solution with a size of at most $\beta \cdot  k$.

\begin{algorithm}
	\begin{algorithmic}[1]
		\Configuration  $0 < p \leq 1$, $\delta \in \goodd$ and a $\procext{\delta}{p}{r}{T}$ $\procalg{\delta}{p}$ for $(\mathcal{G}, \Pi)$.		
		\Input Hypergraph $G \in \mathcal{G}$, integers $0 \leq k \leq \abs{V(G)}$ and $T \leq t \leq  \frac{\beta}{\delta} \cdot k $ 
			\State $S = \procalg{\delta}{p}(G,t)$
			\State $Y = \A\left( G \setminus S, \frac{\beta \cdot k - \delta \cdot t}{\alpha} \right)$
			\State Return $S \cup Y$
	\end{algorithmic}
	\caption{$\randext$}
	\label{algo:rand_and_extend}
\end{algorithm}

Our main algorithm, given in \Cref{algo:iterative_rand_and_extend}, begins by selecting a value for $t^*$.
This value ensures the following: if the set $S$ in \cref{algo:rand_and_extend} contains
at least $t^*$ many elements from a solution, then $G \setminus S$ has a solution of size at
most $\frac{\beta \cdot k - \delta \cdot t^*}{\alpha}$. Note that this further implies
that the set returned by \cref{algo:rand_and_extend} has size at most $\beta \cdot k$.
Furthermore, \cref{algo:iterative_rand_and_extend} 
utilizes \cref{algo:rand_and_extend} and
executes it multiple times to ensure a $\beta$-approximate solution is attained with a constant probability.
\begin{algorithm}
	\begin{algorithmic}[1]
		\Configuration  $0 < p \leq 1$, $\delta \in \goodd$ and a $\procext{\delta}{p}{r}{T}$ $\procalg{\delta}{p}$ for $(\mathcal{G}, \Pi)$.				
		\Input Hypergraph $G \in \mathcal{G}$, integer $0 \leq k \leq \abs{V(G)}$
		\State $t^* \coloneqq \left\lceil  \frac{\beta - \alpha}{\delta  -\alpha}  \cdot k  \right\rceil $ if $\beta<\alpha$, and $t^* \coloneqq \left\lfloor \frac{\beta - \alpha}{\delta  -\alpha}  \cdot k \right\rfloor $ if $\beta>\alpha$ \label{iterative:tstar}
			
			\If{$t^* < T$}
			\State Return $W$ if and only if there exists $W \subseteq V(G)$ of size at most $k$ such that $W \in \sat_\Pi(G)$\label{line:W_at_most_k}
			\Else
			\State $\mathcal{S} = \emptyset$
			\For{$2 \cdot p^{-t^*} \cdot (t^* + 1)^r$ times}\label{line:for_loop_in_iter}
			\State $\mathcal{S} = \mathcal{S} \cup \Bigl\{\randext(G, k, t^*)\Bigr\} $
			\EndFor
			\State Return a minimum sized set in $\mathcal{S}$
			\EndIf
		\end{algorithmic}
		\caption{$\iterrandext$}
		\label{algo:iterative_rand_and_extend}
	\end{algorithm}

\begin{lemma}\label{lemma:beta_approx}
	Let $0 < p \leq 1$, $\delta \in \goodd$ and a $\procext{\delta}{p}{r}{T}$ $\procalg{\delta}{p}$ for $(\mathcal{G}, \Pi)$.
	Then \cref{algo:iterative_rand_and_extend} is a randomized parameterized $\beta$-approximation algorithm for $\sgpivd$ with running time
	\begin{equation*}
		f\left( \delta, p \right) ^{k} \cdot n^{\Oh(1)}
	\end{equation*}
	where $f(\delta,p)$ is given by 
	\begin{equation*}
		f(\delta, p) \coloneqq  \exp\left( \frac{(\delta - \beta) \cdot \ln(c) + (\beta - \alpha) \cdot \ln\left( \frac{1}{p} \right) }{\delta - \alpha} \right).
	\end{equation*}
\end{lemma}

The proof of \cref{lemma:beta_approx} can be found in \cref{sec:proc_to_approx}.

\Cref{algo:iterative_rand_and_extend} relies on the existence of a $\proc{\delta}{p}$, however, insofar we did not show how to design one. 
We generate a $\proc{\delta}{p}$ from a sampling step (\Cref{definition:sampling_step}) via a simple algorithm which iteratively invokes the sampling step. The pseudo-code of  the algorithm is given in \Cref{algorithm:multi_sample}.

\begin{algorithm}
	\begin{algorithmic}[1]
		\Configuration A number $\delta \geq 1$ and a sampling step $\mathcal{R}$ for $\sgpivd$ with success probability $q$ for some $0 < q \leq 1$.		
		\Input Hypergraph $G \in \mathcal{G}$, integer $t \geq 0$ 
		\State $S \gets \emptyset$		
		\While{$\abs{S} < \delta \cdot t$ and $G \not\in \Pi$}
		\State $ v = \mathcal{R}\left(G\right) $\label{line:v_multi_sample}
		\State $G = G \setminus \{v\}$
		\State $S = S \cup \{v\} $		
		\EndWhile
		\State \Return $S$
	\end{algorithmic}
	\caption{MultiSample}\label{algorithm:multi_sample}
\end{algorithm}

Define
\begin{equation}\label{eq:def_p_delta_q}
	\phi(\delta, q) \coloneqq \exp\left( -\delta \cdot \D{\frac{1}{\delta}}{q} \right).
\end{equation}
The following lemma state that \Cref{algorithm:multi_sample} is indeed a $\proc{\delta}{p}$ for $p = \phi(\delta,q)$. 
\begin{lemma}\label{lemma:multi_step_proc}
	Let $0 < q \leq 1$ and $\mathcal{R}$ be a sampling step for $\gpivd$ with success
	probability $q$. Then, for any $1 \leq \delta \leq \frac{1}{q}$, \cref{algorithm:multi_sample}
	is a $\proc{\delta}{\phi\left( \delta, q \right) }$ for $\sgpivd$.
\end{lemma}

The proof of \cref{lemma:multi_step_proc} can be found in \cref{sec:sampling_to_proc}.
\cref{lemma:multi_step_proc} implies that for $\delta = \frac{1}{q}$,
 there exists a
$\proc{\frac{1}{q}}{1}$ for $\gpivd$.
Observe that this serves as a $\proc{\delta}{1}$
for $\delta > \frac{1}{q}$.
Therefore, we make the following observation.

\begin{observation}\label{observation:delta_1_proc}
	Let $0 < q \leq 1$ and $\mathcal{R}$ be a sampling step for $\gpivd$ with success
	probability $q$. Then, for any $\delta > \frac{1}{q}$,  there is a
	$\proc{\delta}{1}$ for $\gpivd$.
\end{observation}

Furthermore, we can now prove \Cref{lem:sampling_to_approximation} using   \Cref{lemma:multi_step_proc} together with  \Cref{observation:proc_to_approx}.
\begin{proof}[Proof of \cref{lem:sampling_to_approximation}]
	Assume there is a sampling step with success probability $q$ for $\sgpivd$, and let $1\leq \beta\leq \frac{1}{q}$. Then by \Cref{lemma:multi_step_proc} there is a $\proc{\beta}{\phi(\beta,q)}$ for $\sgpivd$. Hence, by \Cref{observation:proc_to_approx} there is a parameterized $\beta$-approximation for $\sgpivd$ which runs in time $d^k\cdot n^{\Oh(1)}$ where 
	$$
	d=\frac{1}{\phi(\beta,q)} = \exp\left(\beta\cdot \D{\frac{1}{\beta}}{q} \right).
	$$
\end{proof}

Using  the procedure from \Cref{lemma:multi_step_proc} together with \Cref{lemma:beta_approx} (\Cref{algo:iterative_rand_and_extend}) we get the following results. 
\begin{lemma}\label{lemma:approx_sampling}
	Suppose that:
	\begin{itemize}
		\item There is a sampling step with success probability $q\in(0,1)$ for $\sgpivd$.
		\item There is a randomized parameterized $\alpha$-approximation algorithm for $\sgpivd$ which runs in time $c^{k}\cdot n^{\Oh(1)}$. 
	\end{itemize}
	Then,  there is a randomized parameterized
	$\beta$-approximation algorithm for $\sgpivd$ with running time
	\begin{equation*}
		\biggl(\min_{\delta \in \goodd \cap [1, \frac{1}{q}]}\tilde{f}\,(\delta, q)\biggr)^{k} \cdot n^{\Oh(1)}
	\end{equation*}
	where $\tilde{f}(\delta ,q)$ is defined as
\begin{equation*}
	\tilde{f}\,(\delta, q) \coloneqq  f(\delta,\phi(\delta,q) )=\exp\left( \frac{(\delta - \beta) \cdot \ln(c) + (\beta - \alpha) \cdot \ln\left( \frac{1}{\phi(\delta, q)} \right)  }{\delta - \alpha} \right).
\end{equation*}
\end{lemma}

\cref{lemma:approx_sampling} provides a $\beta$ approximation algorithms whose running time is a solution for an optimization problem. 
The final step towards the proof of \Cref{theorem:summary_sampling_step} is to solve this optimization problem.
\begin{lemma}\label{lemma:min_f_convert_equiv}
	It holds that
	\begin{equation*}
		\min_{\delta \in \goodd \cap [1, \frac{1}{q}]}\tilde{f}\,(\delta, q) = \runtime.
	\end{equation*}
\end{lemma}

The proofs of \cref{lemma:approx_sampling,lemma:min_f_convert_equiv} can be found in \cref{sec:optim_delta}.
We now have everything to proceed with the proof of \cref{theorem:summary_sampling_step}.

\begin{proof}[Proof of \cref{theorem:summary_sampling_step}]
	The claim follows immediately from \cref{lemma:approx_sampling,lemma:min_f_convert_equiv}.
\end{proof}

\subsection{Converting Procedures to Approximation Algorithms}
\label{sec:proc_to_approx}
In this section we will prove \cref{lemma:beta_approx}.
To accomplish this, we will examine some properties of \cref{algo:rand_and_extend}.
First, we will show that \cref{algo:rand_and_extend} indeed returns a solution $W \in \sat_{\Pi}(G)$.
Subsequently, we will establish that with a certain probability, the set returned has size at most $\beta \cdot k$.
Equipped with these results, we will proceed with the proof of \cref{lemma:beta_approx}.

\begin{lemma}\label{lem:randext_ret_yes}
	Let $G \in \mathcal{G}$ be a hypergraph, and $0 \leq k \leq n$, $T \leq t \leq \frac{\beta}{\delta} \cdot k$ be integers.
	Let $W$ denote the set returned by $\randext(G,k,t)$, then it holds that $W \in \sat_\Pi(G)$.
\end{lemma}

\begin{proof}
	Let $S$ be as defined in \cref{algo:rand_and_extend}.
	Since $\A$ is a parameterized $\alpha$-approximation
	algorithm, by \Cref{definition:alpha_approx}, $Y$ is a solution of $G \setminus S$,
	i.e. $Y \in \sat_{\Pi}(G \setminus S)$.
	Equivalently, it holds that $\left( G \setminus S \right) \setminus Y \in \Pi$.
	Since $G \setminus \left( Y \cup S \right) = \Bigl(\left( G \setminus S \right) \setminus Y\Bigr) \in \Pi$, it also holds
	that $W = (Y \cup S) \in \sat_{\Pi}(G)$.
\end{proof}

\begin{lemma}\label{lem:tstar_bound}
	Let $t^*$ be as defined in \cref{algo:iterative_rand_and_extend}, then it holds that
	\begin{equation*}
		\max\biggl(0,\,\frac{\beta - \alpha}{\delta - \alpha} \cdot k - 1\biggr) \leq t^* \leq \frac{\beta - \alpha}{\delta - \alpha}\cdot k + 1.
	\end{equation*}
	Moreover, it also holds that $t^* \leq \frac{\beta}{\delta} \cdot k$.
\end{lemma}

\begin{proof}
	First, notice that for $\delta \in \goodd$, it holds that either $\delta > \beta > \alpha$ or $\delta < \beta < \alpha$.
	Hence, the term $\frac{\beta - \alpha}{\delta - \alpha}$ is always positive. By definition of $t^*$, we also have
	\begin{equation*}
		   \max\biggl(0,\,\frac{\beta - \alpha}{\delta - \alpha} \cdot k - 1\biggr) \leq \left\lfloor   \frac{\beta - \alpha}{\delta  -\alpha}  \cdot k   \right\rfloor \leq t^* \leq \left\lceil  \frac{\beta - \alpha}{\delta  -\alpha}  \cdot k  \right\rceil <    \frac{\beta - \alpha}{\delta  -\alpha}  \cdot k  + 1.
	\end{equation*}
	If $\beta > \alpha$, then we have
	\begin{equation*}
		t^* = \left\lfloor    \frac{\beta - \alpha}{\delta  -\alpha}  \cdot k    \right\rfloor \leq  \frac{\beta - \alpha}{\delta  -\alpha}  \cdot k  < \frac{\beta}{\delta} \cdot k
	\end{equation*}
	where the last inequality is true because $\alpha < \beta < \delta$. Similarly,
	if $\beta < \alpha$, then $\delta < \beta$ as well and we have
	\begin{equation*}
		t^* = \left\lceil     \frac{\alpha - \beta}{\alpha - \delta}  \cdot k  \right\rceil < \left\lceil \frac{\alpha - \beta}{\alpha - \beta}  \cdot k \right\rceil = k < \frac{\beta}{\delta} \cdot k.
	\end{equation*}
\end{proof}

\begin{lemma}\label{lem:randext_exists}
	Let $G \in \mathcal{G}$ be a hypergraph, $0 \leq k \leq n$ be an integer and let $t^*$ be as defined in \cref{algo:iterative_rand_and_extend} such that $t^* \geq T$.
	Moreover, let $Z$ be the set returned by $\randext(G,k,t^*)$.
	If $\OPT_{\cG,\Pi}(G) \leq k$, then $\abs{Z} \leq \beta \cdot k$ with probability at least $\frac{p^{t^*}}{2 \cdot (t^* + 1)^{r}}$.
\end{lemma}

\begin{proof}
	Suppose $\OPT_{\cG,\Pi}(G) \leq k$.
	Observe that \cref{lem:tstar_bound} implies $0 \leq t^* \leq \frac{\beta}{\delta} \cdot  k$.
	We have
	\begin{align}
		k - t^* &= \frac{\beta}{\alpha}\cdot k + \left( 1 - \frac{\beta}{\alpha} \right) \cdot k -\frac{\delta \cdot t^*}{\alpha} - \left( t^* - \frac{\delta \cdot t^*}{\alpha} \right)\nonumber\\
				     &= \frac{\beta \cdot k - \delta\cdot t^*}{\alpha} + \left( 1 - \frac{\beta}{\alpha} \right) \cdot k + t^* \cdot \left(\frac{ \delta - \alpha }{\alpha}\right)\label{eq:kminustdelta}.
	\end{align}

	\begin{claim}\label{claim:simple_ineq}
		It holds that $t^* \cdot \frac{\delta - \alpha}{\alpha}  \leq -\left( 1 - \frac{\beta}{\alpha} \right)\cdot k$.
	\end{claim}

	\begin{claimproof}
		Recall that for $(\delta) \in \goodd$, it holds that either $\delta > \beta > \alpha$ or $\delta < \beta < \alpha$.
		
		If $\beta > \alpha$, we have that $\frac{\delta - \alpha}{\alpha} > 0$ and $t^* = \left\lfloor \frac{\beta - \alpha}{\delta  -\alpha}  \cdot k \right\rfloor \leq  \frac{\beta - \alpha}{\delta  -\alpha}  \cdot k$. Hence,
		\begin{equation*}
			t^* \cdot \frac{\delta - \alpha}{\alpha} \leq \frac{\beta - \alpha}{\delta  -\alpha}  \cdot k \cdot \frac{\delta - \alpha}{\alpha} = \frac{\beta - \alpha}{\alpha} \cdot k = -\left( 1 - \frac{\beta}{\alpha} \right)\cdot k .
		\end{equation*}
		Similarly, if $\beta < \alpha$, then we have that $\frac{\delta - \alpha}{\alpha} < 0$ and $t^* = \left\lceil  \frac{\beta - \alpha}{\delta  -\alpha}  \cdot k  \right\rceil \geq  \frac{\beta - \alpha}{\delta  -\alpha}  \cdot k$. Therefore, we also get
		\begin{equation*}
			t^* \cdot \frac{\delta - \alpha}{\alpha} \leq \frac{\beta - \alpha}{\delta  -\alpha}  \cdot k \cdot \frac{\delta - \alpha}{\alpha} = \frac{\beta - \alpha}{\alpha} \cdot k = -\left( 1 - \frac{\beta}{\alpha} \right)\cdot k.
		\end{equation*}
	\end{claimproof}

	\cref{claim:simple_ineq} implies that
\begin{align*}
	k - t^* \,&= \frac{\beta \cdot k - \delta\cdot t^*}{\alpha} + \left( 1 - \frac{\beta}{\alpha} \right) \cdot k + t^* \cdot \left(\frac{ \delta - \alpha }{\alpha}\right)\label{eq:kminustdelta}\\
	&\leq \,\frac{\beta \cdot k - \delta \cdot t^*}{\alpha} + \left( 1 - \frac{\beta}{\alpha} \right) \cdot k - \left( 1 - \frac{\beta}{\alpha} \right) \cdot k\\	
	&= \frac{\beta \cdot k - \delta \cdot t^*}{\alpha}.	
\end{align*}

Finally, since $t^* \leq \frac{\beta}{\delta}\cdot k$, it also holds that $0 \leq \frac{\beta \cdot k - \delta \cdot t^*}{\alpha}$.
Therefore we get
\begin{equation}\label{eq:kminust_ub}
	\max\left( 0, k - t^* \right) \leq \frac{\beta \cdot k - \delta \cdot t^*}{\alpha}.
\end{equation}

	Now let $S$ be as defined in \cref{algo:rand_and_extend}.
	According to \Cref{definition:proc} and \eqref{eq:kminust_ub}, with probability at least $\frac{p^{t^*}}{(t^* + 1)^{r}}$, it holds that $G \setminus S$ has a solution of size at most $\max\left( 0, k - t^* \right) \leq \frac{\beta \cdot k - \delta \cdot t^*}{\alpha}$.
	Note that in this scenario, the set returned by $\A(G \setminus S, \frac{\beta \cdot k - \delta \cdot t^*}{\alpha})$, i.e. $Y$, satisfies
	\begin{equation}\label{eq:Y_ineq_prob}
		\abs{Y} \leq \alpha \cdot \frac{\beta \cdot k - \delta \cdot t^*}{\alpha} = \beta \cdot k - \delta \cdot t^*
	\end{equation}
	with probability at least $\frac{1}{2}$ by \cref{definition:alpha_approx}.
	Moreover, since $Z = Y \cup S$ and $\abs{S}  \leq \delta \cdot t$ it holds that
	\begin{equation}\label{eq:Y_ineq_implies}
		\abs{Y} \leq \beta \cdot k - \delta \cdot t^* \implies \abs{Z} \leq  \beta \cdot k.
	\end{equation}
	Therefore we have
	\begin{align*}
		\Pr\left( \abs{Z} \leq \beta \cdot k \right) &\geq \Pr\left( \abs{Y} \leq \beta \cdot k - \delta \cdot t^* \right)\\
							     &\geq \Pr\left(  \abs{Y} \leq \beta \cdot k - \delta \cdot t^*  \mid   G \setminus S \text{ has a solution of size at most} \max\left( 0, k - t^* \right) \right)\\
									&\quad \cdot \Pr\Bigl( G \setminus S \text{ has a solution of size at most} \max\left( 0, k - t^* \right) \Bigr) \\
									&\geq \frac{1}{2} \cdot \frac{p^{t^*}}{(t^* + 1)^{r}}
	\end{align*}
	where the first inequality follows from \eqref{eq:Y_ineq_implies} and the last one follows from \eqref{eq:Y_ineq_prob}.
\end{proof}

Now we are ready to prove \cref{lemma:beta_approx}.

\begin{proof}[Proof of \cref{lemma:beta_approx}]
	In order to show that \cref{algo:iterative_rand_and_extend} is a
	randomized parameterized $\beta$-approximation algorithm for $\gpivd$
	problem, we need to prove that it satisfies
	\cref{definition:alpha_approx}.
	To that end, let $S$ denote the set returned by \cref{algo:iterative_rand_and_extend}
	and assume that $\OPT_{\cG,\Pi}(G) \leq k$.
	
	If $t^* < T$, by \cref{line:W_at_most_k} in \cref{algo:iterative_rand_and_extend},
	the set $S$ returned by the algorithm
	satisfies $S \in \sat_{\Pi}(G)$ and $\abs{S} \leq k < \beta \cdot k$ with probability 1.
	In this case \cref{algo:iterative_rand_and_extend} satisfies \cref{definition:alpha_approx}.
	
	If $t^* \geq T$, consider an iteration $i$ of the for loop in \cref{line:for_loop_in_iter}
	in \cref{algo:iterative_rand_and_extend}, for $1 \leq i \leq 2 \cdot p^{-t^*} \cdot (t^* + 1)^r$.
	Let $Y_i$ denote the set returned by $\randext(G, k, t^*)$ in iteration $i$.
	Recall that $t^* \leq \frac{\beta}{\delta} \cdot k$, therefore
	\cref{lem:randext_ret_yes} implies that $Y_i \in \sat_{\Pi}(G)$.

	On the other hand, by \cref{lem:randext_exists}, 
	it holds that 
	$\abs{Y_i} \leq \beta \cdot k$ with probability at least
	$\frac{p^{t^*}}{2 \cdot (t^* + 1)^{r}}$. Therefore we get
	\begin{align}
		\Pr\left(\abs{S} > \beta \cdot k\right) &= \Pr\left( \abs{Y_i} > \beta \cdot k \text{ for } 1 \leq i \leq 2 \cdot p^{-t^*} \cdot (t^* + 1)^r \right) \notag \\
									 &= \left( 1 - \frac{p^{t^*}}{2 \cdot (t^* + 1)^{r}} \right) ^{2 \cdot p^{-t^*} \cdot (t^* + 1)^{r}} \notag \\
									 &\leq e^{-1} \label{eq:pr_S_exceed_betak}
	\end{align}
	where the second equality holds since each iteration of the for loop is independent.
	Finally, \eqref{eq:pr_S_exceed_betak} implies that

	\begin{equation*}
		\Pr\left( \abs{S} \leq \beta \cdot k \right) \geq \left(1 - \frac{1}{e}\right) \geq \frac{1}{2}.
	\end{equation*}

	Now let us consider the running time of the algorithm.

	\begin{claim}\label{claim:running_time_iter_randandextend}
		The running time of the algorithm is $c^{\frac{\beta \cdot k - \delta \cdot t^*}{\alpha}} \cdot p^{-t^*}  \cdot n^{\Oh(1)}$.
	\end{claim}

	\begin{claimproof}
		Observe that if $t^* \geq T$, each execution of $\randext$ takes time $c^{\frac{\beta \cdot k - \delta \cdot t^*}{\alpha}} \cdot n^{\Oh(1)}$.
		This is because the algorithm executes the $\procext{\delta}{p}{r}{T}$ $\procalg{\delta}{p}$ and the parameterized approximation algorithm, where the former has polynomial running time and the latter has a running time of $c^{\frac{\beta \cdot k - \delta \cdot t^*}{\alpha}} \cdot n^{\Oh(1)}$.
		Since the number of iterations is $2 \cdot p^{-t^*} \cdot (t^* + 1)^{r}$, the total running time becomes
		\begin{align*}
			c^{\frac{\beta \cdot k - \delta\cdot t^*}{\alpha}} \cdot p^{-t^*} \cdot (t^* + 1)^{r} \cdot n^{\Oh(1)} = c^{\frac{\beta \cdot k - \delta \cdot t^*}{\alpha}} \cdot p^{-t^*}  \cdot n^{\Oh(1)}
		\end{align*}
		since $t^* \leq \left\lceil  \frac{\beta - \alpha}{\delta - \alpha} \cdot k \right\rceil  = \Oh(k) = \Oh(n)$.

		Now suppose $t^* < T$ and observe that $k = \Oh(1)$ because $t^* < T = \Oh(1)$ and $k \leq \left( t^* + 1 \right)\cdot \frac{\delta - \alpha}{\beta - \alpha} = \Oh(1)$. Since the algorithm goes over all sets $W \subseteq V(G)$ of size at most $k$, the running time is at most $n^{\Oh(k)} = n^{\Oh(1)}$. Therefore we can conclude that the running time of the algorithm is upper bounded by $c^{\frac{\beta \cdot k - \delta \cdot t^*}{\alpha}} \cdot p^{-t^*}  \cdot n^{\Oh(1)}$.	
	\end{claimproof}

	Since $t^* \geq \frac{\beta - \alpha}{\delta - \alpha} \cdot k  - 1$, by \cref{lem:tstar_bound}, we have
	\begin{align}
		\frac{\beta \cdot k - \delta \cdot t^*}{\alpha} &\leq \frac{\beta \cdot k - \delta \cdot \biggl(\frac{\beta - \alpha}{\delta - \alpha} \cdot k - 1\biggr)}{\alpha} \nonumber \\ 
						   &= \frac{\beta \cdot k}{\alpha} - \frac{(\beta - \alpha) \cdot \delta \cdot k}{(\delta - \alpha) \cdot \alpha} + \frac{\delta}{\alpha}.\label{eq:bkdt_ub}
	\end{align}			

	Therefore it holds that
	\begin{align}
		c^{\frac{\beta \cdot k - \delta \cdot t^*}{\alpha}} \cdot p^{-t^*} &= \exp\left( \left( \frac{\beta \cdot k - \delta \cdot t^*}{\alpha} \right)  \cdot \ln(c) + t^* \cdot \ln\left( \frac{1}{p} \right)  \right) \notag \\
		&\leq \exp\left( \left( \frac{\beta \cdot k}{\alpha} - \frac{(\beta - \alpha) \cdot \delta \cdot k}{(\delta - \alpha) \cdot \alpha}  + \frac{\delta}{\alpha}\right)  \cdot \ln(c) + \left( \frac{\beta - \alpha}{\delta - \alpha}  \cdot k + 1 \right)  \cdot \ln\left( \frac{1}{p} \right) \right)  \notag \\
		&= \exp\left( \left(   \frac{\beta \cdot k \cdot (\delta - \alpha) - (\beta - \alpha) \cdot \delta \cdot k}{(\delta - \alpha) \cdot \alpha} \right)  \cdot \ln(c) + \left( \frac{\beta - \alpha}{\delta - \alpha} \cdot k \right)  \cdot \ln\left( \frac{1}{p} \right)  \right) \cdot c^{\frac{\delta}{\alpha}} \cdot \frac{1}{p}\notag \\		
		&= \exp\left( \left(   \frac{\alpha \cdot k \cdot (\delta - \beta)}{\alpha \cdot (\delta - \alpha)}  \right)  \cdot \ln(c) + \left( \frac{\beta - \alpha}{\delta - \alpha} \cdot k \right)  \cdot \ln\left( \frac{1}{p} \right)  \right)\cdot c^{\frac{\delta}{\alpha}} \cdot \frac{1}{p}\notag \\				
								  &= \exp\left( \frac{\delta - \beta}{\delta - \alpha} \cdot k\cdot \ln(c) + \frac{\beta - \alpha}{\delta - \alpha}  \cdot k \cdot \ln\left( \frac{1}{p} \right)  \right)\cdot c^{\frac{\delta}{\alpha}} \cdot \frac{1}{p}\notag \\
								  &= \exp\left( \frac{(\delta - \beta) \cdot \ln(c) + (\beta - \alpha) \cdot \ln\left( \frac{1}{p} \right) }{\delta - \alpha} \right)^{k} \cdot c^{\frac{\delta}{\alpha}} \cdot \frac{1}{p},\label{eq:rt_formula_bound}
	\end{align}
	where the inequality follows from \eqref{eq:bkdt_ub} and \cref{lem:tstar_bound}.
	Therefore, by \cref{claim:running_time_iter_randandextend} and \eqref{eq:rt_formula_bound}, the running time of the algorithm is
	\begin{align*}
		f(\delta, p)^{k} \cdot n^{\Oh(1)}.
	\end{align*}
\end{proof}

\subsection{Converting Sampling Steps to Procedures}
\label{sec:sampling_to_proc}

In this section, we will prove \cref{lemma:multi_step_proc} by developing
several auxiliary lemmas. Let $0 < q \leq 1$, $1 \leq \delta \leq \frac{1}{q}$ and $\mathcal{R}$
be a sampling step  for $\gpivd$, with success probability $q$.
Consider \cref{algorithm:multi_sample} with these parameters.
We will demonstrate that there exists integers $r$ and $T$
such that \cref{algorithm:multi_sample} is a $\procext{\delta}{\phi\left( \delta, q \right) }{r}{T}$.
To that end, we will need to show that given a hypergraph $G \in
\mathcal{G}$ and $t \geq 0$ as input,
\cref{algorithm:multi_sample} runs in polynomial time and satisfies 
\cref{prop:first,prop:second}, as defined in \cref{definition:proc}.

Note that neither the running time of the algorithm nor \cref{prop:first}
depend on the values of $r$ and $T$. Therefore, irrespective of the values of $r$
and $T$, we will show that \cref{algorithm:multi_sample} runs in
polynomial time and that \cref{prop:first} holds.
Then, we will show that there exists $r$ and $T$ for which \cref{prop:second} holds,
implying the truth of  \cref{lemma:multi_step_proc}.

\begin{lemma}\label{lemma:multi_step_proc_rt}
	\cref{algorithm:multi_sample} runs in polynomial time.
\end{lemma}

\begin{proof}
	Since $\mathcal{R}$ is a sampling step for $\gpivd$, it runs in polynomial time.
	Moreover, after $n$ steps, $\overline{G}$ becomes empty and therefore belongs
	to $\Pi$, as $\Pi$ is hereditary and includes the empty graph.
	Therefore, the number of iterations of the while loop in \cref{algorithm:multi_sample} is at
	most $n$.
	Finally, membership to $\Pi$ can be tested in polynomial time since $\Pi$
	is polynomial-time decidable.
	Therefore, the whole algorithm runs in polynomial time.	
\end{proof}

\begin{lemma}\label{lemma:multi_step_proc_prop_1}
	\cref{algorithm:multi_sample} satisfies \cref{prop:first} in \cref{definition:proc}.
\end{lemma}

\begin{proof}
	Observe that the while loop in \cref{algorithm:multi_sample} runs at most
	$\delta \cdot t$ times. Moreover, in each iteration of the for loop, the
	size of $S$ increases by at most 1. Therefore the claim follows.
\end{proof}

Next, we demonstrate a simple feature of hereditary hypergraph properties.
Intuitively, removing a vertex from a hypergraph does not increase the size
of the optimal solution. The proof of \cref{lemma:hereditary_opt_decrease}
can be found in \cref{sec:omitted_proofs}.
\begin{restatable}{lemma}{heredoptdecrease}
	\label{lemma:hereditary_opt_decrease}
	Let $\Pi$ be a hereditary hypergraph property and $G$ be a hypergraph.
	For any $v \in V(G)$, it holds that
	\begin{equation*}
		 0 \leq \OPT_{\Pi}(G) - \OPT_{\Pi}(G \setminus v ) \leq 1.
	\end{equation*}
\end{restatable}

Let $\xi_1, \ldots, \xi_n$ be i.i.d. binary random variables and $\nu \in (0,1]$ such that $\Pr\left( \xi_i = 1 \right) \geq \nu$
for all $1 \leq i \leq n$.
The following inequality can be shown using standard arguments
\begin{equation*}
	\Pr\Biggl(\sum_{i = 1}^{n} \xi_i \geq t\Biggr) \geq \exp\Biggl(- \frac{n}{t} \cdot \D{\frac{t}{n}}{\nu}\Biggr) \cdot n^{\Oh(1)}.
\end{equation*}
In the following lemma, we prove a similar statement in our setting where the i.i.d. assumption is dropped. 
Its proof can be found in \cref{sec:prob_results}.

\begin{restatable}{lemma}{probmainresult}\label{lemma:prob_main_result}
	Let $\delta, t \geq 1$ be integers, $\nu \in (0,1]$ be a real number
	and $\xi_1, \ldots, \xi_{\left\lfloor \delta \cdot t \right\rfloor} \in
	\{0,1\}$ be random variables such that
	\begin{equation*}
		\Pr\left( \xi_j =1\,\middle|\, \xi_1 = x_1,\ldots,\xi_{j-1}
	= x_{j - 1}\right)\geq \nu
	\end{equation*}
	for all $1 \leq j \leq \left\lfloor \delta \cdot t \right\rfloor$ and $\left( x_1, \ldots,
x_{j-1} \right) \in \{0,1\}^{j-1}$. 
Then, there exist integers $r$ and $T$ that depend on
$\delta$, such that for $t \geq T$ it holds that
			\begin{equation*}
				\Pr\biggl( \sum_{j = 1}^{\floor{\delta \cdot t}} \xi_j \geq t\biggr) \geq \left( \delta\cdot t + 1 \right) ^{-r} \cdot \exp\biggl( -\delta \cdot \D{\frac{1}{\delta}}{\nu}\biggr)^{t}.
			\end{equation*}
\end{restatable}

Given a hypergraph $G \in
\mathcal{G}$ and $t \geq 0$ as input, let $\ell \leq \floor{\delta \cdot t}$ be
the number of iterations of the while loop in
\cref{algorithm:multi_sample}. Let $G_0 \coloneqq G$ and for $1 \leq i
\leq \ell$, let $G_i$ denote the hypergraph $G$ at the end
of the $i$'th iteration. Similarly, let $v_i$ denote the vertex $v$ at
the end of the $i$'th iteration, i.e. $v_i = \mathcal{R}\left( G_{i -1}
\right)$. For $\ell + 1 \leq i \leq \floor{\delta \cdot t}$, we let
$G_i \coloneqq G_{i - 1}$ and $v_i \coloneqq v_{i - 1}$.
Furthermore, we define the random variables $Z_0 \coloneqq 0$ and
\begin{equation*}
	Z_i = \begin{cases}
		\OPT\left( G_{i - 1} \right) - \OPT\left( G_i \right) &\text{if } 1 \leq i \leq \ell\\
		1 & \text{if } \ell < i \leq \floor{\delta \cdot t}
	\end{cases}
\end{equation*}
for $1 \leq i \leq \floor{\delta \cdot t}$.
Intuitively, for $1 \leq i \leq \ell$, $Z_i$ measures the decrease in
the optimal solution size, from $G_{i - 1}$ to $G_i$. Note that, by
definition, for $i \leq \ell$ it holds that $G_i = G_{i - 1} \setminus
\{v\}$ for some $v \in V\left( G_{i - 1} \right)$.
Therefore, we have $Z_i \in \{0,1\}$ by
\cref{lemma:hereditary_opt_decrease}.

\begin{lemma}\label{lemma:pr_Z_j_ind}
	For $1 \leq j \leq \floor{\delta \cdot t}$ and $(x_1, \ldots, x_{j -1}) \in \{0,1\}^{j-1}$ it holds that
	\begin{equation*}
		\Pr\Bigl( Z_j = 1 \mid Z_1 = x_1, \ldots, Z_{j- 1}  = x_{j- 1}\Bigr) \geq q.
	\end{equation*}
\end{lemma}

\begin{proof}
	By the law of total probability, it holds that
	\begin{align}
		\Pr\Bigl( Z_j = 1 \mid Z_1 &= x_1, \ldots, Z_{j- 1}  = x_{j- 1}\Bigr) = \notag\\
				  &\Pr\Bigl( Z_j = 1 \mid Z_1 = x_1, \ldots, Z_{j- 1}  = x_{j- 1}, j \leq \ell\Bigr) \cdot \Pr\left(j \leq \ell\right) + \notag\\
				  &\Pr\Bigl( Z_j = 1 \mid Z_1 = x_1, \ldots, Z_{j- 1}  = x_{j- 1}, j > \ell\Bigr) \cdot \Pr\left(j > \ell\right).\label{eq:Z_j_prob_1}
	\end{align}

	We have
	\begin{align}
		\Pr\Bigl( &Z_j = 1 \mid Z_1 = x_1, \ldots, Z_{j- 1}  = x_{j- 1},\, j \leq \ell\Bigr)  \geq q\label{eq:Z_j_prob_2}
	\end{align}
	because $j \leq \ell$ implies that $G_{j} = G_{j - 1} \setminus v$ where $v = \mathcal{R}\left( G_{j - 1} \right)$. Since $\mathcal{R}$ is a sampling step for $\Pi$ with success probability $q$, the inequality follows.
	Similarly, we have
	\begin{align}
		\Pr\Bigl( &Z_j = 1 \mid Z_1 = x_1, \ldots, Z_{j- 1}  = x_{j- 1}, j > \ell\Bigr) = 1,\label{eq:Z_j_prob_3}  				
	\end{align}
	which holds because $j > \ell$ implies that $Z_j = 1$ with probability 1, by definition.
	Therefore, by \eqref{eq:Z_j_prob_1}, \eqref{eq:Z_j_prob_2} and \eqref{eq:Z_j_prob_3}
	\begin{equation*}
		\Pr\Bigl( Z_j = 1 \mid Z_1 = x_1, \ldots, Z_{j- 1}  = x_{j- 1}\Bigr) \geq q.
	\end{equation*}		
\end{proof}

In the following lemma, we establish a lower bound on the probability
that the graph returned by the algorithm, i.e. $G_{\floor{\delta \cdot t}}$,
has a solution of size at most $\max\left( 0, k - t \right)$.
This lower bound is equal to the probability that the sum of $Z_i$
for $i$ from 1 to $\floor{\delta \cdot t}$ exceeds $t$.

\begin{lemma}\label{lemma:pr_G_deltat_soln_lb}
	It holds that
	\begin{equation*}
		\Pr\Bigl( G_{\floor{\delta \cdot t}} \text{ has a solution of size at most } \max\left( 0, k - t \right)  \Bigr) \geq \Pr\Biggl(\sum_{i =1}^{\floor{\delta \cdot t}} Z_i \geq t\Biggr).
	\end{equation*}
\end{lemma}

\begin{proof}
Let $A$ be the event that $G_{\floor{\delta \cdot t}}$ has a solution
of size at most $\max\left( 0, k - t
\right) $, and let $B$ be the event that $\sum_{i
=1}^{\floor{\delta \cdot t}} Z_i \geq t$. Finally, let
$C$ be the event that $\ell = \floor{\delta \cdot t}$.
In the following, we say that an event $X$
implies another event $Y$ if $X \subseteq Y$.

Let us write $B = \left( B \cap C \right) \cup \left( B \cap \overline{C} \right) $. 
It holds that $(B \cap \overline{C}) \subseteq \overline{C}\subseteq A$, because if
$\ell < \floor{\delta \cdot t}$, then $G_{\floor{\delta \cdot t}} = G_{\ell}  \in \Pi$ and therefore
$G_{\floor{\delta \cdot t}}$ has a solution of size 0. On the other hand,
$B \cap C$ implies that $\ell = \floor{\delta \cdot t}$ and
$\sum_{i =1}^{\floor{\delta \cdot t}} Z_i \geq t$.
In this case, for each $1 \leq i \leq \floor{\delta \cdot t}$,
it holds that $i \leq \floor{\delta \cdot t} = \ell$ and
$Z_i = \OPT\left( G_{i - 1} \right) - \OPT\left( G_i \right)$.
Moreover,
	\begin{align*}
		\Biggl(\sum_{i =1}^{\floor{\delta \cdot t}} Z_i\Biggr) &= \sum_{i =1}^{\floor{\delta \cdot t}} \OPT\left( G_{i - 1} \right) - \OPT\left( G_i \right)\\
								       &= \OPT\left( G_0 \right) - \OPT\left( G_{\floor{\delta \cdot t}} \right) 
	\end{align*}
	which implies that
	\begin{align*}
		\OPT\left( G_{\floor{\delta \cdot t}} \right) &= \OPT\left( G_0 \right) - \Biggl(\sum_{i =1}^{\floor{\delta \cdot t}} Z_i\Biggr)\\
							      &\leq k - \Biggl(\sum_{i =1}^{\floor{\delta \cdot t}} Z_i\Biggr)\\
							      &\leq k - t\\
							      &\leq \max\left( 0, k - t \right) .
	\end{align*}
	Therefore, $(B \cap C) \subseteq A$ and it holds that $B \subseteq A$. Hence we get
	\begin{equation*}
		\Pr\Bigl( G_{\floor{\delta \cdot t}} \text{ has a solution of size at most } \max\left( 0, k - t \right)  \Bigr) = \Pr(A) \geq  \Pr(B) = \Pr\Biggl(\sum_{i =1}^{\floor{\delta \cdot t}} Z_i \geq t\Biggr).
	\end{equation*}
\end{proof}

\begin{lemma}\label{lemma:multi_step_proc_prop_2}
	There exists $r,T \geq 0$ such that \cref{algorithm:multi_sample} satisfies \cref{prop:second} in \cref{definition:proc}.
\end{lemma}

\begin{proof}
	
Suppose that $\OPT_{\cG,\Pi}(G) \leq k$ for some $k \geq 0$ and let $S$
be the set returned by \cref{algorithm:multi_sample}.

In order to prove the \cref{prop:second} holds, we will show that there exists integers $r,T \geq 0$
such that for $t \geq T$, it holds that $(G \setminus S) = G_{\left\lfloor \delta \cdot t \right\rfloor}$ has a solution of size at
most $\max\left( 0, k -t \right)$, with probability at least
$\frac{(\phi(\delta,q))^{t}}{(t+1)^{r}}$. 
By \cref{lemma:pr_G_deltat_soln_lb} it holds that
\begin{equation}\label{eq:multi_step_intermediary}
	\Pr\Bigl( G_{\floor{\delta \cdot t}} \text{ has a solution of size at most } \max\left( 0, k - t \right)  \Bigr) \geq \Pr\Biggl(\sum_{i =1}^{\floor{\delta \cdot t}} Z_i \geq t\Biggr).	
\end{equation}
	In light of \eqref{eq:multi_step_intermediary}, our goal is to establish a lower bound for the probability that $\sum_{i =1}^{\floor{\delta \cdot t}} Z_i \geq t$.
	By \cref{lemma:prob_main_result,lemma:pr_Z_j_ind},
 there exist $r$ and $T$ that depend on $\delta$ such that for
	all $t \geq T$, we have 
		\begin{align}\label{eq:r_T_prob_ineq}
			 \Pr\Biggl(\sum_{i =1}^{\floor{\delta \cdot t}} Z_i \geq t\Biggr) \geq \left( \delta\cdot t + 1 \right)
		^{-r} \cdot \exp\biggl( -\delta \cdot
		\D{\frac{1}{\delta}}{q}\biggr)^{t}.
		\end{align}

		Finally, we let $T' = \max\left( \delta, T \right)$ and $r' = 2 \cdot r$.
		Then, for all $t \geq T'$, it holds that

		\begin{equation}\label{eq:delta_t_plus_1}
			\left( \delta\cdot t + 1 \right)^{-r} = \left( \delta\cdot t + 1 \right)^{-\frac{r'}{2}} \geq
			\left( \delta\cdot t + \delta \right)^{-\frac{r'}{2}} = \delta^{-\frac{r'}{2}} \cdot \left( t + 1 \right)^{-\frac{r'}{2}} \geq \left( t + 1 \right)^{-r'}
		\end{equation}
		where the last inequality holds because $t \geq T' \geq \delta - 1$. By \cref{lemma:pr_G_deltat_soln_lb},
		\eqref{eq:multi_step_intermediary}, \eqref{eq:r_T_prob_ineq} and \eqref{eq:delta_t_plus_1},
		for all $t \geq T'$ we have
		\begin{align*}
			\Pr\Bigl( G_{\floor{\delta \cdot t}} \text{ has a solution of size at most } \max\left( 0, k - t \right)  \Bigr) &\geq \Pr\Biggl(\sum_{i =1}^{\floor{\delta \cdot t}} Z_i \geq t\Biggr)\\
			&\geq \left( \delta\cdot t + 1 \right)^{-r} \cdot \exp\biggl( -\delta \cdot \D{\frac{1}{\delta}}{q}\biggr)^{t}\\
			&\geq \left( t + 1 \right)^{-r'} \cdot \exp\biggl( -\delta \cdot \D{\frac{1}{\delta}}{q}\biggr)^{t}\\			
			&= (t + 1)^{-r'} \cdot \Bigl( \phi(\delta, q)\Bigr)^{t},
		\end{align*}
		which implies that \cref{prop:second} holds for $T'$ and $r'$.
\end{proof}
Finally, the proof of \cref{lemma:multi_step_proc} simply follows from
\cref{lemma:multi_step_proc_rt,lemma:multi_step_proc_prop_1,lemma:multi_step_proc_prop_2}.

\subsection{Which procedure to choose? Optimizing $\delta$.}
\label{sec:optim_delta}
In this section, we prove \cref{lemma:approx_sampling,lemma:min_f_convert_equiv}.
We first give the proof of \cref{lemma:approx_sampling}.
Then, with the aim of proving \cref{lemma:min_f_convert_equiv}, we 
analyze functions that arise during our analysis of the running time of the
algorithm. 

Observe that \cref{lemma:beta_approx} and \cref{observation:delta_1_proc} together
give a randomized parameterized $\beta$-approximation algorithm for $\sgpivd$
whose running time depends on $\delta$. More specifically, for $\delta > \frac{1}{q}$,
this running is time equal to
\begin{equation*}
	\exp\left( \frac{(\delta - \beta) \cdot \ln(c) + (\beta - \alpha) \cdot \ln\left( 1 \right) }{\delta - \alpha} \right)^{k} \cdot n^{\Oh(1)} = c^{\frac{\delta - \beta}{\delta - \alpha} \cdot k} \cdot n^{\Oh(1)}.
\end{equation*}
Since $\frac{\delta - \beta}{\delta - \alpha} = 1 - \frac{\beta - \alpha}{\delta - \alpha}$ is increasing for $\delta \in \goodd$,
the running time also increases for $\delta > \frac{1}{q}$.
Therefore it doesn't make sense to consider values of
$\delta > \frac{1}{q}$.
This is why the range of $\delta$ is constrained to be less than or equal to $\frac{1}{q}$ in \cref{lemma:approx_sampling}.
\begin{proof}[Proof of \cref{lemma:approx_sampling}]
	For each $\delta \in \goodd$ such that $1 \leq \delta \leq \frac{1}{q}$, by \Cref{lemma:multi_step_proc} there is a $\proc{\delta}{\phi(\delta,q)}$ for $\sgpivd$. Therefore,  
	by \cref{lemma:beta_approx}, there is a randomized parameterized
	$\beta$-approximation algorithm for $\sgpivd$ with running time
	\begin{equation*}
		f(\delta, \phi(\delta,q))^k\cdot n^{\Oh(1)}=\tilde{f}\left(\delta ,q \right)^{k} \cdot n^{\Oh(1)}.
	\end{equation*}
	In particular, one can consider all possible $\delta \in \Bigl(\goodd \cap [1,\frac{1}{q}]\Bigr)$
	and choose the $\delta$ that minimizes the running time.
	Therefore, there is a randomized  parameterized $\beta$-approximation algorithm for $\sgpivd$ which runs in time $\left(\min_{\delta\in  \goodd \cap [1,\frac{1}{q}]} \tilde{f}(\delta,q)\right)^k \cdot n^{\Oh(1)}$. 
\end{proof}

Our next goal is to prove \cref{lemma:min_f_convert_equiv}.
For fixed $0 < q \leq 1$, let us define
\begin{align*}
	h_{q}(\delta) \coloneqq \ln\Bigl(\tilde{f}(\delta,q)\Bigr) &=  \frac{(\delta - \beta) \cdot \ln(c) + (\beta - \alpha) \cdot \ln\left( \frac{1}{\phi(\delta, q)} \right) }{\delta - \alpha}\\
	&=  \frac{\delta - \beta}{\delta - \alpha} \cdot \ln(c) + \frac{\beta - \alpha}{\delta - \alpha}\cdot \ln\left( \frac{1}{\phi(\delta,q)} \right)
\end{align*}
We omit the subscript and write $h(\delta)$ instead of $h_{q}(\delta)$ whenever the values are clear from the context.
Note that minimizing $h$ is equivalent to minimizing $\tilde{f}$ as $\ln$ is a monotone increasing function,
in other words
\begin{equation}\label{eq:f_h_equivalence}
	\min_{\delta \in \goodd \cap [1,\frac{1}{q}]} \tilde{f}(\delta,q) = \exp\biggl(\min_{\delta \in \goodd \cap [1,\frac{1}{q}]} h_{q}(\delta)\biggr).
\end{equation}
Also observe that for any fixed $\delta \in \goodd$, $h(\delta)$ is a convex combination of $\ln(c)$ and $\ln\left( \frac{1}{\phi(\delta,q)} \right)$.
To make use of this property, let us define the following function which is linear in the variable $x$
\begin{align}
	m_{\delta,q}(x) &\coloneqq \frac{\delta - x}{\delta - \alpha} \cdot \ln(c) + \frac{x - \alpha}{\delta - \alpha}\cdot \ln\left( \frac{1}{\phi(\delta,q)} \right)\notag\\ 
				 &= \ln(c) +  s_q(\delta) \cdot (x - \alpha) \label{eq:m_slope}
\end{align}
where
\begin{equation*}
	s_q(\delta) \coloneqq \frac{\ln\left( \frac{1}{\phi(\delta,q)} \right) - \ln(c)}{\delta - \alpha}.
\end{equation*}

Moreover, we have
\begin{equation}\label{eq:m_h_equiv}
	h_{q}(\delta) = m_{\delta,q}(\beta),
\end{equation}
and using this equivalence, the running time for an
$x$ approximation, for a fixed $\delta$, can be stated as $d^{k} \cdot n^{\Oh(1)}$
where $d = \exp\Bigl( m_{\delta,q}(x) \Bigr)$.
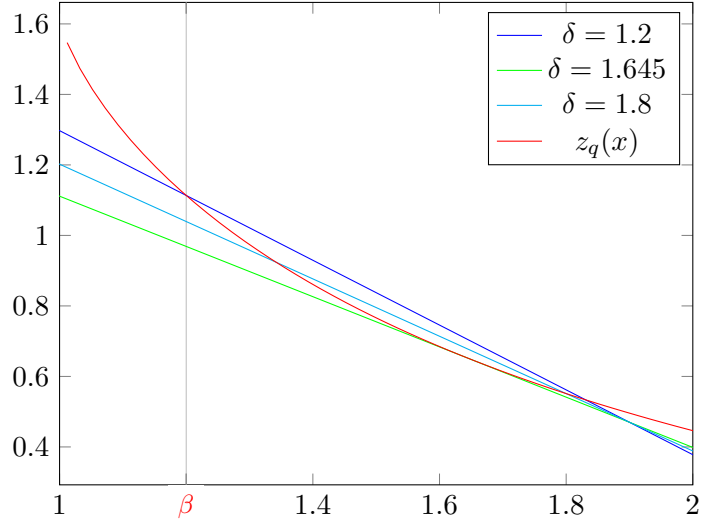
\begin{figure}[h!]
	\centering
	\begin{tikzpicture}

    \def\kl#1#2{(#1 * ln(#1 / #2) + (1 - #1) * ln((1 - #1) / (1 - #2)))}

     \begin{axis}[
         samples=100,                
         xmin=1, xmax=2,
         width=10cm,
         height=8cm,
         extra x ticks={ 1.2 },
         extra x tick labels={ $\beta$ },
         extra x tick style={grid=major, tick label style={font=\small, fill=white, text=red}}
     ] 

     \addplot[
         samples=100,
         blue
     ] {(1.2 - x)/(1.2 - 1.9)*ln(1.6)  + (x - 1.9)/(1.2 - 1.9)*1.2*\kl{ 1/1.2 }{ 0.2 }};
    \addlegendentry{$\delta = 1.2$}     

     \addplot[
         samples=100,
         green
     ] {(1.645 - x)/(1.645 - 1.9)*ln(1.6)  + (x - 1.9)/(1.645 - 1.9)*1.645*\kl{ 1/1.645 }{ 0.2 }};
    \addlegendentry{$\delta = 1.645$}     

     \addplot[
         samples=100,
         cyan
     ] {(1.8 - x)/(1.8 - 1.9)*ln(1.6)  + (x - 1.9)/(1.8 - 1.9)*1.8*\kl{ 1/1.8 }{ 0.2 }};
    \addlegendentry{$\delta = 1.8$}     

     \addplot[
         samples=500,
         red                    
     ] { x*\kl{ 1/x }{ 0.2 } };
     \addlegendentry{$z_q(x)$}
    
    \end{axis}    
\end{tikzpicture}		
	\caption{Comparison of the functions $m_{\delta, q}(x)$ with varying $\delta$ values (blue, green, and cyan), alongside the function $\ln\left( \frac{1}{\phi(x,q)}\right)$ (red). Recall that $m_{\delta, q}(x)$ is a linear function of $x$ and $m_{ \delta, q}(\beta) = h_{ q}(\delta)$. Also observe that the functions $\ln\left( \frac{1}{\phi(x,q)}\right)$ and $m_{ \delta, q}(x)$ meet at $\delta$.
}
	\label{fig:mplots}
\end{figure}

In the following we will demonstrate that for $0 < q < 1$, the value of $\delta$ that minimizes $h_{q}(\delta)$ (equivalently, $\tilde{f}(\delta, q)$) and can be found by analyzing $s_q(\delta)$.

\begin{lemma}\label{lemma:minim_h_slope}
	Let $0 < q < 1$. It holds that
	\begin{equation*}
		\min_{\delta \in \goodd \cap [1,\frac{1}{q}]} h_{q}(\delta) = \begin{cases}
		\ln(c) + (\beta - \alpha) \cdot \biggl(\min_{\delta \in \goodd \cap [1,\frac{1}{q}]} \;s_q(\delta)\biggr) &\text{if } \beta > \alpha \\
		\ln(c) + (\beta - \alpha) \cdot \biggl(\max_{\delta \in \goodd \cap [1,\frac{1}{q}]} \;s_q(\delta)\biggr) &\text{if } \beta < \alpha.	
	\end{cases}  
	\end{equation*}
\end{lemma}

\begin{proof}
	Observe that
	\begin{equation*}
		h_{q}(\delta) = \ln(c) + s_{q}(\delta) \cdot (\beta - \alpha)
	\end{equation*}
	which follows from \eqref{eq:m_slope} and \eqref{eq:m_h_equiv}.
	Assume that $\beta > \alpha$.
	Since $\alpha, \beta$ and $c$ are independent of $\delta$ and $(\beta -\alpha) > 0$,
	the value of $\delta$ that minimizes $h_{q}(\delta)$
	is the one for which the value of $s_q(\delta)$ is minimized.
	The case of $\beta < \alpha$ follows in the same way by noting that in this case
	$(\beta - \alpha) < 0$.
\end{proof}

\cref{lemma:minim_h_slope} states that depending on the values of $\alpha$ and $\beta$,
minimizing $h_{q}(\delta)$ is equivalent to either minimizing
or maximizing $s_q(\delta)$.
Therefore, in what follows, we will study the analytical properties of
$s_q(\delta)$.

\begin{restatable}{lemma}{sderivalternative}
	\label{lemma:s_deriv_alternative}
	Let $0 < q \leq 1$.
	For $\delta \in (1,\infty)\setminus{\alpha}$, define
	\begin{equation*}
		\Gamma_q(\delta) \coloneqq 
			-\alpha \cdot \D{\frac{1}{\alpha}}{q} +  \alpha \cdot \D{\frac{1}{\alpha}}{\frac{1}{\delta}} + \ln(c) 
	\end{equation*}

	It holds that
	\begin{equation}\label{eq:s_deriv_sign}
		\sign\biggl(\frac{\partial}{\partial\,\delta}\,s_{q}(\delta)\biggr) = \sign\biggl(\Gamma_q(\delta)\biggr).
	\end{equation}
	Moreover, it also holds that
	$\frac{\partial}{\partial\,\delta}\,s_{q}(\delta) = 0$ if and only if $\Gamma_q(\delta) = 0$.	
\end{restatable}

The proof of \cref{lemma:s_deriv_alternative} can be found in \cref{sec:omitted_proofs}.
The following lemma describes the behavior of the function $s_q(\delta)$ over specific
intervals. It demonstrates that $s_q(\delta)$ is unimodal over intervals, meaning that it exhibits a
strictly decreasing trend followed by a strictly increasing trend, or vice versa,
depending on the interval.

\begin{restatable}{lemma}{squnimodal}
	\label{lemma:s_q_unimodal}
	The function $s_q(\delta)$ is strictly decreasing for $\alpha \leq \delta \leq \sdeltar$
	and strictly increasing for $\sdeltar \leq \delta \leq \frac{1}{q}$.
	Moreover, if $\alpha > 1$, then the function $s_q(\delta)$ is strictly
	increasing for $1 \leq \delta \leq \sdeltal$ and strictly decreasing for
	$\sdeltal \leq \delta \leq \alpha$.
\end{restatable}

The proof of \cref{lemma:s_q_unimodal} can be found in \cref{sec:omitted_proofs}.

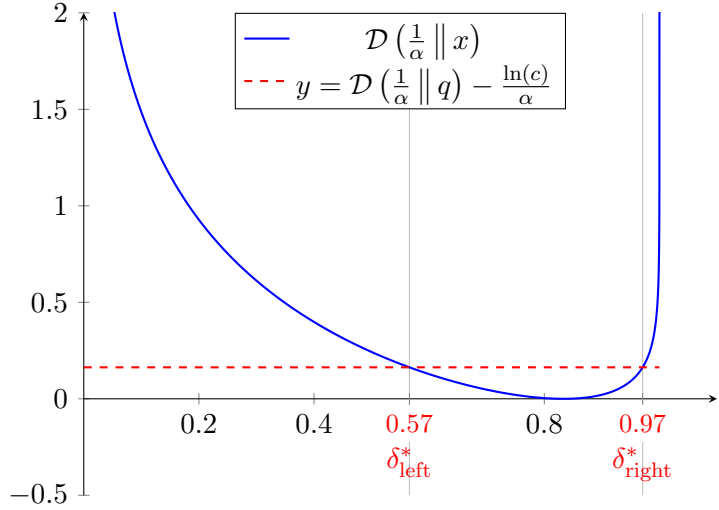
\begin{figure}[h!]
	\centering
	\begin{tikzpicture}
    \begin{axis}[
	legend style={at={(0.5,0.8)}, anchor=south},
        domain=0.01:0.99,
        samples=1000,
        ymin=-0.5, ymax=2,  %
        xtick={0, 0.2, 0.4, 0.8},
        xmin=0, xmax=1.1,
        extra x ticks={ 0.5657220994862653, 0.9708196369986882 },
        extra x tick style={grid=major, tick label style={font=\small, fill=white, text=red}},
        axis x line=center,  %
        axis y line=left,  %
        height=8cm,  %
        width=10cm,  %
    ]
    \def\q{ 0.5 }
    \def\alph{ 1.2 }
    \def\c{ 1.1 } %
    \def\r1{ 0.5657220994862653 }
    \def\r2{ 0.9708196369986882 }    

    \def\kl#1#2{(#1 * ln(#1 / #2) + (1 - #1) * ln((1 - #1) / (1 - #2)))}

    \addplot[
        domain=0.01:0.9999999,
        samples=10000,
        thick,
        blue
    ] 
    { \kl{1/ \alph}{x} }; 
    \addlegendentry{$\D{\frac{1}{\alpha}}{x}$}

    \addplot[
        domain=0:1,
        samples=2,
        thick,
        red,
        dashed
    ] 
    { \kl{1/ \alph}{\q} - ln(\c)/\alph };
    \addlegendentry{$y = \D{\frac{1}{\alpha}}{q} - \frac{\ln(c)}{\alpha}$}

    \node[below,red] at (axis cs:0.5657220994862653,-0.2) {$\sdeltal$};
    \node[below,red] at (axis cs:0.9708196369986882,-0.2) {$\sdeltar$};

    \end{axis}
\end{tikzpicture}		
	\caption{The plot of the function $\psi(x) = \D{\frac{1}{\alpha}}{x}$ and $y = \D{\frac{1}{\alpha}}{q} - \frac{\ln(c)}{\alpha}$ for $\alpha = 1.2, q = 0.5$ and $c = 1.1$. Note that $\psi(x)$ has a zero at $\frac{1}{\alpha}$ and it is monotone in intervals $(1,\frac{1}{\alpha}]$ and $[\frac{1}{\alpha},1)$. }
	\label{fig:klplot}
\end{figure}

With these definitions and results, we are now able to calculate the minimum (or maximum)
of $s_q(\delta)$ over $\delta \in \goodd \cap [1,\frac{1}{q}]$.
In \cref{lemma:min_s_beta_>_alpha,lemma:max_s_beta_<_alpha}, we separately consider
the two complementary cases based on whether $\alpha$ is greater than $\beta$ or not.

\begin{lemma}\label{lemma:min_s_beta_>_alpha}
	Let $0 < q \leq 1$ such that $c\leq \exp\left(\alpha\cdot \D{\frac{1}{\alpha}}{q}\right)$.
	Suppose that $\frac{1}{q} > \beta > \alpha$ and $\beta < \sdeltar$.
	Then
	\begin{equation*}
		\biggl(\,\min_{\delta \in \goodd \cap [1,\frac{1}{q}]} s_q(\delta)\,\biggr) = s_{q}(\sdeltar).
	\end{equation*}
\end{lemma}

\begin{proof}
		We first show that  $\alpha \leq \sdeltar \leq \frac{1}{q}$. 
		By definition, it holds that $\sdeltar \geq \alpha$.
		To prove that $\sdeltar \leq \frac{1}{q}$, suppose for a contradiction that
		$\sdeltar > \frac{1}{q}$, i.e. $\frac{1}{\sdeltar} < q$.
		Since $\frac{1}{\sdeltar} < q < \frac{1}{\alpha}$, it holds that
		\begin{equation*}
			\D{\frac{1}{\alpha}}{q} - \frac{\ln(c)}{\alpha} = \D{\frac{1}{\alpha}}{\frac{1}{\sdeltar}} > \D{\frac{1}{\alpha}}{q},
		\end{equation*}
		where the equality follows from the definition of $\sdeltar$.
		Therefore we arrive at a contradiction.
		
		By \Cref{lemma:s_q_unimodal} it holds that $s_q(\delta)$ is decreasing in $[\alpha, \sdeltar]$ and increasing in $[\sdeltar,\frac{1}{q}]$, therefore 
			\begin{equation*}
			\biggl(\,\min_{\delta \in \goodd \cap [1,\frac{1}{q}]} s_q(\delta)\,\biggr) = s_{q}(\sdeltar).
		\end{equation*}
\end{proof}

Next, we state the analogue of \cref{lemma:min_s_beta_>_alpha} for the $\beta > \alpha$ case.
The proof \cref{lemma:max_s_beta_<_alpha}, which is nearly identical to the proof of \cref{lemma:min_s_beta_>_alpha},
can be found in \cref{sec:omitted_proofs}.

\begin{restatable}{lemma}{maxsbetasmalleralpha}
	\label{lemma:max_s_beta_<_alpha}
	Let $(\alpha, \beta) \in \params$, $0 < q \leq 1$, and $c \geq 1$ such that $c\leq \exp\left(\alpha\cdot \D{\frac{1}{\alpha}}{q}\right)$.
	Suppose that $1 < \beta < \alpha < \frac{1}{q}$ and $\beta > \sdeltal$.
	Then
	\begin{equation*}
		\biggl(\,\max_{\delta \in \goodd \cap [1,\frac{1}{q}]} s_q(\delta)\,\biggr) = s_{q}(\sdeltal).
	\end{equation*}	
\end{restatable}

Finally, we present the proof of \cref{lemma:min_f_convert_equiv}.

\begin{proof}[Proof of \cref{lemma:min_f_convert_equiv}]
	First, let us assume that $\beta \geq \alpha$.
	By \eqref{eq:f_h_equivalence}, it holds that
	\begin{align}
		\min_{\delta \in \goodd \cap [1,\frac{1}{q}]} \tilde{f}(\delta,q) &= \exp\biggl(\min_{\delta \in \goodd \cap [1,\frac{1}{q}]} h_{q}(\delta)\biggr) \notag\\
										  &= \exp\Biggl(\ln(c) + (\beta - \alpha) \cdot \biggl(\min_{\delta \in \goodd \cap [1,\frac{1}{q}]} s_q(\delta)\biggr)\Biggr).			\label{eq:tilde_f_s_q_equiv}
	\end{align}
	where the second equality follows from \cref{lemma:minim_h_slope}.
	 For $\sdeltar > \beta$, we have
	\begin{align*}
		  \exp\Biggl(\ln(c) + (\beta - \alpha) \cdot \biggl(\min_{\delta \in \goodd \cap [1,\frac{1}{q}]} s_q(\delta)\biggr)\Biggr) &=  \exp\Biggl(\ln(c) + (\beta - \alpha) \cdot \biggl(s_{q}(\sdeltar)\biggr)\Biggr)\\
&= c \cdot \exp\Biggl(\frac{\ln\left( \frac{1}{\phi\left( \sdeltar, q \right) } \right) - \ln(c) }{\sdeltar - \alpha} \cdot \left( \beta - \alpha \right) \Biggr)\\
&= c \cdot \exp\Biggl(\frac{\sdeltar\cdot \D{\frac{1}{\sdeltar}}{q} - \ln(c) }{\sdeltar - \alpha} \cdot \left( \beta - \alpha \right) \Biggr)\\
&= \runtime
	\end{align*}
	where the first equality follows from \cref{lemma:min_s_beta_>_alpha}, and the second and third equalities follow from the definition of the functions $s_{q}(\delta)$ and $\phi(\delta, q)$ respectively. If $\beta \geq \sdeltar$, then it holds that
	\begin{equation*}
		\goodd \cap [1,\frac{1}{q}] \subset [\sdeltar, \frac{1}{q}].
	\end{equation*}
	Therefore, by \cref{lemma:s_q_unimodal}, the function $s_q(\delta)$ is strictly increasing over $\goodd \cap [1,\frac{1}{q}]$.
	Hence, we get
	\begin{align*}
		\exp\Biggl(\ln(c) + (\beta - \alpha) \cdot \biggl(\min_{\delta \in \goodd \cap [1,\frac{1}{q}]} s_q(\delta)\biggr)\Biggr) &=  \exp\Biggl(\ln(c) + (\beta - \alpha) \cdot \biggl(s_{q}(\beta)\biggr)\Biggr)\\
																	  &= \exp\Biggl(\ln(c) + \frac{\ln\left( \frac{1}{\phi\left( \beta, q \right) } \right) - \ln(c) }{\beta - \alpha} \cdot \left( \beta - \alpha \right) \Biggr)\\
																	  &= \exp\Biggl(\ln(c) +\frac{\beta\cdot \D{\frac{1}{\beta}}{q} - \ln(c) }{\beta  - \alpha} \cdot \left( \beta - \alpha \right) \Biggr)\\
																	  &= \exp\left(\beta\cdot \D{\frac{1}{\beta}}{q}\right)\\
																	  &= \runtime.
	\end{align*}
	Therefore the claim holds for $\beta \geq \alpha$. The proof for the case $\beta < \alpha$ is nearly identical to the one above and is left to the reader.
\end{proof}

\section{Properties of  $\sdeltal$ and $\sdeltar$}
\label{sec:sdelta_expl_formula}

In this section we show properties of the functions $\sdeltal$ and $\sdeltar$ which has been defined in \eqref{eq:deltast_def}.  We first prove \Cref{lemma:sdelta_well_defined} which shows the functions are well defined. Then, we provide a closed formula for $\sdeltal$ and $\sdeltar$ for some special cases. 
We also use the   closed formulas   to prove 
 \cref{theorem:summary_sampling_step_alpha_1,theorem:summary_sampling_step_alpha_2} which provide a closed formula for $\runtime$ in the these special cases.

\sdeltawelldefined*
\begin{proof}
	Note that $c\leq \exp\left(\alpha\cdot \D{\frac{1}{\alpha}}{q}\right)$ implies that $\D{\frac{1}{\alpha}}{q} - \frac{\ln(c)}{\alpha} \geq 0$.
	Furthermore, observe that $\D{\frac{1}{\alpha}}{x}$ is a differentiable, non-negative convex function of $x$
	and has a global minimizer at $\frac{1}{\alpha}$.
	Therefore, $\D{\frac{1}{\alpha}}{x}$ is strictly decreasing for $x \leq \frac{1}{\alpha}$
	and strictly increasing for $x \geq \frac{1}{\alpha}$.
	
	Let us now consider the range $I_1 = [\alpha, \infty)$.
	If $\delta \in I_1$, then $\frac{1}{\delta} \leq \frac{1}{\alpha}$,
	consequently the function $\D{\frac{1}{\alpha}}{\frac{1}{\delta}}$ is a strictly
	increasing function of $\delta$ for $\delta \in I_1$.
	Additionally, it holds that
	\begin{equation}\label{eq:D_to_inf}
		\lim_{x \to 1} \D{\frac{1}{\alpha}}{x} = \lim_{x \to 0} \D{\frac{1}{\alpha}}{x} = \infty.
	\end{equation}
	Hence, there exists a unique value of $\delta \in I_1$ such that
	\begin{equation}\label{eq:D_eq_claim}
		\D{\frac{1}{\alpha}}{\frac{1}{\delta}} = \D{\frac{1}{\alpha}}{q} - \frac{\ln(c)}{\alpha} \geq 0,
	\end{equation}
	which implies that $\sdeltar$ is well-defined.
	Similarly, let's consider the case where $\alpha > 1$ and define the interval $I_2=(1,\alpha]$.
	In this case, the function $\D{\frac{1}{\alpha}}{\frac{1}{\delta}}$ is strictly
	decreasing for $\delta \in I_2$.
	By \eqref{eq:D_to_inf}, it holds that there exists a unique value of $\delta \in I_2$
	such that \cref{eq:D_eq_claim} holds.
\end{proof}

Next, we  first give a closed formula for $\sdeltar$ in case $\alpha=1$. 

\begin{lemma}\label{lemma:sampling_2_approx_alpha_1}
	Let $\alpha = 1$,  $0 < q \leq 1$ and $c \geq 1$
	such that $c \leq \frac{1}{q}$.
	Then we have
	\begin{equation*}
		\sdeltar(\alpha, c, q) = \frac{1}{q \cdot c}.
	\end{equation*}
\end{lemma}
\begin{proof}
 Since $1\leq c\leq \frac{1}{q}$ it holds that $\frac{1}{c\cdot q} \geq 1 = \alpha$. Furthermore,
 $$
 \D{\frac{1}{\alpha}}{\frac{1}{\frac{1}{c\cdot q}}} \,=\,  -\ln(q\cdot c) \,=\, -\ln(q)-\ln(c) \,=\, \D{\frac{1}{\alpha}}{q} -\ln(c).  
  $$
  The first and last equalities holds as $\alpha=1$ and due to \eqref{eq:KL_eq_1}.  Therefore, by the definition of $\sdeltar$ in~\eqref{eq:deltast_def} we have $		\sdeltar(\alpha, c, q) = \frac{1}{q \cdot c}$. 
\end{proof}

 \cref{theorem:summary_sampling_step_alpha_1} is a simple consequence of \Cref{lemma:sampling_2_approx_alpha_1}.
 \simpslaphaone*
\begin{proof}
	Note that by \cref{lemma:sampling_2_approx_alpha_1}, we have that $\sdeltar=\sdeltar(1,c,q) = \frac{1}{q \cdot c}$.
	By \eqref{eq:runtime}, if $\beta < \frac{1}{q \cdot c}$, then 
	\begin{equation}
		\label{eq:simple_alpha_one_first}
	\begin{aligned}
	\runtime[1,\beta,c,q]\,&=\, c \cdot \exp\biggl(\frac{ \sdeltar \cdot \D{\frac{1}{\sdeltar}}{q}- \ln(c) }{\sdeltar - 1} \cdot \left( \beta - 1 \right) \biggr) \\
	&=\,  c \cdot \exp\biggl(\frac{ \frac{1}{c\cdot q} \cdot \D{c\cdot q}{q}- \ln(c) }{\frac{1}{cq} - 1} \cdot \left( \beta - 1 \right) \biggr).
	\end{aligned}
	\end{equation}
	By the formula of Kullback-Leibler divergence we have,
 	\begin{equation}
 				\label{eq:simple_alpha_one_second}
 	\begin{aligned}\frac{1}{c\cdot q} \cdot \D{c\cdot q}{q}- \ln(c) \,&=\, \frac{ c\cdot q }{c\cdot q}\cdot \ln\left( \frac{c\cdot q}{q}\right)+ \frac{1-c\cdot q}{c\cdot q } \cdot \ln \left(\frac{1-c\cdot q}{1-q}\right)-\ln(c)
 		\\
 		&=\,\left( \frac{1}{c\cdot q}- 1\right)\cdot   \ln \left(\frac{1-c\cdot q}{1-q}\right).
 		\end{aligned}
 		\end{equation}
 		By \eqref{eq:simple_alpha_one_first} and \eqref{eq:simple_alpha_one_second} we have 
 		$$
 			\begin{aligned}
 			\runtime[1,\beta,c,q]\,
 			&=\,  c \cdot \exp\biggl(\frac{ \frac{1}{c\cdot q} \cdot \D{c\cdot q}{q}- \ln(c) }{\frac{1}{cq} - 1} \cdot \left( \beta - 1 \right) \biggr)\\
 			&=\, c \cdot \exp\biggl(\frac{ \left( \frac{1}{c\cdot q}- 1\right)\cdot   \ln \left(\frac{1-c\cdot q}{1-q}\right) }{\frac{1}{cq} - 1} \cdot \left( \beta - 1 \right) \biggr)\\
 			&= c\cdot  \left( \frac{1-c\cdot q}{1-q}\right)^{\beta-1}.
 		\end{aligned}
 		$$

	On the other hand,
	if $\beta \geq \frac{1}{q \cdot c}$, again by \eqref{eq:runtime} it follows that
	\begin{align*}
		\runtime[1,\beta,c,q] &=  \exp\Bigl(\beta \cdot \D{\frac{1}{\beta}}{q}\Bigr).
	\end{align*}
	Therefore, the theorem follows. 
\end{proof}

Finally,  we also provide a closed formula for $\sdeltal$ in case $\alpha=2$ and $c=1$. Here, we rely on the fact that $\D{\frac{1}{2}}{x}$ is symmetric around $\frac{1}{2}$. 
\begin{lemma}\label{lemma:sampling_2_approx_alpha_2}
	Let $\alpha = 2$, $0 < q \leq \frac{1}{\alpha} = \frac{1}{2}$ and $c = 1$.
	Then we have
	\begin{equation*}
		\sdeltal(\alpha,c,q) = \frac{1}{1 - q}.
	\end{equation*}
\end{lemma}
\begin{proof}
We first observe that $\frac{1}{1-q} > 1$ since $q>0$ and $\frac{1}{1-q} \leq \frac{1}{1-\frac{1}{2}}\leq 2=\alpha$ as $q\leq \frac{1}{2}$. Thus, $\frac{1}{1-q} \in (1,\alpha]$.  Furthermore,
$$
\D{\frac{1}{\alpha}}{\frac{1}{\frac{1}{1-q}}} \,=\,\D{\frac{1}{2}}{1-q} \, = \D{\frac{1}{2}}{q} = \D{\frac{1}{\alpha}}{q} -\frac{\ln(c)}{\alpha},
$$
the second equality holds as $\D{\frac{1}{2}}{x} =\frac{1}{2}\cdot \ln \left(\frac{1}{2x} \right) + \frac{1}{2}\cdot \ln \left(\frac{1}{2\cdot (1-x)}\right) $ is symmetric around $x=\frac{1}{2}$, and the last equality holds as $\ln(c)=\ln(1)=0$.  Thus, we have $\sdeltal(\alpha,c,q)=\frac{1}{1-q}$ by its definition in \eqref{eq:deltast_def}.
\end{proof}
\simplealphatwo*
\begin{proof}
	 By  \Cref{lemma:sampling_2_approx_alpha_2} it holds that 
	 $\sdeltal=\sdeltal(2,1,q) = \frac{1}{1 - q}$ by \cref{lemma:sampling_2_approx_alpha_2}. 
	In case $\beta>\frac{1}{1-q}$, by  \eqref{eq:runtime} it holds that 
	 \begin{equation*}
	 	\label{eq:simple_alpha_two_first}
	 	\begin{aligned}
	\runtime[2,\beta,1,q]&= 1 \cdot \exp\biggl(\frac{\sdeltal \cdot \D{\frac{1}{\sdeltal}}{q}- \ln(1) }{\sdeltal - 2} \cdot \left( \beta - 2 \right) \biggr)\\
	&= \exp\biggl(\frac{\frac{1}{1-q} \cdot \D{1-q}{q} }{\frac{1}{1-q} - 2} \cdot \left( \beta - 2 \right) \biggr)\\
	& = \exp\Biggl(\frac{\frac{1-q}{1-q} \cdot  \ln \left(\frac{1-q}{q}\right)+ \frac{q}{1-q}\cdot \ln\left( \frac{q}{1-q}\right)}{\frac{1}{1-q} - 2} \cdot \left( \beta - 2 \right) \Biggr)\\
	&=\exp\Biggl(\frac{\left(\frac{1}{1-q}-2\right) \ln\left( \frac{q}{1-q}\right) }{\frac{1}{1-q} - 2} \cdot \left( \beta - 2 \right) \Biggr)\\
	&= \left(\frac{1}{1-q}\right)^{\beta-2}. 
	\end{aligned}
	 \end{equation*}

	In case $1 \leq \beta \leq \frac{1}{1 - q}$, again by \eqref{eq:runtime} it follows that
	\begin{align*}
		\runtime[2,\beta,1,q] &=  \exp\biggl(\beta \cdot \D{\frac{1}{\beta}}{q}\biggr).
	\end{align*}
	Therefore,the theorem  holds.
\end{proof}

\section{Sampling Steps}
\label{sec:sampling_step}
In this section we  provide sampling steps used by our applications. \Cref{sec:fvs,sec:povd} give the sampling steps for \fvs and \povd. In \Cref{sec:forbidden} we give the generic sampling step for $\gpivd$ problems in which $\Pi$ is  defined by a finite set of forbidden sub-hypergraphs. 

\subsection{Feedback Vertex Set}
\label{sec:fvs}
In this section we will prove \cref{lemma:fvs_sampling}, that is we provide 
a sampling step for \fvs (\FVS) with success probability $\frac{1}{4}$. 
Let $\mathcal{G}$ denote the set of graphs and
let $\fvsPi$ denote the set of graphs without cycles.
Observe that $\fvsPi$ is a hereditary hypergraph property.
For an input graph $G$, the $\fvs$ problem asks whether there exists a set $S \subseteq V(G)$ of size $k$ such that $G \setminus S \in \fvsPi$, i.e. $G \setminus S$ is acyclic.

The sampling step for \FVS\ is given in \Cref{algorithm:fvs_sampling_step}.  It starts by iteratively removing vertices of degree at most $1$,
as described in \cref{algorithm:core_graph}.
Given an input graph $G$, we refer to the resulting graph from this procedure
as $\core{G}$.
It is evident that the sets of cycles in $G$ and $\core{G}$
are equal because a vertex $v \in V(G)$ with degree at most 1 cannot be
part of a cycle.
Therefore, it holds that
\begin{equation}\label{eq:G_cycle_core}
	\OPT_{\fvsPi}(\core{G}) = \OPT_{\fvsPi}(G).		
\end{equation}

After computing $\core{G})$, the sampling step defines a weight for every vertex.
The weight of a vertex $v$ of degree $2$ (in $\core{G}$) is zero, and the weight of every other vertex is its degree. The algorithm returns a random vertex, so the probability of every vertex to be returned is proportional to is weight.  The algorithm also handles a corner case which occurs if $\core{G}$ contains a cycle of vertices of degree $2$ by sampling a random vertex from this cycle. 
\begin{algorithm}
	\begin{algorithmic}[1]
		\Input Graph $G$
		\While{$G$ has a vertex $u$ of degree at most 1}
			\State $G\leftarrow G- \{u\}$, i.e., remove $u$ from $G$
		\EndWhile
		\State \Return $G$
	\end{algorithmic}
	\caption{Computing the $\core{G}$ for a graph $G$}\label{algorithm:core_graph}
\end{algorithm}

\begin{algorithm}
	\begin{algorithmic}[1]
		\Input Graph $G$
		\State $G \gets \core{G}$
		\If{$G$ has a connected component $C$ with maximum degree 2}
			\State \Return a vertex $v \in C$ uniformly at random \label{line:conn_comp_max_deg_2}
		\EndIf
		\State For each $v \in V(G)$, define $w(v) = \begin{cases}
			0 &\text{if } \deg(v) = 2\\
			\deg(v) &\text{otherwise} 
		\end{cases}$
		\State $W \gets \sum_{v \in V(G)} w(v)$
		\State Sample a vertex $v \in V(G)$ with respect to probabilities $\frac{w(v)}{W}$
		\State \Return $v$
	\end{algorithmic}
	\caption{Sampling Step for $\fvs$}\label{algorithm:fvs_sampling_step}
\end{algorithm}

The proof of the next lemma is an adjustment of the arguments in \cite{beckerRandomizedAlgorithmsLoop2000} (see also \cite{cyganParameterizedAlgorithms2015a}). 
\begin{lemma}\label{lemma:fvs_proc}
	 \cref{algorithm:fvs_sampling_step} is a sampling step for $\FVS$ with success probability~$\frac{1}{4}$.
\end{lemma}
 \Cref{lemma:fvs_sampling} is an immediate consequence of \Cref{lemma:fvs_proc}.
\begin{proof}[Proof of \Cref{lemma:fvs_proc}]
	Initially, let us demonstrate that the procedure operates within
	polynomial time. Computing the $\core{G}$ and determining whether the
	graph has a connected component with maximum degree $2$ takes linear time.
	Similarly, computing the values
	$w(v)$ for each vertex $v \in V(G)$ also takes linear time. Therefore,
	the entire algorithm runs in linear time.

	By \eqref{eq:G_cycle_core}, we can, without loss of generality,
	assume that $G$ has no vertices of degree less than 2.
	Similarly, we can also assume that each component of $G$ has at least one vertex of
	degree at least 3, because otherwise the algorithm returns a vertex $v$ at
	\cref{line:conn_comp_max_deg_2} which satisfies
	\begin{equation*}
		\OPT_{\fvsPi}(G \setminus v) \leq \OPT_{\fvsPi}(G) - 1
	\end{equation*}
	with probability 1. 

	\begin{claim}
		There exists a minimum feedback vertex set of $G$ such that each vertex in it
		has degree at least 3.
	\end{claim}

	\begin{claimproof}
		Let $B$ be a minimum feedback vertex set. Since $G$ does not contain
		any vertices of degree less than 2, $B$ also does not contain
		any vertices of degree less than 2. Suppose $B$ contains a
		vertex $v$ of degree exactly 2. We claim that there is always a
		vertex $u \in V(G)$ of degree at least 3 such that $\left( B
		\setminus \{v\} \right) \cup \{u\}$ is also a minimum feedback vertex
		set. Observe that this is enough to prove the claim, as we can
		apply this step repeatedly as long as $B$ contains a vertex of
		degree 2.

		Since $v$ is in a connected component of which contains a vertex of degree $3$ or more, there is a path from $v$ to $u$, such that every vertex on the path expect $u$ is of degree $2$. Thus, every cycle $C$ in $G$ that contains $v$ also contains $u$. Therefore, $\left( B \setminus \{v\} \right) \cup \{u\}$ is also a minimum 
		feedback vertex set of $G$. 
	\end{claimproof}

	Let $B$ be a minimum feedback vertex set of $G$ such that each vertex has degree
	at least 3.
	In the following, for a set $X \subseteq V(G)$, we let $w(X) \coloneqq \sum_{x \in X} w(x)$.	
	The following claim argues that
	the weight of $B$ is large.
	\begin{claim}\label{claim:B_weight_large}
		It holds that
		\begin{equation*}
			w(B) \geq \frac{w\left( V(G) \setminus B \right)}{3}.
		\end{equation*}
	\end{claim}

	\begin{claimproof}
		Define $R = V(G) \setminus B$. We have
		\begin{align*}
			w(R) &= \sum_{v \in R \colon \deg(v) \geq 3} \deg(v)\\
			     &= \left( \sum_{v \in R} \deg(v) \right)  - \left( \sum_{v \in R \colon \deg(v) = 2} \deg(v) \right) \\
			     &= \left( \sum_{v \in R} \deg(v) \right) - 2 \cdot \abs{\{v \in R \mid \deg(v) = 2\} }\\
			     &= \left( \sum_{v \in R} \deg(v) \right) - 2 \cdot \abs{R_2}
		\end{align*}
		where $R_2 \coloneqq \{v \in R \mid \deg(v) = 2\}$. Observe that $\left( \sum_{v \in R} \deg(v) \right)$
		is equal to the number of edges between the vertices in $B$ and $R$,
		plus twice the number of edges between vertices in $R$.
		Therefore we get
		\begin{align*}
			w(R) &= \abs{E(B,R)} + 2 \cdot \abs{E(G[R])} - 2 \cdot \abs{R_2}\\
			     &\leq \abs{E(B,R)} + 2 \cdot \abs{R} - 2 \cdot \abs{R_2}\\
			     &= \abs{E(B,R)} + 2 \cdot \abs{R_{\geq 3}} 
		\end{align*}
		where $\abs{R_{\geq 3}} \coloneqq \{v \in R \mid \deg(v) \geq 3\}$
		and $\abs{E(B,R)} = \abs{\{(u,v) \mid u \in B, v \in R, (u,v) \in E(G)\}}$.
		Furthermore, the inequality holds because $G[R]$ is a forest by definition.
		Since each $v \in R_{\geq 3}$ contributes at least 3 to
		$w(R)$, we have $3\cdot \abs{R_{\geq 3}} \leq w(R)$ and
		\begin{align*}
			w(R) &\leq \abs{E(B,R)} + 2 \cdot \frac{w(R)}{3}
		\end{align*}
		which implies that $\frac{w(R)}{3} \leq \abs{E(B,R)}$.
		Finally, since $w(B) \geq \abs{E(B,R)}$, it holds that
		\begin{equation*}
			w(B) \geq \abs{E(B,R)} \geq \frac{w(R)}{3}.
		\end{equation*}
	\end{claimproof}

	\cref{claim:B_weight_large} implies that
	\begin{equation*}
		\Pr\left( v \in B \right) = \frac{w(B)}{W} = \frac{w(B)}{w(B) + w\left( V(G) \setminus B \right) } \geq \frac{w(B)}{w(B) + 3\cdot w(B)} = \frac{1}{4}.
	\end{equation*}

	Finally, since $v \in B$ implies $\OPT_{\fvsPi}(G \setminus v) \leq \OPT_{\fvsPi}(G) - 1$,
	we get
	\begin{equation*}
		\Pr\biggl(\OPT_{\fvsPi}(G \setminus v) \leq \OPT_{\fvsPi}(G) - 1\biggr) \geq \Pr\left( v \in B \right) \geq \frac{1}{4}.
	\end{equation*}
\end{proof}

\subsection{Pathwidth One Vertex Deletion}
\label{sec:povd}

Let $\mathcal{G}$ denote the set of graphs,
and let $\povdPi$ denote the set of graphs with pathwidth at most 1. Given a graph
$G \in \mathcal{G}$ as input, the $\spovd$ ($\povd$) problem asks for a set
$S \subseteq V(G)$ such that $G \setminus S$ belongs to $\povdPi$, i.e., $G
\setminus S$ has pathwidth at most 1.

Let $T_2$ be the graph with 7 vertices, where we take three paths with 3 vertices each,
and identify one
of the degree 1 vertices of each path (see \cref{fig:T2}).
Let $\Omega$ denote the set of all cycle graphs together with the graph $T_2$.
An alternative characterization of $\povdPi$ is the following: a graph $G$
belongs to $\povdPi$ if and only if $G$ has no subgraph isomorphic to a graph in
$\Omega$ \cite{philipQuarticKernelPathwidthOne2010,
tsurFasterAlgorithmPathwidth2022}.
Note that the set of forbidden subgraphs in the case of $\fvs$ is $\Omega \setminus \{T_2\}$,
hence, it is natural to adapt the sampling step for $\fvs$ to $\spovd$.

\begin{figure}[htpb]
	\centering
	\includegraphics{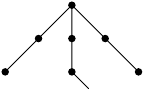}
	\caption{The graph $T_2$}
	\label{fig:T2}
\end{figure}

\begin{algorithm}
	\begin{algorithmic}[1]
		\Input Hypergraph $G \in \left( \mathcal{G} \setminus \povdPi \right) $
		\If{$G$ has a subgraph $Z$ isomorphic to $T_2$}
			\State \Return a vertex $v \in V(Z)$ uniformly at random \label{line:T2}
		\ElsIf
			\State \;Run \cref{algorithm:fvs_sampling_step} (Sampling Step for $\fvs$) on $G$
		\EndIf
	\end{algorithmic}
	\caption{Sampling step for $\povd$}\label{algorithm:povd_sampling_step}
\end{algorithm}

Our sampling step for \spovd, given in 
\cref{algorithm:povd_sampling_step} i,
first checks whether $G$ has a subgraph isomorphic to $T_2$. If this is the case, is samples a random vertex of this subgraph. 
If not, then $G$ should have a subgraph isomorphic to cycle and the sampling step for \FVS\ is used.

The following lemma show that \cref{lemma:povd_sampling} holds.
\begin{lemma}\label{lemma:fvs_sampling_steo}
	\cref{algorithm:povd_sampling_step} is a sampling step for $\spovd$
	with success probability $\frac{1}{7}$.
\end{lemma}

\begin{proof}
	First, let us show that \cref{algorithm:povd_sampling_step} runs in polynomial time.
	Checking whether $G$ has a subgraph $Z$ isomorphic to $T_2$ can be done
	by going over all subgraphs of $G$ of size 7, which takes polynomial time.
	Moreover, \cref{algorithm:fvs_sampling_step} runs in polynomial time
	by \cref{lemma:fvs_proc}.

	Define the partition $(\mathcal{G}_1,\mathcal{G}_2)$ of $\left( \mathcal{G} \setminus \povdPi \right)$ where
	\begin{equation*}
		\mathcal{G}_1 \coloneqq \{G \in \left( \mathcal{G} \setminus \povdPi \right) \mid G \text{ has a subgraph isomorphic to } T_2 \}. 
	\end{equation*}
	and $\mathcal{G}_2 \coloneqq \biggl(\Bigl( \mathcal{G} \setminus \povdPi \Bigr) \setminus \mathcal{G}_1\biggr)$.

	First assume that $G \in \mathcal{G}_1$ and let $Z$ be the subgraph of $G$ which
	is isomorphic to $T_2$. Then, the algorithm returns a vertex $v \in V(Z)$
	sampled uniformly at random, at \cref{line:T2} of \cref{algorithm:povd_sampling_step}.
	Now let $S \in \sat_{\povdPi}(G)$, and observe that $v \in S$ implies that
	\begin{equation*}
		\OPT_{\povdPi}(G \setminus v) \leq \Bigl(\OPT_{\povdPi}(G) - 1\Bigr).
	\end{equation*}
	Moreover, $(V(Z) \cap S) \neq \emptyset$ because otherwise $S$ is not a solution.
	Therefore
	\begin{equation}\label{eq:povd_eq_1}
		\Pr\Bigl( \OPT_{\povdPi}(G \setminus v) \leq \OPT_{\povdPi}(G) - 1 \Bigr) \geq \Pr\left( v \in S \right) = \frac{\abs{V(Z) \cap S}}{\abs{V(Z)}} \geq \frac{1}{7}.
	\end{equation}

	Now assume that $G \in \mathcal{G}_2$.
	By the alternative
	characterization of $\povdPi$, it follows that $G$ contains
	a subgraph isomorphic to a graph in $\Omega \setminus \{T_2\}$.
	Therefore, for $G \in \mathcal{G}_2$, it holds that
	\begin{equation*}
		G\setminus S \in \povdPi \iff G\setminus S \in \fvsPi.
	\end{equation*}

	Therefore the problems $\gpivd[\mathcal{G}_2, \povdPi]$ and  $\gpivd[\mathcal{G}_2, \fvsPi]$
	are equivalent. Moreover, \cref{algorithm:fvs_sampling_step} is a sampling step for
	$\gpivd[\mathcal{G}_2, \povdPi]$ with success probability $\frac{1}{4}$, and by definition
	it returns a vertex $v$ such that
	\begin{equation}\label{eq:povd_eq_2}
		\Pr\Bigl( \OPT_{\povdPi}(G \setminus v) \leq \OPT_{\povdPi}(G) - 1 \Bigr) \geq \frac{1}{4}.
	\end{equation}

	Therefore, by \eqref{eq:povd_eq_1} and \eqref{eq:povd_eq_2},
	\cref{algorithm:povd_sampling_step} is a sampling step for $\spovd$ with a
	success probability of $\frac{1}{7}$.
\end{proof}

\subsection{\boldmath $(\mathcal{G},\Pi)$-Vertex Deletion for a finite set of forbidden sub-hypergraphs}
\label{sec:forbidden}

We are left to prove \Cref{lemma:finite_forb_sampling}. That is, we 
 describe a sampling step for $\gpivd[\mathcal{G}, \Pi^{\Omega}]$ where
$\Pi^{\Omega}$ is described by a finite set of forbidden subhypergraphs (\Cref{definition:forbidden_hypergraph}).

In the remainder of this section, let $\mathcal{G}$ be a fixed, closed set of hypergraphs, and let $\Omega = \{F_1, \ldots, F_\ell\}$ be a fixed finite set of hypergraph.  
Recall  $\eta(\Omega) \coloneqq \max_{1 \leq i \leq \ell} \abs{V\left( F_i
\right) }$.
The idea  in the sampling step for $\gpivd[\mathcal{G}, \Pi^{\Omega}]$ is very simple, if a hypergraph $G \in \mathcal{G}$ does not belong to $\Pi^{\Omega}$,
then $G$ should have a subhypergraph $Z$ isomorphic to $F_i$ for some $1 \leq i \leq \ell$.
Moreover, any solution $S \in \sat_{\Pi^{\Omega}}(G)$ should contain a vertex from $Z$, otherwise
$(G \setminus S) \not\in \Pi^{\Omega}$. We combine these ideas in \cref{algorithm:sampling_step_ffh}.

\begin{algorithm}
	\begin{algorithmic}[1]
		\Configuration A closed set of hypergraphs $\mathcal{G}$, a set of hypergraphs $\Omega = \{F_1, \ldots, F_\ell\}$, the hypergraph property $\Pi^{\Omega}$
		\Input $G \in \mathcal{G} \setminus \Pi^{\Omega}$
		\For{$1 \leq i \leq \ell$}
			\For{$Z \subseteq V(G)$ such that $\abs{Z} = \abs{F_i}$}
			\If{$G[Z]$ is isomorphic to $F_i$}
					\State Let $v \in Z$ be a vertex sampled uniformly at random 
					\State \Return $v$\label{line:first_ffh}
				\EndIf
			\EndFor
		\EndFor
	\end{algorithmic}
	\caption{Sampling step for $\gpivd[\mathcal{G}, \Pi^{\Omega}]$}\label{algorithm:sampling_step_ffh}
\end{algorithm}

\begin{lemma}\label{lemma:sampling_step_ffh}
	\cref{algorithm:sampling_step_ffh} is a sampling step for $\gpivd[\mathcal{G}, \Pi^{\Omega}]$ with success probability $\frac{1}{\eta(\Omega)}$.
\end{lemma}
We note that  \cref{lemma:finite_forb_sampling} follows immediately from \Cref{lemma:sampling_step_ffh}.
\begin{proof}[Proof of \Cref{lemma:sampling_step_ffh}]
	For each $1 \leq i \leq \ell$, the algorithm goes over all subsets of
	$V(G)$ of size $\abs{F_i}$, which takes time at most
	$n^{\Oh(\abs{F_i})} = n^{\Oh( \eta(\Omega))} = n^{\Oh(1)}$. Checking
	whether $G[Z]$ is isomorphic to $F_i$ takes constant time since
	$\abs{F_i}$ is constant. Hence all in all, the algorithm runs in
	polynomial time.
	Moreover, the algorithm always returns a vertex because
	$G \in \left( \mathcal{G} \setminus \Pi^{\Omega} \right)$
	and there exists $Z \subseteq V(G)$ such that $Z$ is isomorphic to
	$F_i$ for some $F_i \in \Pi$.

	Now let $v$ be the output of the algorithm and $Z \subseteq V(G)$
	be the set $v$ is sampled from.
	Consider $S \in \sat_{\Pi^{\Omega}}(G)$.
	Observe that if $v \in S$, then $S \setminus \{v\}$ is a solution for $G \setminus v$.
	Therefore,
	\begin{equation}\label{eq:sampling_step_ffh_eq_1}
		\Pr\Bigl(v \in S\Bigr) \leq \Pr\Bigl( \OPT_\Pi(G \setminus v) \leq \OPT_\Pi(G) - 1 \Bigr).
	\end{equation}

	Next, observe that $(Z \cap S) \neq \emptyset$, because otherwise we would
	have $(G\setminus S) \not\in \Pi^{\Omega}$ because $Z$ is isomorphic to $F_i$.
	Since $S \in \sat_{\Pi^{\Omega}}(G)$, this implies that $(Z \cap S) \neq \emptyset$.
	
	Since $v \in Z$ is sampled uniformly, we get
	\begin{equation}\label{eq:sampling_step_ffh_eq_2}
		\Pr\left( v \in S \right) = \frac{\abs{S \cap Z}}{\abs{Z}} \geq \frac{1}{\abs{Z}} \geq \frac{1}{\eta(\Omega)}.
	\end{equation}

	Finally, by \eqref{eq:sampling_step_ffh_eq_1} and \eqref{eq:sampling_step_ffh_eq_2}, it holds that
	\begin{equation*}
		\Pr\Bigl( \OPT_\Pi(G \setminus v) \leq \OPT_\Pi(G) - 1 \Bigr) \geq \frac{1}{\eta(\Omega)}.
	\end{equation*}
\end{proof}

\section{Additional Applications}
\label{sec:addtl_applications}
In this section, following \cref{sec:applications}, we present additional applications of our results.

\subsection{\hs{3}}
Observe that \hs{3} is equivalent to $\gpivd$ where $\mathcal{G}$ is the set of
all hypergraphs with edge cardinality $3$ and $\Pi$ is the set of all edgeless
hypergraphs. Moreover, let $F$ be a hypergraph with a single edge of cardinality
3 and define $\Omega \coloneqq \{F\}$ such that $\eta\left( \Omega \right) = 3$. Furthermore, 
let $\Pi^{\Omega}$ be as in \cref{definition:forbidden_hypergraph}.
Note that $\gpivd$ is also equivalent to $\gpivd[\mathcal{G}, \Pi^{\Omega}]$, because for $G \in \mathcal{G}$,
it holds that $G \in \Pi$ if and only if $G$ doesn't have an edge, i.e. there is no vertex induced
subhypergraph of $G$ isomorphic to $F$.
By \cref{lemma:sampling_step_ffh}, there is
a sampling step for $\hs{3}$ with success probability $\frac{1}{3}$.

We will utilize the FPT algorithm from \cite{wahlstromAlgorithmsMeasuresUpper2007}
which runs in time $2.076^{k} \cdot n^{\Oh(1)}$.
\begin{lemma}
	There is a randomized parameterized $\beta$-approximation algorithm
	for $\hs{3}$ with running time
	\begin{equation*}
		d^{k} \cdot n^{\Oh(1)}
	\end{equation*}
	where we have	
	\begin{equation}\label{eq:alpha_1_eq}
		d = \begin{cases}
			2.076 \cdot \left(0.462\right)^{\beta - 1} &\text{if } 1 < \beta < 1.445\\
			\exp\left( \beta \cdot \D{\frac{1}{\beta}}{\frac{1}{3}} \right) &\text{if } 1.445 \leq \beta < 3.
		\end{cases}
	\end{equation}	
\end{lemma}

\begin{proof}
	The lemma follows from the FPT algorithm of \cite{wahlstromAlgorithmsMeasuresUpper2007} ($\alpha = 1, c = 2.076$) and \cref{theorem:summary_sampling_step_alpha_1} by setting $q = \frac{1}{3}$.
\end{proof}

In \cref{fig:3hsplot}, we compare our results with those from \cite{brankovicParameterizedApproximationAlgorithms2012}, \cite{Fellows2018}, and \cite{KulikS2020}.

\begin{figure}[h!]
	\centering
	\input{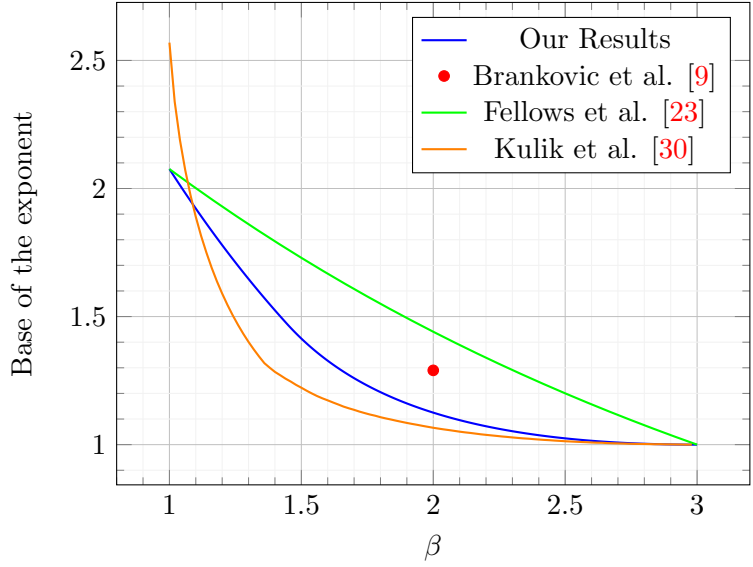}		
	\caption{Comparison of the running times of various algorithms for \hs{3}.
	The $x$-axis corresponds to the approximation ratio, while the $y$-axis corresponds to the base of the exponent in the running time.
	A point $(\beta, d)$ in the plot describes a running time of the form $d^{k} \cdot n^{\Oh(1)}$
	for a $\beta$-approximation. The red point corresponds to the 2-approximation algorithm from \cite{brankovicParameterizedApproximationAlgorithms2012}, with a
	running time of $1.29^{k} \cdot n^{\Oh(1)}$. Even though our result outperforms \cite{brankovicParameterizedApproximationAlgorithms2012} and \cite{Fellows2018}, it only improves upon \cite{KulikS2020} for values of $\beta$ such that $\beta \lessapprox 1.16$.}	
	\label{fig:3hsplot}
\end{figure}

\subsection{$\pathvc{4}$}
Recall the definition of \pathvc{\ell} from \cref{sec:pi_fixed_appl}. 
According to \cref{lemma:finite_forb_sampling}, there is
a sampling step for\\$\pathvc{4}$ with a success probability $\frac{1}{4}$.
Moreover, there exists an FPT algorithm that runs in time
$2.138^{k} \cdot n^{\Oh(1)}$ ($\alpha = 1, c = 2.138$) \cite{cervenyGeneratingFasterAlgorithms2023},
and a 3-approximation algorithm
that runs in polynomial time ($\alpha = 3, c = 1$) \cite{Camby2014}.

\begin{lemma}\label{lemma:4pvc_running_time}
	There is a randomized parameterized $\beta$-approximation algorithm
	for $\pathvc{4}$ with running time
	\begin{equation*}
		d^{k} \cdot n^{\Oh(1)}
	\end{equation*}
	where we have
	\begin{equation*}%
		d = \begin{cases}
			2.138 \cdot \left(0.621\right) ^{\beta - 1} &\text{if } 1 < \beta \leq 1.871\\
			\exp\left( \beta \cdot \D{\frac{1}{\beta}}{\frac{1}{4}} \right) &\text{if } 1.871 < \beta \leq 2.357\\
			\exp\biggl( \frac{ 2.357 \cdot \D{ 0.424 }{ 0.25 }  }{ 0.643 } \cdot (3 - \beta)  \biggr) &\text{if } 2.357 < \beta \leq 3 
		\end{cases}
	\end{equation*}	
\end{lemma}

\begin{proof}
	Let $\mathcal{A}_1$ denote the FPT algorithm from \cite{cervenyGeneratingFasterAlgorithms2023}, which runs in $2.138^{k} \cdot n^{\Oh(1)}$ time
	($\alpha = 1, c = 2.138$).
	Similarly, let $\mathcal{A}_2$ denote the 3-approximation algorithm from \cite{Camby2014} that runs in polynomial time
	(i.e., $\alpha = 3$, $c = 1$).
	Note that it suffices to consider $\beta \leq 3$ in the following, because for $\beta > 3$, $\mathcal{A}_2$ serves as a polynomial-time $\beta$ approximation algorithm.

	By using $\mathcal{A}_1$ and \cref{theorem:summary_sampling_step_alpha_1}, the first $\beta$-approximation algorithm we obtain
	has the running time $d^{k} \cdot n^{\Oh(1)}$ where
	\begin{equation}\label{eq:4pathvc_eq_1}
		d = \begin{cases}
			2.138 \cdot \left(0.621\right) ^{\beta - 1} &\text{if } 1 < \beta < 1.871\\
			\exp\left( \beta \cdot \D{\frac{1}{\beta}}{\frac{1}{4}} \right) &\text{if } 1.871 \leq \beta \leq 3.
		\end{cases}
	\end{equation}
	On the other hand, we have that $\sdeltal\left( 3,1,\frac{1}{4}
	\right) = 2.357$. Therefore, by using $\mathcal{A}_2$ and
	\cref{theorem:summary_sampling_step}, for every $1 < \beta \leq
	3$ there exists a parameterized $\beta$-approximation algorithm
	which runs in time $\runtime[3, \beta, 1, \frac{1}{4}]^{k}
	\cdot n^{\Oh(1)}$ where
	\begin{equation}\label{eq:4pathvc_eq_2}
		\runtime[3, \beta, 1, \frac{1}{4}] = \begin{cases}
			\exp\biggl( \frac{ 2.357 \cdot \D{ 0.424 }{ 0.25 }  }{ 0.643 } \cdot (3 - \beta)  \biggr) &\text{if } 2.357 < \beta \leq 3 \\
			\exp\left( \beta \cdot \D{\frac{1}{\beta}}{\frac{1}{4}} \right) &\text{if } 1 < \beta \leq 2.357 .
		\end{cases}
	\end{equation}
	The lemma follows by selecting the smaller value between
	\eqref{eq:4pathvc_eq_1} and \eqref{eq:4pathvc_eq_2} for each $1 < \beta
	< 3$.
\end{proof}

See \cref{fig:4pathvcplot} for a plot of $d$ in \cref{lemma:4pvc_running_time}, depending on
the approximation ratio $\beta$.

\begin{figure}[h!]
	\centering
	\input{figures/4pathvc.tex}		
	\caption{A plot of the running times of various algorithms for $\pathvc{4}$.
		The $x$-axis corresponds to the approximation ratio, while the $y$-axis corresponds to the base of the exponent in the running time.
	A point $(\beta, d)$ in the plot describes a running time of the form $d^{k} \cdot n^{\Oh(1)}$
	for a $\beta$-approximation. }	
	\label{fig:4pathvcplot}
\end{figure}

\subsection{\dfvst}
In this section, we demonstrate how our techniques extend to directed graph problems.
Recall that in the \dfvst problem, we are given a tournament graph $G$ and we would like to
find  a set of vertices $S$ such that $G \setminus S$ doesn't have any directed cycles.
Although our technique normally applies to hypergraphs, we adapt it
for directed graphs by defining $\mathcal{G}$ as the set of all tournament
graphs. 
Similarly, we define the graph property $\Pi$ to consist of all tournament graphs that are cycle free.
Note that a tournament is acyclic if and only if it contains no directed triangle.
It is not hard to show that our results for graph properties, with a finite set of forbidden graphs,
also apply in this setting. We omit the technical details.

By \cref{lemma:finite_forb_sampling}, there is a sampling step for \dfvst
with success probability $\frac{1}{3}$. There is also a 2-approximation algorithm
that runs in polynomial time \cite{lokshtanov2ApproximatingFeedbackVertex2021} ($\alpha = 2, c = 1$).
Moreover, there is an FPT algorithm with running time $1.618^{k} \cdot n^{\Oh(1)}$ ($\alpha = 1, c = 1.618$) \cite{kumarFasterExactParameterized2016}.

\begin{lemma}\label{lemma:dfvst_running_time}
	There is a randomized parameterized $\beta$-approximation algorithm
	for \dfvst with running time
	\begin{equation*}
		d^{k} \cdot n^{\Oh(1)}
	\end{equation*}
	where we have
	\begin{equation*}%
		d = \begin{cases}
			1.618 \cdot \left(0.691\right) ^{\beta - 1} &\text{if } 1 < \beta \leq 1.854\\			
			0.5^{\beta - 2} &\text{if } 1.854 < \beta \leq 2 
		\end{cases}
	\end{equation*}	
\end{lemma}

\begin{proof}
	Let $\mathcal{A}_1$ denote the FPT algorithm from \cite{kumarFasterExactParameterized2016}, which runs in time $1.618^{k} \cdot n^{\Oh(1)}$
	($\alpha = 1, c = 1.618$).
	Similarly, let $\mathcal{A}_2$ denote the $2$-approximation algorithm from \cite{lokshtanov2ApproximatingFeedbackVertex2021} that runs in polynomial time
	(i.e., $\alpha = 2$, $c = 1$).
Note that it suffices to consider $\beta \leq 2$ because of $\mathcal{A}_2$.	
	
	By using $\mathcal{A}_1$ and \cref{theorem:summary_sampling_step_alpha_1}, the first $\beta$-approximation algorithm we obtain
	has the running time $d^{k} \cdot n^{\Oh(1)}$ where
	\begin{equation}\label{eq:dfvst_running_time_1}
		d = \begin{cases}
			1.618 \cdot \left(0.691\right) ^{\beta - 1} &\text{if } 1 < \beta < 1.854\\
			\exp\left( \beta \cdot \D{\frac{1}{\beta}}{\frac{1}{3}} \right) &\text{if } 1.854 \leq \beta < 2.
		\end{cases}
	\end{equation}

	By using $\mathcal{A}_2$ and
	\cref{theorem:summary_sampling_step_alpha_2}, for every $1 < \beta \leq
	2$ there exists a parameterized $\beta$-approximation algorithm
	which runs in time $d^{k}
	\cdot n^{\Oh(1)}$ where
	\begin{equation}\label{eq:dfvst_running_time_2}
		d = \begin{cases}
			\exp\left( \beta \cdot \D{\frac{1}{\beta}}{\frac{1}{3}} \right) &\text{if } 1 < \beta \leq 1.5 \\
			0.5^{\beta - 2} &\text{if } 1.5 < \beta \leq 2 \\
		\end{cases}
	\end{equation}
	The lemma follows by selecting the smaller value between
	\eqref{eq:dfvst_running_time_1} and \eqref{eq:dfvst_running_time_2} for each $1 < \beta
	\leq 2$.
	
\end{proof}

See \cref{fig:dfvstplot} for a plot of $d$ in \cref{lemma:dfvst_running_time}, depending on
the approximation ratio $\beta$.

\begin{figure}[h!]
	\centering
	\input{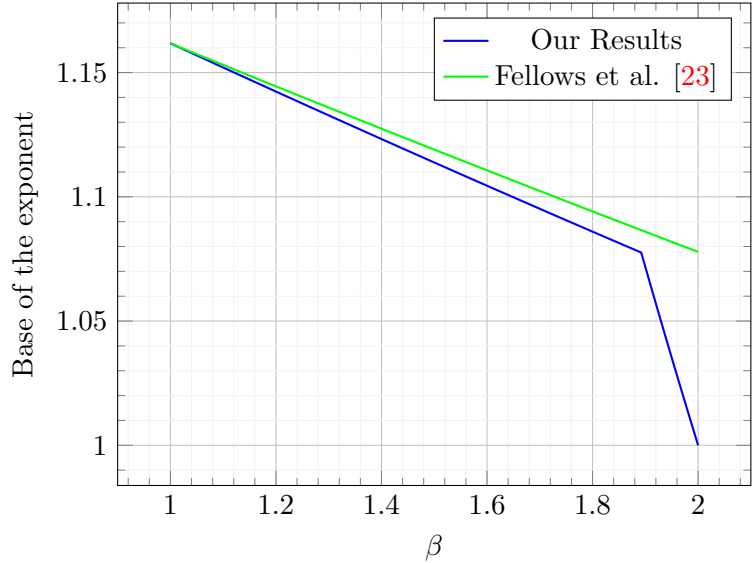}
	\caption{A plot of the running time of our algorithm for \dfvst.
		The $x$-axis corresponds to the approximation ratio, while the $y$-axis corresponds to the base of the exponent in the running time.
	A point $(\beta, d)$ in the plot describes a running time of the form $d^{k} \cdot n^{\Oh(1)}$
	for a $\beta$-approximation. }	
	\label{fig:dfvstplot}
\end{figure}

\subsection{\vc on Graphs with Maximal Degree $3$}
\vcmaxdegthree is the restriction of the \vc problem to graphs with maximum degree 3. 
It can be expressed as a $\gpivd$ problem, where $\mathcal{G}$ corresponds to graphs with maximum degree 3 and the hypergraph property $\Pi$ consists of all edgeless graphs.
Note that $\Pi$ can be described by the forbidden subgraph $K_2$, which is an
edge that consists of two vertices. Therefore, by \cref{lemma:finite_forb_sampling}
there exists a sampling step with success probability $\frac{1}{2}$.

In \cite{bermanApproximationPropertiesIndependent1999}, the authors present a polynomial time approximation algorithm for any approximation ratio arbitrarily close to 7/6.
For simplicity, we will assume that a $\frac{7}{6}$-approximation algorithm exists ($\alpha = \frac{7}{6}, c = 1$).
Note that when we consider $\beta$-approximation algorithms, we can focus on the values of
$\beta$ in the range $1 < \beta \leq \frac{7}{6}$ because of $\mathcal{A}_2$.
Moreover, there exists an FPT algorithm with running time $1.1616^{k} \cdot n^{\Oh(1)}$
\cite{xiaoNoteVertexCover2010} ($\alpha = 1, c = 1.1616$).

\begin{lemma}\label{lemma:vcmaxdegthree_running_time}
	There is a randomized parameterized $\beta$-approximation algorithm
	for\\ $\vcmaxdegthree$ with running time
	\begin{equation*}
		d^{k} \cdot n^{\Oh(1)}
	\end{equation*}
	where we have
	\begin{equation}\label{eq:4pthvc_runtime}
		d = \begin{cases}
				1.1616 \cdot \left(0.8384\right) ^{\beta - 1}  &\text{if } 1 < \beta \leq 1.136 \\ 			
			 \exp\biggl( \frac{ 1.008 \cdot \D{ 0.992 }{ 0.5 }  }{ 0.158 } \cdot (1.166 - \beta)  \biggr) &\text{if } 1.136 < \beta \leq 1.166			
				
		\end{cases}
	\end{equation}	
\end{lemma}

\begin{proof}
	Let $\mathcal{A}_1$ denote the FPT algorithm from \cite{xiaoNoteVertexCover2010}, which runs in $1.1616^{k} \cdot n^{\Oh(1)}$ time
	($\alpha = 1, c = 1.1616$).
	Similarly, let $\mathcal{A}_2$ denote the $\frac{7}{6}$-approximation algorithm from \cite{bermanApproximationPropertiesIndependent1999} that runs in polynomial time
	(i.e., $\alpha = \frac{7}{6}$, $c = 1$).
	Because of $\mathcal{A}_2$, we only consider $\beta \leq \frac{7}{6} \approx 1.166$.
	
	By using $\mathcal{A}_1$ and \cref{theorem:summary_sampling_step_alpha_1}, the first $\beta$-approximation algorithm we obtain
	has the running time $d^{k} \cdot n^{\Oh(1)}$ where
		\begin{equation}\label{eq:vc_deg_3_running_time_1}
			d = \begin{cases}
				1.1616 \cdot \left(0.8384\right) ^{\beta - 1} &\text{if } 1 < \beta < 1.722\\
				\exp\left( \beta \cdot \D{\frac{1}{\beta}}{\frac{1}{4}} \right) &\text{if } 1.722 \leq \beta < 2.
			\end{cases}
		\end{equation}
	On the other hand, we have that $\sdeltal\left( \frac{7}{6} ,1,\frac{1}{2}
	\right) = 1.008$. Therefore, by using $\mathcal{A}_2$ and
	\cref{theorem:summary_sampling_step}, for every $1 < \beta \leq
	2$ there exists a parameterized $\beta$-approximation algorithm
	which runs in time $\runtime[\frac{7}{6}, \beta, 1, \frac{1}{2}]^{k}
	\cdot n^{\Oh(1)}$ where
	\begin{equation*}
		\runtime[\frac{7}{6}, \beta, 1, \frac{1}{2}] = \begin{cases}
			\exp\left( \beta \cdot \D{\frac{1}{\beta}}{\frac{1}{2}} \right) &\text{if } 1 < \beta \leq 1.008 \\
			 \exp\biggl( \frac{ 1.008 \cdot \D{ 0.992 }{ 0.5 }  }{ 0.158 } \cdot (1.166 - \beta)  \biggr) &\text{if } 1.008 < \beta \leq 1.166			
		\end{cases}
	\end{equation*}
\end{proof}

See \cref{fig:vcmaxdegthreeplot} for a plot of $d$ in \cref{lemma:vcmaxdegthree_running_time}, depending on
the approximation ratio $\beta$.

\begin{figure}[h!]
	\centering
	\input{figures/vcmaxdeg3.tex}		
	\caption{A plot of the running time of our algorithm for $\vcmaxdegthree$.
	The $x$-axis corresponds to the approximation ratio, while the $y$-axis corresponds to the base of the exponent in the running time.
	A point $(\beta, d)$ in the plot describes a running time of the form $d^{k} \cdot n^{\Oh(1)}$
	for a $\beta$-approximation. }	
	\label{fig:vcmaxdegthreeplot}
\end{figure}

\section{Comparison to Fidelity Preserving Transformations}
\label{sec:fidelity_proof}
In this section we will prove \cref{lemma:comparison} which implies that Sampling with a Black Box provides better running time than Fidelity Preserving Transformations. 
Here we state the lemma once again for completeness.

\fidelitycomparison*

\begin{proof}
	The statement in \cref{lemma:comparison} is equivalent to
	\begin{equation}
		\ln\Biggl(\runtime[1,\beta,c,\frac{1}{\eta}]\Biggr) < \frac{\eta-\beta}{\eta - 1} \cdot \ln(c).\label{eq:comparison_equiv}
	\end{equation}
	Observe that for $\alpha = 1$, \eqref{eq:deltast_def} becomes
	\begin{equation*}
		-\ln\left( \frac{1}{\sdeltar\Bigl(1,c,\frac{1}{\eta}\Bigr)} \right) = -\ln\left( \frac{1}{\eta} \right)  - \ln(c),
	\end{equation*}
	which is equivalent to
	\begin{equation}\label{eq:sdeltar_simplified}
		\sdeltar\biggl(1,c,\frac{1}{\eta}\biggr) = \frac{\eta}{c}.
	\end{equation}
	Also recall that we assume
	\begin{equation}\label{eq:c_less_than_eta}
		c\leq  \exp\left(\alpha\cdot \D{\frac{1}{\alpha}}{\frac{1}{\eta}}\right) = -\ln\left( \frac{1}{\eta} \right) = \eta.
	\end{equation}
	In the following we will consider the two cases $\beta < \sdeltar$ and $\beta \geq \sdeltar$.
	We will demonstrate that in both cases \eqref{eq:comparison_equiv} holds.

	\begin{claim}\label{claim:beta_smaller_sdelta}
		Let $\beta < \sdeltar = \frac{\eta}{c}$. Then \eqref{eq:comparison_equiv} holds.
	\end{claim}

	\begin{claimproof}
		By substituting \eqref{eq:sdeltar_simplified} in the definition of $\runtime[1,\beta,c,\frac{1}{\eta}]$,
		we get
		\begin{equation}\label{eq:subst_sdelta}
			\ln\Biggl(\runtime[1,\beta,c,\frac{1}{\eta}]\Biggr) = \ln(c) + \frac{\beta - 1}{\frac{\eta}{c} - 1} \cdot \Biggl(\frac{\eta}{c} \cdot \D{\frac{c}{\eta}}{\frac{1}{\eta}} - \ln(c)\Biggr).
		\end{equation}
		Furthermore, by  the definition of the Kullback-Leibler divergence, we get that
		\begin{align}
			\frac{\eta}{c} \cdot \D{\frac{c}{\eta}}{\frac{1}{\eta}} &= \frac{\eta}{c} \cdot \frac{c}{\eta} \cdot \ln\left( \frac{c}{\eta} \cdot \eta \right)  + \frac{\eta}{c} \cdot \left( 1 - \frac{c}{\eta} \right) \cdot \ln\left( \frac{1 - \frac{c}{\eta}}{1 - \frac{1}{\eta}} \right) \notag\\
										&= \ln(c) + \frac{\eta - c}{c} \cdot \ln\left( \frac{\eta - c}{\eta - 1} \right)\label{eq:kl_subst}.
		\end{align}
		Recall that we have $\alpha  = 1 < \beta < \frac{\eta}{c}$, therefore it holds that $c < \eta$.
		Furthermore, by \eqref{eq:subst_sdelta} and \eqref{eq:kl_subst},
		\begin{align}
			\ln\Biggl(\runtime[1,\beta,c,\frac{1}{\eta}]\Biggr) &= \ln(c) + \frac{\beta - 1}{\frac{\eta}{c} - 1} \cdot \frac{\eta - c}{c} \cdot \ln\left( \frac{\eta - c}{\eta - 1} \right) \notag\\
										    &= \ln(c) + (\beta - 1) \cdot \ln\left( \frac{\eta - c}{\eta - 1} \right).\label{eq:ineq_part_1}
		\end{align}
		Next, observe that $\frac{c - 1}{\eta -c} > 0$ since $1 < c < \eta$.
		Therefore, by using the fact that $\ln(1 + x) \geq \frac{x}{1 + x}$ for $x > -1$,
		we obtain
		\begin{equation}\label{eq:ineq_part_2}
			\ln\left( \frac{\eta- 1}{\eta - c} \right) = \ln\left( 1 + \frac{c - 1}{\eta - c} \right) \geq \frac{\frac{c -1}{\eta - c} }{1 + \frac{c-1}{\eta -c }}= \frac{c-1}{\eta- 1} > \frac{\ln(c)}{\eta - 1}
		\end{equation}
		where the last inequality holds because $\ln(x) < x - 1$ for $x > 1$.
		Finally, by \eqref{eq:ineq_part_1} and \eqref{eq:ineq_part_2}, we get that
		\begin{align*}
			\ln\Biggl(\runtime[1,\beta,c,\frac{1}{\eta}]\Biggr) &= \ln(c) - (\beta - 1) \cdot \ln\left( \frac{\eta - 1}{\eta - c} \right) \\
										    &< \ln(c) - \frac{\beta - 1}{\eta - 1} \cdot \ln(c)\\
										    &=  \frac{\eta - \beta}{\eta - 1}  \cdot \ln(c)
		\end{align*}
		and \eqref{eq:comparison_equiv} holds.
	\end{claimproof}

	\begin{claim}\label{claim:beta_greater_sdelta}
		Let $\beta \geq \sdeltar = \frac{\eta}{c}$. Then \eqref{eq:comparison_equiv} holds.
	\end{claim}
	
	\begin{claimproof}
		By definition of $\runtime[1,\beta,c,\frac{1}{\eta}]$, we have
		\begin{align}
			\ln\biggl(\runtime[1,\beta,c,\frac{1}{\eta}]\biggr) = \beta \cdot \D{\frac{1}{\beta}}{\frac{1}{\eta}} &= \beta \cdot \frac{1}{\beta} \cdot \ln\left( \frac{\eta}{\beta} \right)  + \beta \cdot  \left( 1 - \frac{1}{\beta} \right) \cdot \ln\left(\frac{ 1 - \frac{1}{\beta} }{1 - \frac{1}{\eta}}\right)\notag\\
																      &= \ln\left( \frac{\eta}{\beta} \right) + (\beta - 1) \cdot \ln\left( \frac{\beta - 1}{\eta - 1} \cdot \frac{\eta}{\beta} \right).\label{eq:convert_beta_greater_delta}
		\end{align}

		Next, we will first demonstrate that
		\begin{equation}\label{eq:ln_beta_ineq}
			\ln\left( \frac{\beta-1}{\eta- 1} \right) < \frac{\eta}{\eta - 1}\cdot\ln\left( \frac{\beta}{\eta} \right).
		\end{equation}
		Define the function $\nu(x) \coloneqq \frac{\eta}{\eta - 1} \cdot \ln\left( \frac{x}{\eta} \right) - \ln\left( \frac{x-1}{\eta - 1} \right) $, and observe that $\nu(\eta) = 0$.
		Moreover, by standard calculations, we have that $\nu'(x) = \frac{x - \eta}{x\cdot (x - 1) \cdot (\eta - 1)}$.
		Note that $\nu'(x) < 0$ for $x < \eta$, i.e. the function $\nu(x)$ is decreasing for $x < \eta$.
		All in all, this implies that $\nu(\beta) > 0$ for all $\beta < \eta$, and therefore \eqref{eq:ln_beta_ineq} holds.		

		Moreover, we have
		\begin{align}
			(\beta - 1) \cdot \ln\left( \frac{\beta - 1}{\eta - 1} \cdot \frac{\eta}{\beta} \right) &= (\beta - 1) \cdot \ln\left( \frac{\beta - 1}{\eta - 1} \right) + (\beta - 1) \cdot \ln\left( \frac{\eta}{\beta} \right) \notag\\
														&\stackrel{\eqref{eq:ln_beta_ineq}}{<} (\beta - 1) \cdot \frac{\eta}{\eta - 1} \cdot \ln\left( \frac{\beta}{\eta} \right) + (\beta - 1) \cdot \ln\left( \frac{\eta}{\beta} \right)\notag\\
														&= \biggl( \frac{(1-\beta)\cdot \eta}{\eta - 1} + \beta - 1 \biggr)\cdot \ln\left( \frac{\eta}{\beta} \right)\notag\\
														&= \biggl(\frac{\eta - \beta \cdot \eta + \beta \cdot \eta - \beta - \eta + 1}{\eta - 1}\biggr) \cdot \ln\left( \frac{\eta}{\beta} \right) \notag\\
														&= \left( \frac{1 - \beta}{\eta - 1} \right) \cdot \ln\left( \frac{\eta}{\beta} \right)\label{eq:another_ineq_fidelity} 
		\end{align}

		By \eqref{eq:convert_beta_greater_delta} and \eqref{eq:another_ineq_fidelity}, it holds that
		\begin{align*}
			\ln\biggl(\runtime[1,\beta,c,\frac{1}{\eta}]\biggr) &= \ln\left( \frac{\eta}{\beta} \right) + (\beta - 1) \cdot \ln\left( \frac{\beta - 1}{\eta - 1} \cdot \frac{\eta}{\beta} \right)\\
										    &< \ln\left( \frac{\eta}{\beta} \right) + \frac{1- \beta}{\eta - 1} \cdot \ln\left( \frac{\eta}{\beta} \right) \\
										    &= \left( \frac{\eta - \beta}{\eta - 1} \right) \cdot \ln\left( \frac{\eta}{\beta} \right) \\
										    &\leq \left( \frac{\eta - \beta}{\eta - 1} \right) \cdot \ln(c)
		\end{align*}
		where the last step holds because $\frac{n}{\beta} \leq c$ by assumption.
		Therefore \eqref{eq:comparison_equiv} holds.
	\end{claimproof}

	By \cref{claim:beta_smaller_sdelta,claim:beta_greater_sdelta} we conclude that \eqref{eq:comparison_equiv} holds.
	Therefore \cref{lemma:comparison} holds as well.
\end{proof}

\section{Discussion}
\label{sec:discussion}

In this paper we presented Sampling with a Black Box, a simple and generic technique for the design of parameterized approximation algorithms for vertex deletion problems. The technique relies on sampling steps, polynomial time algorithms which return a random vertex whose removal reduces the optimum by one, with some success probability $q$. The technique combines the sampling step with existing parameterized and approximation algorithms to derive efficient parameterized approximation algorithms. We provide application for various problems, such as \fvs, \pathvc{\ell}, \hs{d} and \dfvst. 

We point out two directions for follow up works:
\begin{itemize}
	\item 
While Sampling with a Black Box provides faster parameterized approximation algorithms for multiple problems, it does  not provide a significant improvement for problems which has been extensively studied from this angle, such as \vc  and \hs{3} \cite{brankovicParameterizedApproximationAlgorithms2012, BF13,Fellows2018,KulikS2020}. This can be potentially improved by replacing sampling steps with a more generic procedure which can apply randomized branching rules \cite{KulikS2020}.  Initial research suggests this could result in better running times for several problems.

.

	\item 
	Our results are focused on {\em unweighted } vertex deletion problems. Vertex deletion problems can be naturally generalized for the weighted setting, in which each vertex $v\in V(G)$ has a weight $w(v)$ and the objective is to find a set $S\subseteq V(G)$ of minimum {\em weight} $\sum_{v\in S} w(v)$ such that $G\setminus S$ satisfies the property $\Pi$. 
	In particular, parameterized approximation algorithms for the special case of \wvc  has been recently considered in \cite{MMRS24}. 
	It would be interesting to adjust Sampling with a Black Box to the weighted setting. For example, such a result may  improve the running times of \cite{MMRS24} for \wvc. 
	
	In \cite{EKMNS23} the authors applied a rounding procedure over weights in order to utilize the (approximate) monotone local search technique of \cite{EKMNS24}  in a weighted setting. Intuitively, a similar approach may also be useful for Sampling with a Black Box.  
\end{itemize}

In this paper we designed {\em exponential time} parameterized approximation algorithms for vertex deletion problems.
For many of the considered problems, such as \vc and \hs{3}, it is known that, assuming the Exponential Time Hypothesis (ETH), there is no sub-exponential time parameterized (exact) algorithms (see, e.g. \cite{cyganParameterizedAlgorithms2015a}). However, it is less clear whether the considered problems admit sub-exponential time  parameterized approximation algorithms for approximation ratios close to $1$.  The existence of strictly sub-exponential parameterized  approximation algorithms for \vc for certain approximation ratio has been rules out, assuming ETH, in \cite{BEKP15}. However, we are not aware of a result which rules out a $c^{o(k)}\cdot n^{\Oh(1)}$ parameterized $(1+\eps)$-approximation for \vc (or other vertex deletion problem) for a some constant $\eps>0$. 
It would be interesting to explore whether recent tools in parameterized inapproximability, possibly together with  the stronger  Gap Exponential Time Hypothesis  (GAP-ETH) \cite{Din16,Manu17}, can lead to such a result.

\bibliographystyle{plain}
\bibliography{main}

\appendix

\section{Problem Definitions}
\label{sec:problem_defs}
In this section we will give formal definitions of the problems mentioned in this paper.
Recall that a feedback vertex set of a graph $G$ is a subset of its vertices $S \subseteq V(G)$
such that $G \setminus S$ is acyclic.
\defproblem{\fvs (\FVS)}{A graph $G$, an integer k.}{Does $G$ have a feedback vertex set of size at most $k$?}{$\mathcal{G}$ is the set of all graphs, $\Pi$ is the set of all graphs that have no cycles.}

\defproblem{\povd (\POVD)}{A graph $G$, an integer $k$.}{Is there a set of vertices $S \subseteq V(G)$ of size at most $k$ such that $G\setminus S$ has pathwidth at most 1?}{$\mathcal{G}$ is the set of all graphs and $\Pi$ is the set of all graphs with pathwidth at most 1 (i.e. the set of caterpillar graphs).}

\defproblem{\pathvc{\ell}}{A graph $G$, an integer $k$.}{Is there a set of vertices $S \subseteq V(G)$ of size at most $k$ such that every path of length $\ell$ contains a vertex from $S$?}{$\mathcal{G}$ is the set of all graph and $\Pi$ is the set of graphs with maximum path length $\ell - 1$.}

\defproblem{\hs{d}}{A universe $U$, a set system $\mathcal{S}$ over $U$ where each set $S \in \mathcal{S}$ has size at most $d$, an integer $k$.}{Is there a set $W \subseteq U$ of size at most $k$ such that $W \cap S \neq \emptyset$ for all $S \in \mathcal{S}$?}{$\mathcal{G}$ is the set of all hypergraphs with edge cardinality 3 and $\Pi$ is the set of all edgeless hypergraphs.}

\defproblem{\dfvst}{A tournament graph $G$, an integer $k$.}{Is there a set of vertices $S \subseteq G$ of size at most $k$ such that $G \setminus S$ has no directed cycles?}{$\mathcal{G}$ is the set of all tournament graphs and $\Pi$ is the set of all tournaments that are cycle free.}

\defproblem{\vcmaxdegthree}{A graph $G$ with maximum degree 3, an integer $k$}{Is there a set of vertices $S \subseteq V(G)$ of size at most $k$ such that such that every edge of $G$ has at least one vertex from $S$?}{$\mathcal{G}$ is the set of all graphs with maximum degree 3 and $\Pi$ is the set of all edgeless graphs.}

\section{Probabilistic Concepts}
\label{sec:prob_results}
Let $\D{x}{y}$ denote the Kullback-Leibler divergence between two Bernoulli distributions with parameters $x$ and $y$, i.e.
\begin{align}
	\D{x}{y} &\coloneqq x\cdot \ln\left( \frac{x}{y} \right)   + (1-x) \cdot \ln\left( \frac{1 - x}{1 - y} \right)\notag\\
		 &= x \cdot \ln\biggl(\frac{x}{1 - x} \cdot \frac{1 - y}{y}\biggr) + \ln\left( \frac{1 - x}{1 - y} \right).\label{eq:KL_div_equiv_form} 
\end{align}
In our analysis of procedures, we need the following technical result which is a special case of Theorem~11.1.4 in
\cite{ElementsInformationTheorya}.
\begin{theorem}\label{theorem:binom_lb}
	For any $0 \leq p \leq 1$ and integer $x \geq 1$, let $\prbinom(x,p)$ denote the binomial random variable with success probability $p$ and number of trials $x$. For any $0 \leq y \leq x$, it holds that
	\begin{align*}
		\Pr\left( \prbinom(x,p) \geq y \right) \geq (x + 1)^{-2} \cdot \exp\left( -x \cdot \D{\frac{y}{x}}{p} \right).
	\end{align*}
\end{theorem}

Moreover, we also need the following lemma, which we use to lower bound the tail probability of the sum of certain random variables.

\begin{restatable}{lemma}{problowerbound}
\label{lemma:prob_lower_bound}
	\label{lem:prob_lem}
	Let $X_1,\ldots, X_n\in \{0,1\}$ be random variables, let $p\in [0,1]$
	and assume that 
	\begin{equation*}
		\Pr\left( X_j =1\,\middle|\, X_1 = x_1,\ldots,X_{j-1}
	= x_{j - 1}\right)\geq p
	\end{equation*}
	for all $j\in [n]$ and $\left( x_1, \ldots,
x_{j-1} \right)  \in \{0,1\}^{j-1}$. Then, for every $w\in \mathbb{R}$
it holds that
\begin{equation*}
	\Pr\left( \sum_{j=1}^n X_j \,\geq\, w\right) \geq
\Pr(\prbinom(n,p)\,\geq\, w).
\end{equation*}
\end{restatable}

\begin{proof}
Let $Y_1,\ldots, Y_n$ be $n$ independent Bernoulli  random variables with $\Pr(Y_j=1) = p$ for every $j\in[n]$. Also, define $Q_{\ell}  = \sum_{j=\ell}^n Y_{\ell}$ and $S_{\ell } = \sum_{j=\ell}^{n} X_{\ell}$ for every $\ell \in [n+1]$. By definition, $Q_{n+1}=S_{n+1}=0$.   Furthermore, the distribution of $Q_{1}$ is $\prbinom(n,p)$. 

\begin{claim}
	\label{claim:dominance}
For every $\ell \in [n+1]$ and $w\in \mathbb{R}$ it holds that $\Pr(S_{\ell}\geq w\,|\, \cF_{\ell-1}) \geq \Pr(Q_{\ell}\geq w)$. 
\end{claim}
\begin{claimproof}
We prove the claim by reverse induction over the value of $\ell$. 

\noindent {\bf Base case:} Let $\ell =n+1$. Then $S_{\ell}=0=Q_{\ell}$. Therefore $\Pr(S_{\ell}\geq w\,|\, \cF_{\ell-1}) = \Pr(Q_{\ell}\geq w)$.

\noindent {\bf Induction step:} assume the induction hypothesis holds for $\ell+1\in [n+1]\setminus \{1\}$. Let $w\in \mathbb{R}$. Then,
\begin{equation}
	\label{eq: dominance_eq1}
\begin{aligned}
	\Pr(S_{\ell}\geq w\,|\,\cF_{\ell-1} ) \,&=\,\E \left[ \one_{X_{\ell}=1}\cdot \one_{S_{\ell+1} \geq w-1}+\one_{X_{\ell}=0}\cdot \one_{S_{\ell+1} \geq w}  \,\middle|\, \cF_{\ell-1} \right]\\
	&=\, \E \left[ \one_{X_{\ell}=1}\cdot \E\left[ \one_{S_{\ell+1} \geq w-1}\, \middle|\,\cF_{\ell} \right]+\one_{X_{\ell}=0}\cdot \left[ \one_{S_{\ell+1} \geq w}\,\middle|\,\cF_{\ell}\right]  \,\middle|\, \cF_{\ell-1} \right]\\
	&\geq \, \E \left[ \one_{X_{\ell}=1}\cdot \Pr(Q_{\ell+1} \geq w-1 )+\one_{X_{\ell}=0}\cdot \Pr(Q_{\ell+1}\geq w) \,\middle|\, \cF_{\ell-1} \right]\\
	&= \, \Pr(X_{\ell}=1\,|\,\cF_{\ell-1}) \cdot \Pr(Q_{\ell+1} \geq w-1 ) + \Pr(X_{\ell}=0\,|\,\cF_{\ell-1})\cdot \Pr(Q_{\ell+1} \geq w),
\end{aligned}
\end{equation}
where the second equality follows from the tower property and the inequality holds by the induction hypothesis.  By \eqref{eq: dominance_eq1} we have,
	\begin{equation*}
		\begin{aligned}
			\Pr(S_{\ell}\geq w\,|\,\cF_{\ell-1} ) \,
			&= \, \Pr(X_{\ell}=1\,|\,\cF_{\ell-1}) \cdot \Pr(Q_{\ell+1} \geq w-1 ) + \Pr(X_{\ell}=0\,|\,\cF_{\ell-1})\cdot \Pr(Q_{\ell+1} \geq w)\\
			&\geq\, p\cdot \Pr(Q_{\ell+1} \geq w-1 ) + (1-p)\cdot \Pr(Q_{\ell+1} \geq w) \\
			&=\, \Pr(Y_{\ell}=1) \cdot \Pr(Q_{\ell+1} \geq w-1 \,|\, Y_{\ell}=1) + \Pr(Y_{\ell}=0) \cdot \Pr(Q_{\ell+1} \geq w \,|\, Y_{\ell}=0)\\
			&=\, \Pr(Q_{\ell} \geq w).  
		\end{aligned}
	\end{equation*}
The first inequality holds as $\Pr(X_{\ell}=1\,|\,\cF_{\ell-1}) \geq p$ and $ \Pr(Q_{\ell+1} \geq w-1 ) \geq \Pr(Q_{\ell+1} \geq w)$. The second equality holds a $\Pr(Y_{\ell}=1)=p$ and since $Y_{\ell}$ and $Q_{\ell+1}$ are independent.  Thus, we proved the induction hypothesis holds for $\ell$ and completed the proof. 
\end{claimproof}
By \Cref{claim:dominance}, for every $w\in \mathbb{R}$ it holds that 
$$\Pr\left(\sum_{j=1}^n X_n \geq w \right)  \, =\, \Pr\left(S_1 \geq w\right)\, \geq \,  \Pr\left(Q_1 \geq w\right) \, =\, \Pr(\prbinom(n,p)\,\geq\, w). $$
 \end{proof}

We can combine \cref{theorem:binom_lb,lemma:prob_lower_bound} to obtain the following result.
\probmainresult*

\begin{proof}
		If $\delta = 1$, observe that
		\begin{align}
			\Pr\Biggl( \sum_{j = 1}^{\floor{\delta \cdot t}} \xi_j \geq t\Biggr) = \Pr\Biggl( \sum_{j = 1}^{t} \xi_j \geq t\Biggr) &\geq \Pr\Bigl( \prbinom\left( t, \nu \right) \geq  t \Bigr)\tag{by \cref{lem:prob_lem}}\\
											     &= \Pr\Bigl( \prbinom\left( t, \nu \right) =  t \Bigr)\notag \\
											     &= \nu^{t}\notag \\
											     &= \exp\Bigl(\ln\left( \nu \right) \Bigr)^{t}\notag\\
											     &= \exp\Biggl(-\delta \cdot \D{\frac{1}{\delta}}{\nu}\Biggr)^{t}\notag
		\end{align}
		where the last step holds because $\D{1}{\nu} = \ln\left( \frac{1}{\nu} \right)$.
		Now suppose that $\delta > 1$ and let $T_\delta= \ceil{\frac{1}{\delta-1}}$. For $t \geq T_\delta$ we have
		\begin{align}
			\Pr\biggl( \sum_{j = 1}^{\floor{\delta \cdot t}} \xi_j \geq t\biggr) &\geq \Pr\Bigl( \prbinom\left( \floor{\delta\cdot t}, \nu \right) \geq  t \Bigr)\tag{by \cref{lem:prob_lem}}\\			
					  &\geq \Pr\Bigl( \prbinom\left( \floor{\delta\cdot t}, \nu \right) =  t \Bigr)\notag\\
					  &\geq \Bigl( \floor{\delta\cdot t} + 1 \Bigr) ^{-2} \cdot \exp\Biggl( - \floor{\delta \cdot t} \cdot \D{\frac{t}{\floor{\delta \cdot t}}}{\nu}\Biggr)\label{eq:X_j_prob_last_step},
	  	\end{align}
		where the last step follows from \cref{theorem:binom_lb}.

		\begin{claim}\label{claim:technical_bound}
			There exists a constant $r > 0$ that depends on $\delta$ such that for each $t \geq T_{\delta}$ it holds that
			\begin{equation*}
				\Bigl( \floor{\delta\cdot t} + 1 \Bigr) ^{-2} \cdot \exp\Biggl( - \floor{\delta \cdot t} \cdot \D{\frac{t}{\floor{\delta \cdot t}}}{\nu}\Biggr) \geq \Bigl( \delta\cdot t + 1 \Bigr) ^{-r} \cdot \exp\Biggl( - \delta \cdot t \cdot \D{\frac{1}{\delta}}{\nu}\Biggr).
			\end{equation*}			
		\end{claim}

		\begin{claimproof}
			Let us define $h(x) \coloneqq \frac{1}{x} \cdot \D{x}{\nu}$. Then we have
			\begin{equation*}
			\Bigl( \floor{\delta\cdot t} + 1 \Bigr)^{-2} \cdot \exp\Biggl( - \floor{\delta \cdot t} \cdot \D{\frac{t}{\floor{\delta \cdot t}}}{\nu}\Biggr) = \Bigl( \floor{\delta\cdot t} + 1 \Bigr)^{-2} \cdot \exp \Biggl( -t \cdot h\left( \frac{t}{\floor{\delta \cdot t}} \right)  \Biggr).
			\end{equation*}
	
			Since $h$ is a differentiable function on $\left( \frac{1}{\delta}, \frac{t}{\floor{\delta \cdot t}} \right)$, we can approximate the value of $h\left( \frac{t}{\floor{\delta \cdot t}} \right)$ by
			\begin{equation}\label{eq:h_taylor}
				h\left( \frac{t}{\floor{\delta \cdot t}} \right) = h\left( \frac{1}{\delta} \right) + \left( \frac{t}{\floor{\delta \cdot t}} - \frac{1}{\delta} \right) \cdot h'(w)
			\end{equation}
			for some $w \in \left( \frac{1}{\delta}, \frac{t}{\floor{\delta \cdot t}} \right) $, using Mean Value Theorem. 
			Note that
			\begin{align}
				\frac{t}{\floor{\delta \cdot t}} - \frac{1}{\delta} &= \frac{\delta \cdot t - \floor{\delta \cdot t}}{\delta \cdot \floor{\delta \cdot t}} \notag \\
										    &\leq \frac{1}{\delta \cdot \floor{\delta \cdot t}} \notag \\
										    &< \frac{1}{\delta \cdot \left( \delta \cdot t - 1 \right) } \notag \\
										    &< \frac{1}{\delta \cdot t} \label{eq:second_to_last}\\
										    &< \frac{\delta - 1}{\delta} \label{eq:last},
			\end{align}
			where \eqref{eq:second_to_last} and hold because $t \geq T_\delta \geq \frac{1}{\delta - 1}$, which implies that $\delta \cdot t - t \geq 1 \iff \delta \cdot t - 1 \geq t$.
			By \eqref{eq:last}, the value of $h'(w)$ is upper bounded by $C_\delta \coloneqq \max_{w'\in \left(\frac{1}{\delta}, \frac{\delta-1}{\delta} \right)}  h'(w')$. Note that $C_\delta$ only depends on $\delta$ hence it is a constant.

			Finally, by \eqref{eq:h_taylor} and \eqref{eq:second_to_last}, it holds that
			\begin{align}
				t \cdot h\left( \frac{t}{\floor{\delta \cdot t}} \right) &= t \cdot h\left( \frac{1}{\delta} \right) + t \cdot \left( \frac{t}{\floor{\delta \cdot t}} - \frac{1}{\delta} \right) \cdot h'(w)\notag\\
											 &< t \cdot h\left( \frac{1}{\delta} \right) + \frac{1}{\delta} \cdot C_\delta.\label{eq:t_times_h_ub}
			\end{align}
			Therefore, 
			\begin{align*}
				\exp\biggl( - \floor{\delta \cdot t} \cdot \D{\frac{t}{\floor{\delta \cdot t}}}{\nu}\biggr) &= \exp \biggl( -t \cdot h\left( \frac{t}{\floor{\delta \cdot t}} \right)  \biggr)\\			
																							    &\overset{\eqref{eq:t_times_h_ub}}{\geq}  \exp \Biggl(-t \cdot h\left( \frac{1}{\delta} \right) \Biggr) \cdot \exp\left( -\frac{1}{\delta} \cdot C_\delta \right) \\
				&\geq  \exp \Biggl(-\delta \cdot t \cdot \D{\frac{1}{\delta}}{\nu}\Biggr) \cdot \Bigl( \delta \cdot T_\delta + 1 \Bigr)^{-r+2}\\
			\end{align*}
			
			where $r > 2$ is a large enough constant that depends on $\delta$ such that $\exp\Bigl( \frac{1}{\delta} \cdot C_\delta \Bigr) \leq \Bigl( \delta \cdot T_\delta + 1 \Bigr)^{r - 2}$.
		\end{claimproof}

		Finally, by \cref{claim:technical_bound} and \eqref{eq:X_j_prob_last_step}, there exist $r,T_\delta > 0$ that depend on $\delta$ such that for $t \geq T_\delta$ we have
		\begin{equation*}
			\Pr\biggl( \sum_{j = 1}^{\floor{\delta \cdot t}} \xi_j \geq t\biggr) \geq \left( \delta\cdot t + 1 \right) ^{-r} \cdot \exp\biggl( - \delta \cdot t \cdot \D{\frac{1}{\delta}}{\nu}\biggr).
		\end{equation*}
\end{proof}

\section{Technical Claims}
\label{sec:tech_claims}

\begin{lemma}
\label{lemma:z_s_deriv_formula}
	It holds that
	\begin{equation}\label{eq:z_deriv}
		\frac{\partial}{\partial\,\delta}\,\Biggl(\ln\left( \frac{1}{\phi(\delta,q)} \right)\Biggr) = \ln\left( \frac{1 - \frac{1}{\delta}}{1 - q} \right)
	\end{equation}
	and
	\begin{align}
		\frac{\partial}{\partial\,\delta}\,s_{q}(\delta)     &= \frac{(\alpha - 1) \cdot \ln\left( \frac{1 -q}{1 - \frac{1}{\delta}} \right) + \ln\left( \delta \cdot q \right) + \ln(c)}{(\delta - \alpha)^{2}}.\label{eq:s_deriv}	
	\end{align}	
\end{lemma}

\begin{proof}
	It holds that

	\begin{align}\label{eq:KL_deriv}
		\frac{\partial}{\partial\,a} \,\D{a}{b}= \ln\! \left(\frac{a}{1 - a} \cdot \frac{1 - b}{b}\right).
	\end{align}

	Therefore, using the product rule for the derivative, we get
	\begin{align*}
		\frac{\partial}{\partial\,\delta}\, \Biggl(\ln\left( \frac{1}{\phi(\delta,q)} \right)\Biggr) &= \frac{\partial}{\partial\,\delta}\, \left( \delta \cdot \D{\frac{1}{\delta}}{q} \right) \\
						&= \D{\frac{1}{\delta}}{q} + \delta \cdot \ln\left( \frac{\frac{1}{\delta}}{1 - \frac{1}{\delta}} \cdot \frac{1-q}{q} \right) \cdot \left( -\frac{1}{\delta^{2}} \right) &\text{by } \eqref{eq:KL_deriv} \\
						&= \frac{1}{\delta} \cdot \ln\left( \frac{\frac{1}{\delta}}{q} \cdot \frac{1 - q}{1 - \frac{1}{\delta}}\right) + \ln\left( \frac{1 - \frac{1}{\delta}}{1 - q} \right) -\frac{1}{\delta}\cdot \ln\left( \frac{\frac{1}{\delta}}{1 - \frac{1}{\delta}} \cdot \frac{1 - q}{q} \right)\\
						&= \ln\left( \frac{1 - \frac{1}{\delta}}{1 - q} \right),
	\end{align*}
	therefore \eqref{eq:z_deriv} holds.
	Similarly, by using the quotient rule for the derivative, we get
	\begin{align*}
		\frac{\partial}{\partial\,\delta}\,s_{q}(\delta) &= \frac{\partial}{\partial\,\delta}\,\Biggl(\frac{\ln\left( \frac{1}{\phi(\delta,q)} \right) - \ln(c)}{\delta - \alpha}\Biggr)\\
		&= \frac{\Bigl( \frac{\partial}{\partial\,\delta} \,\ln\left( \frac{1}{\phi(\delta,q)} \right)\Bigr) \cdot (\delta - \alpha) - \Bigl( \ln\left( \frac{1}{\phi(\delta,q)} \right) - \ln(c) \Bigr)}{(\delta - \alpha)^{2}}\\
					       &= \frac{(\delta - \alpha) \cdot \ln\left( \frac{1 - \frac{1}{\delta}}{1 - q} \right) - \ln\left( \frac{1}{\phi(\delta,q)} \right) + \ln(c)}{(\delta-  \alpha)^{2}}. &\text{by } \eqref{eq:z_deriv}
	\end{align*}

	Using the definition of $\phi(\delta,q)$, we further have
	\begin{align*}
		\frac{\partial}{\partial\,\delta}\,s_{q}(\delta) &= \frac{(\delta - \alpha) \cdot \ln\left( \frac{1 - \frac{1}{\delta}}{1 - q} \right) - \delta \cdot \D{\frac{1}{\delta}}{q} + \ln(c)}{(\delta-  \alpha)^{2}}\\
							       &= \frac{(\delta - \alpha) \cdot \ln\left( \frac{1 - \frac{1}{\delta}}{1 - q} \right) - \delta \cdot \left( \frac{1}{\delta} \cdot \ln\left( \frac{1}{\delta \cdot q} \right) + \left( 1 - \frac{1}{\delta} \right)\cdot \ln\left( \frac{1 - \frac{1}{\delta}}{1 - q} \right)   \right)  + \ln(c)}{(\delta-  \alpha)^{2}}\\
							       &= \frac{(\delta - \alpha - \delta + 1) \cdot \ln\left(\frac{ 1 - \frac{1}{\delta} }{1 -q}\right) - \ln\left( \frac{1}{\delta \cdot q} \right)  + \ln(c)}{(\delta - \alpha)^{2}}\\
							       &= \frac{(\alpha - 1) \cdot \ln\left( \frac{1 -q}{1 - \frac{1}{\delta}} \right) + \ln\left( \delta \cdot q \right) + \ln(c)}{(\delta - \alpha)^{2}}.
	\end{align*}
\end{proof}

Next, we state an equivalence which will be used frequently in the following section.

\begin{lemma}\label{lemma:exp_s_substitute}
	It holds that
	\begin{equation*}
		\exp\Bigl( \ln(c) + (\beta - \alpha) \cdot s_{q}(\beta) \Bigr) = \exp\left( \beta \cdot \D{\frac{1}{\beta}}{q} \right).
	\end{equation*}
\end{lemma}

\begin{proof}
	The proof simply follows by substituting:
	\begin{align*}
		\exp\Bigl( \ln(c) + (\beta - \alpha) \cdot s_{q}(\beta) \Bigr) &= \exp\Biggl( \ln(c) + (\beta - \alpha) \cdot \frac{\ln\left( \frac{1}{\phi(\delta,q)} \right) -\ln(c)}{(\beta - \alpha)} \Biggr)\\
											&= \exp\Biggl(\ln\left( \frac{1}{\phi(\delta,q)} \right)\Biggr)\\
											&= \exp\Biggl( \beta \cdot \D{\frac{1}{\beta}}{q} \Biggr).
	\end{align*}
\end{proof}

\section{Omitted Proofs}
\label{sec:omitted_proofs}

\heredoptdecrease*
\begin{proof}
	Let $A \in \sat_{\Pi}(G)$ such that $\abs{A} = \OPT_{\Pi}(G)$.
	It holds that
	\begin{align*}
		\Bigl(G \setminus v\Bigr) \setminus \Bigl(A \setminus \{v\} \Bigr) &= G \setminus \Bigl(A \cup \{v\} \Bigr) \in \Pi
	\end{align*}
	because $\Pi$ is hereditary and $G \setminus \Bigl(A \cup \{v\} \Bigr)$ is a vertex induced subhypergraph of $G \setminus A$, which belongs to $\Pi$ by the definition of $A$.
	Therefore, $A \setminus \{v\}$ is a solution for $G \setminus v$ and
	\begin{equation}\label{eq:hereditary_opt_decrease_eq_1}
		\OPT_{\Pi}(G \setminus v ) \leq \,\abs{A \setminus \{v\} }\, \leq \OPT_{\Pi}(G).
	\end{equation}

	Similarly, let $X \in \sat_{\Pi}(G \setminus v)$ such that $\abs{X} = \OPT_{\Pi}( G \setminus v )$.
	We have
	\begin{align*}
		G \setminus \left( X \cup \{v\}  \right) = \Bigl(\left( G \setminus v \right) \setminus X \Bigr)\in \Pi 
	\end{align*}
	by definition of $X$. Therefore, $( X \cup \{v\}) \in \sat_{\Pi}\left( G \right) $ and we have
	\begin{equation}\label{eq:hereditary_opt_decrease_eq_2}
		\OPT_{\Pi}\left( G  \right) \leq \abs{X \cup \{v\} } \leq \abs{X} + 1 \leq \OPT_{\Pi}( G \setminus v ) + 1.
	\end{equation}

	Finally, \eqref{eq:hereditary_opt_decrease_eq_1} and \eqref{eq:hereditary_opt_decrease_eq_2} together imply the lemma.
\end{proof}

Now, we provide the previously omitted proof of \cref{lemma:s_q_unimodal}. For the sake of completeness, we restate the lemma below.

\squnimodal*

\begin{proof}
	By \cref{lemma:s_deriv_alternative}, the sign of the derivative of
	$s_q(\delta)$, i.e. $\sign\biggl(\frac{\partial}{\partial\,\delta}\,s_{q}(\delta)\biggr)$,
	agrees with the sign of $\Gamma_q(\delta)$. Observe that the only term
	in $\Gamma_q(\delta)$ that depends on $\delta$ is $\alpha \cdot \D{\frac{1}{\alpha}}{\frac{1}{\delta}}$.
	
	Consider the values of $\delta$ such that $\delta \geq \alpha$, which implies that
	$\frac{1}{\delta} \leq \frac{1}{\alpha}$. Since $\D{\frac{1}{\alpha}}{x}$ is a strictly
	decreasing function for $x \leq \frac{1}{\alpha}$, and $\frac{1}{\delta}$ is a
	strictly decreasing function of $\delta$, it follows that $\D{\frac{1}{\alpha}}{\frac{1}{\delta}}$
	is a strictly increasing function of $\delta$ for $\delta \geq \alpha$.
	Furthermore, observe that $\Gamma_q(\sdeltar) = 0$. Therefore,
	\begin{equation*}
		\sign\biggl(\frac{\partial}{\partial\,\delta}\,s_{q}(\delta)\biggr) = \sign\left( \Gamma_q(\delta) \right) < 0
	\end{equation*}
	for $\alpha \leq \delta < \sdeltar$. Similarly,
	\begin{equation*}
		\sign\biggl(\frac{\partial}{\partial\,\delta}\,s_{q}(\delta)\biggr) = \sign\left( \Gamma_q(\delta) \right) > 0		
	\end{equation*}
	for $\sdeltar < \delta \leq \frac{1}{q}$. Therefore, $s_q(\delta)$ is strictly decreasing for $\alpha \leq \delta \leq \sdeltar$
	and strictly increasing for $\sdeltar \leq \delta \leq \frac{1}{q}$.

	Now, assume that $\alpha > 1$, and consider the values of $\delta$ such that $\delta \leq \alpha$,
	which implies that $\frac{1}{\delta} \geq \frac{1}{\alpha}$. Since $\D{\frac{1}{\alpha}}{x}$
	is a strictly increasing function of $x$ for $x \geq \frac{1}{\alpha}$, it holds that
	$\D{\frac{1}{\alpha}}{\frac{1}{\delta}}$ is a strictly decreasing function
	of $\delta$ for $\delta \geq \alpha$. We also have that $\Gamma_q(\sdeltar) = 0$.
	Therefore, it holds that
	\begin{equation*}
		\sign\biggl(\frac{\partial}{\partial\,\delta}\,s_{q}(\delta)\biggr) = \sign\left( \Gamma_q(\delta) \right) >0
	\end{equation*}
	for $1 \leq \delta < \sdeltal$ and
	\begin{equation*}
		\sign\biggl(\frac{\partial}{\partial\,\delta}\,s_{q}(\delta)\biggr) = \sign\left( \Gamma_q(\delta) \right) < 0		
	\end{equation*}
	for $\sdeltal < \delta \leq \alpha$. Therefore, $s_q(\delta)$ is a strictly increasing
	function for $1 \leq \delta \leq \sdeltal$ and a strictly decreasing function for
	$\sdeltal \leq \delta \leq \alpha$.

	Therefore the lemma holds.
\end{proof}

Next we present the missing proof of \cref{lemma:max_s_beta_<_alpha}. For completeness, we again state the lemma
here.

\maxsbetasmalleralpha*
\begin{proof}
	By the conditions of the lemma it holds that $\sdeltal \in \goodd \cap [1,\frac{1}{q}]$.
	 By  \Cref{lemma:s_q_unimodal} the function $s_q(\delta)$ is increasing in $[1,\sdeltal]$ and decreasing in $[\sdeltal, \alpha]$.  Therefore,
	 $$
	 \max_{\delta \in \goodd \cap [1,\frac{1}{q}]} s_q(\delta) = s_q(\sdeltal).	 $$
\end{proof}

Let us now give the omitted proof of \cref{lemma:s_deriv_alternative}.

\sderivalternative*
\begin{proof}
	Since $(\delta - \alpha)^{2} > 0$, by \eqref{eq:s_deriv} it holds that $\frac{\partial}{\partial\,\delta}\,s_{\alpha,c,q}(\delta) = 0$ if and only if
	\begin{equation}\label{eq:s_deriv_iff}
		(\alpha - 1) \cdot \ln\left( \frac{1 -q}{1 - \frac{1}{\delta}} \right) + \ln\left( \delta \cdot q \right) + \ln(c) = 0.
	\end{equation}
	Moreover,
	\begin{equation}\label{eq:s_deriv_sign_inter}
		\sign\biggl(\frac{\partial}{\partial\,\delta}\,s_{\alpha,c,q}(\delta)\biggr) = \sign\biggl((\alpha - 1) \cdot \ln\left( \frac{1 -q}{1 - \frac{1}{\delta}} \right) + \ln\left( \delta \cdot q \right) + \ln(c)\biggr).
	\end{equation}
	where we let $\Psi$ denote $\Psi \coloneqq (\alpha - 1) \cdot \ln\left( \frac{1 -q}{1 - \frac{1}{\delta}} \right) + \ln\left( \delta \cdot q \right) + \ln(c)$ for
	for the sake of presenting the following material.
	We have
	\begin{align*}
		\Psi &= (\alpha - 1) \cdot \ln\left( \frac{1 -q}{1 - \frac{1}{\delta}} \right) + \ln\left( \delta \cdot q \right) + \ln(c)\\
		     &= -(\alpha - 1) \cdot \ln\left( \frac{1 - \frac{1}{\delta}}{1 - q} \right) -\ln\left(\frac{\frac{1}{\delta}}{q}\right) + \ln(c) \\
														  &= -\alpha \cdot \Biggl(\left( 1 - \frac{1}{\alpha} \right)\cdot \ln\biggl( \frac{ 1 - \frac{1}{\delta} }{1 - q}  \biggr) + \frac{1}{\alpha} \cdot \ln\left(\frac{\frac{1}{\delta}}{q}\right)  \Biggr) + \ln(c)\\
														  &= -\alpha \cdot \Biggl( \frac{1}{\alpha} \cdot \ln\left( \frac{\frac{1}{\delta}}{q} \cdot \frac{1 - q}{1 - \frac{1}{\delta}} \right) + \ln\left( \frac{1 - \frac{1}{\delta}}{1 - q} \right)  \Biggr) + \ln(c).
	\end{align*}
	
	Next, by dividing and multiplying the term inside the logarithm by the same value, we get

	\begin{align}
	       \Psi &= -\alpha \cdot \Biggl[\frac{1}{\alpha} \cdot \ln\! \left(\frac{\frac{1}{\alpha}}{1 - \frac{1}{\alpha}} \cdot \frac{1 - q}{q} \cdot \frac{1 - \frac{1}{\alpha}}{\frac{1}{\alpha}}\cdot \frac{\frac{1}{\delta}}{1 - \frac{1}{\delta}}\right) + \ln\left(\frac{1 - \frac{1}{\alpha}}{1 - q} \cdot \frac{1 - \frac{1}{\delta}}{1 - \frac{1}{\alpha}} \right)\Biggr] + \ln(c)\notag\\
	       &= -\alpha \cdot \Biggl[\frac{1}{\alpha} \cdot \ln\! \left(\frac{\frac{1}{\alpha}}{1 - \frac{1}{\alpha}} \cdot \frac{1 - q}{q} \right) + \ln\left(\frac{1 - \frac{1}{\alpha}}{1 - q} \right)\Biggr] - \alpha \cdot \Biggl[ \frac{1}{\alpha} \cdot \ln\! \left(\frac{1 - \frac{1}{\alpha}}{\frac{1}{\alpha}}\cdot \frac{\frac{1}{\delta}}{1 - \frac{1}{\delta}}\right) + \ln\left( \frac{1 - \frac{1}{\delta}}{1 - \frac{1}{\alpha}} \right)\Biggr] + \ln(c)\notag\\
	       &= -\alpha \cdot \D{\frac{1}{\alpha}}{q} +  \alpha \cdot \D{\frac{1}{\alpha}}{\frac{1}{\delta}} + \ln(c)\label{eq:s_deriv_numer_alter}			
	\end{align}
	where the last step follows from \eqref{eq:KL_div_equiv_form}.
	Finally, the lemma holds by \eqref{eq:s_deriv_iff}, \eqref{eq:s_deriv_sign_inter} and \eqref{eq:s_deriv_numer_alter}.
\end{proof}

\end{document}